\documentclass[pra, onecolumn, superscriptaddress, nofootinbib]{revtex4-2}
\usepackage[a4paper, left=1in, right=1in, top=1in, bottom=1in]{geometry}

\usepackage{cmap} 
\usepackage[utf8]{inputenc}
\usepackage[english]{babel}
\usepackage[T1]{fontenc}
\usepackage{appendix}
\usepackage{braket}
\usepackage{url, amsfonts, tikz, physics, listings, xcolor, graphicx, float}
\usepackage[shortlabels]{enumitem}
\usepackage{dsfont}
\usepackage{amsmath, amssymb, amstext, amscd, amsthm, makeidx, graphicx, url, mathrsfs, mathtools, longdivision, polynom, bbm, complexity}
\usepackage{nicefrac}
\usepackage[linesnumbered,ruled,vlined]{algorithm2e}
\usepackage{xcolor}
\usepackage[normalem]{ulem}
\usepackage{booktabs}

\definecolor{blueviolet}{rgb}{0.2, 0.2, 0.6}
\definecolor{webgreen}{rgb}{0,.5,0}
\definecolor{webbrown}{rgb}{.6,0,0}
\usepackage[pdftex,
  bookmarks=false,
  colorlinks=true, 
  urlcolor=webbrown,
  linkcolor=blueviolet, 
  citecolor=webgreen,
  pdfstartpage=1,
  pdfstartview={FitH},  
  bookmarksopen=false
  ]{hyperref}
\usepackage{cleveref}

\makeatletter
\newcommand\RedeclareMathOperator{%
  \@ifstar{\def\rmo@s{m}\rmo@redeclare}{\def\rmo@s{o}\rmo@redeclare}%
}

\newcommand\rmo@redeclare[2]{%
  \begingroup \escapechar\m@ne\xdef\@gtempa{{\string#1}}\endgroup
  \expandafter\@ifundefined\@gtempa
     {\@latex@error{\noexpand#1undefined}\@ehc}%
     \relax
  \expandafter\rmo@declmathop\rmo@s{#1}{#2}}
\newcommand\rmo@declmathop[3]{%
  \DeclareRobustCommand{#2}{\qopname\newmcodes@#1{#3}}%
}
\@onlypreamble\RedeclareMathOperator
\makeatother

\allowdisplaybreaks

\RedeclareMathOperator*{\E}{{\mathbb{E}}}

\DeclareMathOperator*{\argmin}{arg\,min}

\DeclareMathOperator*{\dtr}{d_{\mathrm{tr}}}
\DeclareMathOperator*{\davg}{d_{\mathrm{avg}}}
\DeclareMathOperator*{\dworst}{d_{\mathbin{\Diamond}}}

\newcommand{\floor}[1]{\left\lfloor#1\right\rfloor}
\newcommand{\ceil}[1]{\left\lceil#1\right\rceil}
\newcommand{\bigo}{\mathcal{O}}

\newcommand*\Let[2]{Let #1 $\gets$ #2}

\makeatletter
\DeclareFontFamily{OMX}{MnSymbolE}{}
\DeclareSymbolFont{MnLargeSymbols}{OMX}{MnSymbolE}{m}{n}
\SetSymbolFont{MnLargeSymbols}{bold}{OMX}{MnSymbolE}{b}{n}
\DeclareFontShape{OMX}{MnSymbolE}{m}{n}{
    <-6>  MnSymbolE5
   <6-7>  MnSymbolE6
   <7-8>  MnSymbolE7
   <8-9>  MnSymbolE8
   <9-10> MnSymbolE9
  <10-12> MnSymbolE10
  <12->   MnSymbolE12
}{}
\DeclareFontShape{OMX}{MnSymbolE}{b}{n}{
    <-6>  MnSymbolE-Bold5
   <6-7>  MnSymbolE-Bold6
   <7-8>  MnSymbolE-Bold7
   <8-9>  MnSymbolE-Bold8
   <9-10> MnSymbolE-Bold9
  <10-12> MnSymbolE-Bold10
  <12->   MnSymbolE-Bold12
}{}

\let\llangle\@undefined
\let\rrangle\@undefined
\DeclareMathDelimiter{\llangle}{\mathopen}%
                     {MnLargeSymbols}{'164}{MnLargeSymbols}{'164}
\DeclareMathDelimiter{\rrangle}{\mathclose}%
                     {MnLargeSymbols}{'171}{MnLargeSymbols}{'171}
\makeatother
\newcommand{\kket}[1]{|#1\rrangle}
\newcommand{\bbra}[1]{\llangle #1|}
\newcommand{\bbrakket}[2]{\llangle #1|#2 \rrangle}

\numberwithin{equation}{section}

\newtheorem{theorem}{Theorem}
\newtheorem{prop}{Proposition}
\newtheorem{lemma}{Lemma}
\newtheorem{corollary}{Corollary}
\newtheorem{definition}{Definition}

\newtheorem{claim}{Claim}

\newcommand{\stkout}[1]{\ifmmode\text{\sout{\ensuremath{#1}}}\else\sout{#1}\fi}
\newif\ifverbose
\verbosefalse
\newcommand{\ins}[1]{\ifverbose\textcolor{blue}{#1}\else#1\fi}
\newcommand{\edit}[2]{\ifverbose\textcolor{red}{\stkout{#1} #2}\else#2\fi}
\newcommand{\del}[1]{\ifverbose\textcolor{red}{\stkout{#1}}\fi}

\begin{document}
\title{Learning quantum states and unitaries of bounded gate complexity}

\setlength{\skip\footins}{12pt}
\def\thefootnote{$\star$}\footnotetext{These authors contributed equally to this work.}\def\thefootnote{\arabic{footnote}}

\author{Haimeng Zhao$^\star$}
\email{haimengzhao@icloud.com}
\affiliation{Institute for Quantum Information and Matter, Caltech, Pasadena, CA, USA}
\affiliation{Tsinghua University, Beijing, China}

\author{Laura Lewis$^\star$}
\email{llewis@alumni.caltech.edu}
\affiliation{Institute for Quantum Information and Matter, Caltech, Pasadena, CA, USA}
\affiliation{Google Quantum AI, Venice, CA, USA}

\author{Ishaan Kannan$^\star$}
\email{ishaan@alumni.caltech.edu}
\affiliation{Institute for Quantum Information and Matter, Caltech, Pasadena, CA, USA}

\author{Yihui~Quek}
\email{yquek@mit.edu}
\affiliation{Harvard University, 17 Oxford Street, Cambridge, MA, USA}
\affiliation{Massachusetts Institute of Technology, Cambridge, MA, USA}

\author{Hsin-Yuan Huang}
\email{hsinyuan@caltech.edu}
\affiliation{Institute for Quantum Information and Matter, Caltech, Pasadena, CA, USA}
\affiliation{Google Quantum AI, Venice, CA, USA}
\affiliation{Massachusetts Institute of Technology, Cambridge, MA, USA}

\author{Matthias C.~Caro}
\email{matthias.caro@fu-berlin.de}
\affiliation{Institute for Quantum Information and Matter, Caltech, Pasadena, CA, USA}
\affiliation{Freie Universit\"at Berlin, Berlin, Germany}


\begin{abstract}
While quantum state tomography is notoriously hard, most states hold little interest to practically-minded tomographers. Given that states and unitaries appearing in Nature are of bounded gate complexity, it is natural to ask if efficient learning becomes possible.
In this work, we prove that to learn a state generated by a quantum circuit with $G$ two-qubit gates to a small trace distance, a sample complexity scaling linearly in $G$ is necessary and sufficient. We also prove that the optimal query complexity to learn a unitary generated by $G$ gates to a small average-case error scales linearly in $G$.
While sample-efficient learning can be achieved, we show that under reasonable cryptographic conjectures, the computational complexity for learning states and unitaries of gate complexity $G$ must scale exponentially in $G$.
We illustrate how these results establish fundamental limitations on the expressivity of quantum machine learning models and provide new perspectives on no-free-lunch theorems in unitary learning.
Together, our results answer how the complexity of learning quantum states and unitaries relate to the complexity of creating these states and unitaries.
\end{abstract}

\maketitle

\vspace{-2em}
{\renewcommand\addcontentsline[3]{} 
\section{Introduction}}

A central problem in quantum physics is to characterize a quantum system by constructing a full classical description of its state or its unitary evolution based on data from experiments. These two tasks, named \emph{quantum state tomography}~\cite{banaszek2013focus,blume2010optimal,gross2010quantum,hradil1997quantum} and \emph{quantum process tomography}~\cite{mohseni2008quantum,o2004quantum,scott2008optimizing,chuang1997prescription,dariano2001quantum}, are (in)famous for being ubiquitous yet highly expensive. The applications of tomography include quantum metrology~\cite{giovannetti2006quantum, giovannetti2011advances}, verification~\cite{kokail2021entanglement, carrasco2021theoretical}, benchmarking~\cite{mohseni2008quantum,o2004quantum,scott2008optimizing,chuang1997prescription,levy2021classical,merkel2013self,blume2017demonstration,huang2022foundations}, and error mitigation \ins{\cite{cai2022quantum,van2023probabilistic, gupta2024probabilistic}}. Yet tomography provably requires exponentially many (in the system size $n$) copies of the unknown state \cite{haah2017sample, o2016efficient} or runs of the unknown process \cite{haah2023query}. This intuitively arises from the exponential scaling of the number of parameters needed to describe an {\em arbitrary} quantum system. 

But the situation is less dire than it theoretically appears. In practice, tools for analyzing many-body systems often exploit {\em known structures} cleverly to predict their phenomenology or classically simulate them. 
Notable examples include the BCS theory for superconductivity~\cite{bardeen1957theory}, tensor networks~\cite{orus2014practical, torlai2023quantum}, and neural network~\ins{\cite{carleo2017solving, torlai2018neural, zhao2023empirical,favoni2022lattice,apte2024deep,medvidovic2024neural}} Ansätze.
Indeed, while {\em most} of the states or unitaries may have exponential gate complexity \cite{brandao2021models}, such objects are also unphysical: an exponentially-complex state or unitary cannot be produced in Nature with a reasonable amount of time~\cite{poulin2011quantum}.
In particular,~\cite{poulin2011quantum} shows that quantum states/unitaries with bounded gate complexity are precisely those that can be produced by bounded-time evolution of time-dependent local Hamiltonians.

In this work, we study if tomography, too, can benefit from the observation that Nature can only produce states and unitaries with bounded complexity. This gives rise to the following main question.
\vspace*{2mm}
\begin{center}
    \textit{Can we efficiently learn states/unitaries of bounded gate complexity?}
\end{center}
In particular, we consider the following two tasks:
\begin{enumerate}
    \item Given copies (samples) of a pure quantum state $\ket{\psi}$ generated by $G$ two-qubit gates, learn $\ket{\psi}$ to within $\epsilon$ trace distance; see \Cref{fig:access_model}(a).
    \item Given uses (queries) of a unitary $U$ composed of $G$ two-qubit gates, learn $U$ to within $\epsilon$ root mean squared trace distance between output states (average-case learning); see \Cref{fig:access_model}(b).
\end{enumerate}
Note that the $G$ quantum gates can act on arbitrary pairs of qubits without any geometric locality constraint.
By allowing general gates beyond discrete gate sets, this setting encompasses continuous time-dependent Hamiltonian dynamics via Trotterization \cite{poulin2011quantum} and thus analog quantum simulation \cite{georgescu2014quantum}.
It also includes states heavily studied in condensed matter such as symmetry-protected topologically ordered states \cite{zeng2019quantum,malz2024preparation,stephen2024nonlocal} and tensor network states \cite{foss2021holographic,schon2005sequential,huang2015quantum}.
Previously, \cite{aaronson2018shadow} showed that Task 1 can be accomplished with a sample complexity of $\tilde{\bigo}(nG^2/\epsilon^4)$. 
In our work, we present algorithms for both of these tasks that use a number of samples/queries linear in the circuit complexity $G$ up to logarithmic factors. 
Moreover, the sample complexity is independent of system size.
Thus, for $G$ scaling polynomially with the number of qubits, our learning procedures improve upon previous work~\cite{aaronson2018shadow} and have significantly lower sample/query complexities than required for general tomography, answering our central question affirmatively.
We also prove matching lower bounds (up to logarithmic factors), showing that our algorithms are effectively optimal.
Moreover, we show that the focus on average-case learning is crucial in the case of unitaries: unitary tomography up to error $\epsilon$ in diamond distance (a worst-case metric over input states) requires a number of queries scaling exponentially in $G$, establishing an exponential separation between average and worst case.

\begin{figure}[t]
    \centering
    \includegraphics[width=\linewidth]{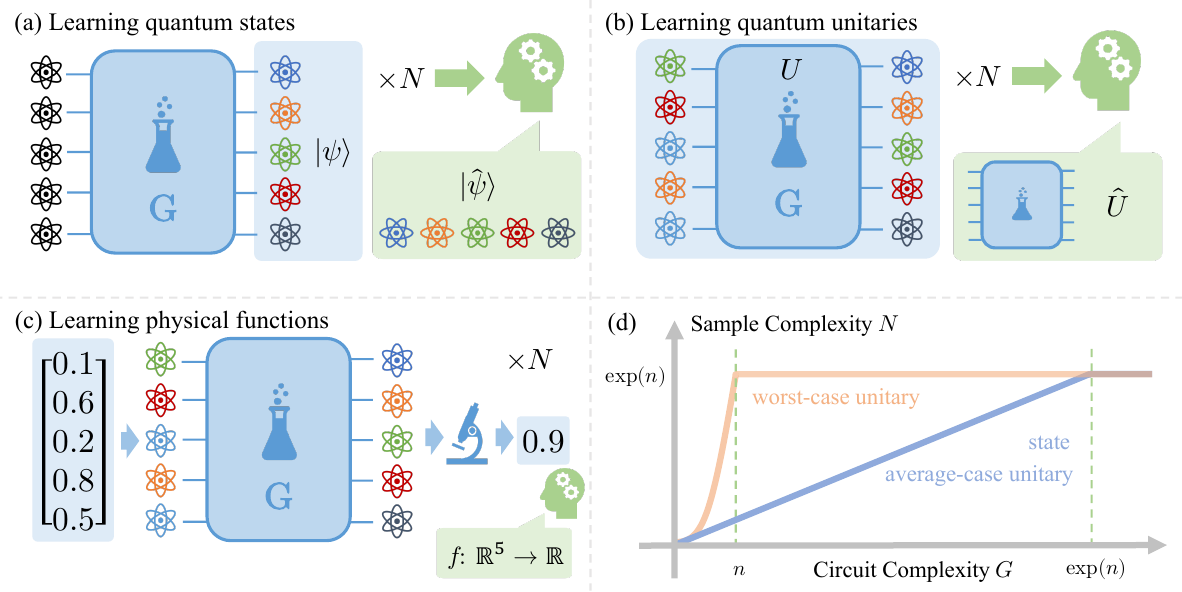}
    \caption{(a)-(c) Schematic overview of the learning models in this work. (a) Learning quantum states with bounded circuit complexity $G$. (b) Learning unitaries with bounded circuit complexity $G$. (c) Learning classical functions from quantum experiments with bounded circuit complexity $G$. (d) \ins{A conceptual depiction of the} \edit{S}{s}ample complexity of learning states in trace distance and unitaries in average-case distance scales linearly with circuit complexity, while that of learning unitaries in worst-case distance scales exponentially.}
    \label{fig:access_model}
\end{figure}

While our learning algorithms for bounded-complexity states and unitaries are efficient in terms of sample/query complexity, they are not computationally efficient. We prove that this is unavoidable. Assuming the quantum subexponential hardness of Ring Learning with Errors (\textsf{RingLWE})~\cite{lyubashevsky2010ideal, regev2009lattices, arunachalam2021quantum, diakonikolas2022cryptographic, aggarwal2023lattice, ananth2023revocable}, any quantum algorithm that learns arbitrary states/unitaries with $\tilde{\bigo}(G)$ gates requires computational time scaling exponentially in $G$.
This result highlights a significant computational complexity limitation on learning even comparatively simple states and unitaries. This result also answers an open question in \cite{anshu2023survey}.
Meanwhile, we show that $\mathrm{poly}(n)$-time algorithms are possible for $G=\mathcal{O}(\log n)$.
Together, this establishes a crossover in computational hardness at $G \sim \log n$, kicking in far before the sample complexity becomes exponential (at $G = \exp(n)$). 
This means that relatively few samples/queries already contain enough information for the learning task, but it is hard to retrieve the information.

Finally, we study two variations of unitary learning which deepen our insights about the problem. The first variation utilizes classical (not quantum) descriptions of input and output pairs, and explains why both learning states and unitaries display a linear-in-$G$ sample complexity: The underlying source of complexity in learning unitaries is, in fact, the readout of input and output quantum states, rather than learning the mapping. We generalize recent quantum no-free lunch theorems \cite{poland2020no, sharma2022reformulation} to reach this conclusion. 
For the second variation we study quantum machine learning (QML) models. We focus on learning classical functions that map variables controlling the input states and the evolution to some experimentally observed property of the outputs (\Cref{fig:access_model}(c)). Surprisingly, we find that certain well-behaved many-variable functions can in fact not (even approximately) be implemented by quantum experiments with bounded complexity. This highlights a fundamental limitation on the functional expressivity of both Nature and practical QML models.

{\renewcommand\addcontentsline[3]{} 
\section{Results}}
In this section, we discuss our rigorous guarantees for learning quantum states and unitaries with circuit complexity $G$. Our sample complexity results are summarized in \Cref{tab:sample_complexity} and \Cref{fig:access_model}(d).

\begin{table}
\begingroup
\setlength{\tabcolsep}{10pt} 
\renewcommand{\arraystretch}{1.5} 
\begin{tabular}{cccc}
\toprule
Sample complexity 
& State 
& Unitary (average-case) 
& Unitary (worst-case) \\ 
\midrule
Upper bound 
& $\tilde{\mathcal{O}}\left(G/\epsilon^2\right)$
& $\tilde{\mathcal{O}}\left(G\min\left\{1/\epsilon^2, \sqrt{2^n}/\epsilon\right\}\right)$ 
& $\tilde{\mathcal{O}}\left(2^n G/\epsilon\right)$ \\
Lower bound 
& $\tilde{\Omega}\left(G/\epsilon^2\right)$ 
& $\Omega\left(G/\epsilon\right)$ 
& $\Omega\left(2^{\min\{G/(2C), n/2\}}/{\epsilon}\right)$ \\ 
\bottomrule
\end{tabular}
\endgroup
\caption{\textbf{Sample complexity of learning $n$-qubit states and unitaries with circuit complexity $G$.} The learning accuracy $\epsilon$ is measured in trace distance for states, root mean squared trace distance for average case unitary learning, and diamond distance for worst case. Here, $C>0$ is some universal constant. Throughout the manuscript, $\tilde{\mathcal{O}}, \tilde{\Theta}$ and $\tilde{\Omega}$ denote that we are suppressing non-leading logarithmic factors.}
\label{tab:sample_complexity}
\end{table}

We also present computational complexity results, where we establish the exponential-in-$G$ growth of computational complexity, implying that $\log n$ gate complexity is a transition point at which learning becomes computationally inefficient. 
In particular, we prove that for circuit complexity $\tilde{\bigo}(G)$, any quantum algorithm for learning states in trace distance or unitaries in average-case distance must use time exponential in $G$, under the conjecture that \textsf{RingLWE} cannot be solved by a quantum computer in sub-exponential time.
Hence, for a number $G$ of gates that scales slightly higher than $\log n$, the learning tasks cannot be solved by any polynomial-time quantum algorithm under the same conjecture.
Meanwhile, for $G=\mathcal{O}(\log n)$, both learning tasks can be solved efficiently in polynomial time.

\vspace{2em}
{\renewcommand\addcontentsline[3]{} 
\subsection{Learning quantum states}}

We consider the task of learning quantum states of bounded circuit complexity.
Let $\ket{\psi} = U\ket{0}^{\otimes n}$ be an $n$-qubit pure state generated by a unitary $U$ consisting of $G$ two-qubit gates acting on the zero state.
Throughout this section, we denote $\rho \triangleq \ketbra{\psi}$.
Given $N$ identically prepared copies of $\rho$, the goal is to output a classical circuit description of a quantum state $\hat{\rho}$ that is $\epsilon$-close to $\rho$ in trace distance: $\dtr(\hat{\rho}, \rho) = \norm{\hat{\rho} - \rho}_1 /2 < \epsilon$.
We establish the following theorem, which states that linear-in-$G$ many samples (up to logarithmic factors) are both necessary and sufficient to learn the unknown quantum state $\ket{\psi}$ within a small trace distance.

\vspace*{1mm}
\begin{theorem}[State learning]
\label{thm:state-learning}
Suppose we are given $N$ copies of an $n$-qubit pure state $\rho = \ketbra{\psi}$, where $\ket{\psi} = U\ket{0}^{\otimes n}$ is generated by a unitary $U$ consisting of $G$ two-qubit gates. Then, $N = \tilde{\Theta}\left(G/\epsilon^2\right)$ copies are necessary and sufficient to learn the state within $\epsilon$-trace distance $\dtr$ with high probability.
\end{theorem}

Previous work~\cite{aaronson2018shadow} obtained a sample complexity of $\tilde{\bigo}(nG^2/\epsilon^4)$ for this task, which we show to be sub-optimal.
Notably, our result achieves the optimal scaling in both $G$ and $\epsilon$ up to logarithmic factors and is independent of the system size $n$.
Thus, we completely characterize the sample complexity, resolving an open question from~\cite{yu2023learning}.
We prove the upper bound in \Cref{app:states-upper}, utilizing covering nets~\cite{vershynin2018high} and quantum hypothesis selection~\cite{bădescu2023improved}.
Our proposed algorithm first creates a covering net over the space of all unitaries consisting of $G$ two-qubit gates.
This can easily be transformed into a covering net over the space of all quantum states generated by $G$ two-qubit gates by applying each element of the unitary covering net to the zero state.
Thus, any quantum state generated by $G$ two-qubit gates is close (in trace distance) to some element of the covering net.
We can then apply quantum hypothesis selection~\cite{bădescu2023improved} to the covering net, which allows us to identify the element in the covering net that is close to the unknown target state $\ket{\psi}$ and achieve the optimal $\epsilon$ dependence.
We also note that our algorithm for learning quantum states does not require knowledge of or access to the unitary $U$ which generates the unknown state $\ket{\psi}$.
Only the condition that some unitary $U$ consisting of $G$ gates generates $\ket{\psi}$ is needed.
The lower bound is proven in \Cref{app:states-lower} by using an information-theoretic argument via reduction to distinguishing a packing net over $G$-gate states~\cite{haah2017sample}.

Our algorithm to learn the unknown quantum state $\ket{\psi}$ is computationally inefficient, as it requires a search over a covering net whose cardinality is exponential in $G$.
We show that for circuits of size $\tilde{\mathcal{O}}(G)$, any quantum algorithm that can learn $\ket{\psi}$ to within a small trace distance given access to copies of this state must use time exponential in $G$, under commonly-believed cryptographic assumptions \cite{lyubashevsky2010ideal, regev2009lattices, arunachalam2021quantum, diakonikolas2022cryptographic, aggarwal2023lattice, ananth2023revocable}.
Meanwhile, the learning task is computationally-efficiently solvable for $G=\mathcal{O}(\log n)$ via junta learning \cite{chen2023testing} and standard tomography methods.
This implies a transition point of computational efficiency at $\log n$ circuit complexity.
Previous work~\cite{yang2023complexity,abbas2023quantum} arrives at similar hardness results for polynomial circuit complexity, but our detailed analysis allows us to sharpen the computational lower bound and obtain this transition point.
\cite{hinsche2023one} also proves computational complexity lower bounds for distribution learning that are similar in spirit.

\begin{theorem}[State learning computational complexity]
    \label{thm:state-comp-complex}
    Suppose we are given $N$ copies of an unknown $n$-qubit pure state $\ket{\psi} = U\ket{0}^{\otimes n}$ generated by an arbitrary unknown unitary $U$ consisting of $\tilde{\mathcal{O}}(G)$ two-qubit gates.
    Suppose that \textsf{RingLWE} cannot be solved by a quantum computer in sub-exponential time.
    Then, any quantum algorithm that learns the state to within $\epsilon$ trace distance $\dtr$ must use $\exp(\Omega(\min\{G, n\}))$ time.
    Meanwhile, for $G=\mathcal{O}(\log n)$, the learning task can be solved in polynomial time.
\end{theorem}

{\renewcommand\addcontentsline[3]{} 
\subsection{Learning quantum unitaries}}
For learning unitaries, a natural distance metric analogous to the trace distance for states is the diamond distance $\dworst(U, V) = \max_{\rho} \|(U\otimes I)\rho(U\otimes I)^\dagger - (V\otimes I)\rho(V\otimes I)^\dagger\|_1$, where $\rho$ is over any arbitrarily extended Hilbert space.
It characterizes the optimal success probability for discriminating between two unitary channels. 
Moreover, it can be reinterpreted in terms of the largest distance between $U\ket{\psi}$ and $V\ket{\psi}$ over all input states $\ket{\psi}$, and thus represents the error we make in the worst case over input states. 
We find that in this worst-case learning task, a number of queries exponential in $G$ is necessary to learn the unitary.

\vspace*{1mm}
\begin{theorem}[Worst-case unitary learning]
    To learn an $n$-qubit unitary composed of $G$ two-qubit gates to accuracy $\epsilon$ in diamond distance $\dworst$ with high probability, any quantum algorithm must use at least $\Omega\left(\nicefrac{2^{\min\{G/(2C), n/2\}}}{\epsilon}\right)$ queries to the unknown unitary, where $C>0$ is a universal constant. 
    Meanwhile, there exists such an algorithm using $\tilde{\mathcal{O}}(\nicefrac{2^nG}{\epsilon})$ queries. 
    \label{thm:worst-case-unitary}
\end{theorem}

The complete proof is given in \Cref{app:worst-case-unitary}\ins{,} and \ins{the proof of the lower bound} relies on the adversary method~\cite{van2019quantum, ambainis2000quantum, belovs2015variations, hoyer2007negative}.
We construct a set of unitaries that a worst-case learning algorithm can successfully distinguish, but that only make minor differences when acting on states so that a minimal number of queries have to be made in order to distinguish them.
The upper bound is achieved by the average-case learning algorithm in \Cref{thm:avg-case-unitary} below when applied in the regime of exponentially small error.

Having established this no-go theorem for worst-case learning, we turn to a more realistic average-case learning alternative. Here, the accuracy is measured using the root mean squared trace distance between output states over Haar-random inputs, $\davg(U, V) = \sqrt{\mathbb{E}_{\ket{\psi}}[\dtr(U\ket{\psi}, V\ket{\psi})^2]}$. 
This metric characterizes the average error when testing the learned unitary on randomly chosen input states.

We find that, similarly to the state learning task, linear-in-$G$ many queries are both necessary and sufficient to learn a unitary in the average case.

\begin{theorem}[Average-case unitary learning]
    There exists an algorithm that learns an $n$-qubit unitary composed of $G$ two-qubit gates to accuracy $\epsilon$ in root mean squared trace distance $\davg$ with high probability using $\tilde{\mathcal{O}}\left(G\min\{\nicefrac{1}{\epsilon^2}, \nicefrac{\sqrt{2^n}}{\epsilon}\}\right)$ queries to the unknown unitary. 
    Meanwhile, $\Omega\left(\nicefrac{G}{\epsilon}\right)$ queries to the unitary or its inverse or the controlled versions are necessary for any such algorithm.
    \label{thm:avg-case-unitary}
\end{theorem}

We show the upper bound in \Cref{app:avg-case-unitary-up} by combining a covering net with quantum hypothesis selection similarly to the upper bound in \Cref{thm:state-learning}. 
Our algorithm achieving the query complexity $\tilde{\bigo}(\nicefrac{G}{\epsilon^2})$ uses maximally entangled states and the Choi–Jamio{\l}kowski duality~\cite{choi1975completely, jamiolkowski1972linear, jiang2013channel}.
With a bootstrap method similar to quantum phase estimation \cite{haah2023query}, we improve the $\epsilon$-dependence to the Heisenberg scaling $\tilde{\bigo}(\nicefrac{1}{\epsilon})$, albeit at the cost of a dimensional factor.
\ins{It is an open question as to whether one can improve the $\epsilon$-dependence without incurring this dimensional factor.}
Without auxiliary systems, we prove a query complexity bound of $\tilde{\bigo}(G\min\{\nicefrac{1}{\epsilon^4}, \nicefrac{\left(\sqrt{2^n}\right)^3}{\epsilon}\})$.
The lower bound is proven in \Cref{app:avg-case-unitary-lo} by mapping to a fractional query problem \cite{haah2023query, berry2015hamiltonian, cleve2009efficient} and making use of a recent upper bound on the success probability in unitary distinguishing tasks \cite{bavaresco2022unitary}. 
In the case of learning generic unitaries, our result yields a $\Omega(4^n/\epsilon)$ lower bound, improving upon the $\Omega(4^n/n^2)$ bound from the recent work \cite{yang2023complexity},  which studies the hardness of learning Haar-(pseudo)random unitaries.

As Haar-random states are hard to generate in practice, we also discuss other input state ensembles of physical interest.
Relying on the equivalence of root mean squared trace distances over different locally scrambled ensembles \cite{kuo2020markovian, hu2023classical}, recently established in \cite{caro2022out}, our algorithm achieves the same average-case guarantee over any such ensemble. 
Notable examples of locally scrambled ensembles include products of Haar-random single-qubit states or of random single-qubit stabilizer states, $2$-designs on $n$-qubit states, and output states of random local quantum circuits with any fixed architecture.

The similar linear-in-$G$ sample/query complexity scaling in \Cref{thm:state-learning,thm:avg-case-unitary} hints at a common underlying source of complexity. 
However, in contrast to state learning, unitary learning comes with two natural such sources: (1) to readout input and output states, and (2) to learn the mapping from inputs to outputs. 
The similarity between learning states and unitaries in terms of complexities suggests that the former may encapsulate the central difficulty in unitary learning whereas the latter may be easy.
This seemingly contradicts recent quantum no-free-lunch theorems~\cite{poland2020no, sharma2022reformulation, wang2023transition}, which state $\Omega(2^n)$ samples are required to learn a generic unitary even from classical descriptions of input-output state pairs, highlighting the difficulty of (2).

To resolve this apparent contradiction, we reformulate the quantum no-free-lunch theorem (\Cref{thm:qnfl}) from a unifying information-theoretic perspective in \Cref{app:classical-description}. 
We highlight that enlarging the space for the classically described data allows to systematically reduce the sample complexity until a single sample suffices to learn a general unitary.
Therefore, the difficulty of learning the mapping, as indicated by quantum no-free-lunch theorems, vanishes when we allow auxiliary systems and query access to the unitary. 
Inspired by this observation, we give two ways of enlarging the representation space with auxiliary systems. The first is fundamentally quantum, making use of entangled input states~\cite{sharma2022reformulation}. The other is purely classical, relying on mixed state inputs~\cite{yu2023optimal}.

\begin{theorem}[Learning with classical descriptions]
    There exists an algorithm that learns a generic $n$-qubit unitary with any non-trivial accuracy and with high success probability using $\bigo(2^n/r)$ classically described input-output pairs with mixed (entangled) input states of (Schmidt) rank $r$.
    Moreover, any such algorithm that is robust to noise needs at least $\Omega(2^n/r)$ samples.
    \label{thm:classical-description}
\end{theorem}

Similarly to the case for state learning, our average-case unitary learning algorithm is not computationally efficient.
We show that this cannot be avoided. Under commonly-believed cryptographic assumptions \cite{lyubashevsky2010ideal, regev2009lattices, arunachalam2021quantum, diakonikolas2022cryptographic, aggarwal2023lattice, ananth2023revocable}, any quantum algorithm that can learn unknown unitaries with circuit size $\tilde{\bigo}(G)$ to a small error in average-case distance from queries must have a computational time exponential in $G$.
This implies the same computational hardness for worst-case unitary learning, and a $\log n$ transition point of computational efficiency.
Note that the hard instances we construct are implementable with a similar number of Clifford and T gates \cite{lai2022learning}. Therefore, together with \Cref{thm:state-comp-complex}, this implies that there is no polynomial time quantum algorithms for learning Clifford+T circuits with $\tilde{\omega}(\log n)$ T gates, answering an open question (the fifth question) in the survey \cite{anshu2023survey} negatively.

\begin{theorem}[Unitary learning computational complexity]
    \label{thm:unitary-comp-complex}
    Suppose we are given $N$ queries to an arbitrary unknown $n$-qubit unitary $U$ consisting of $\tilde{\bigo}(G)$ two-qubit gates.
    Assume that \textsf{RingLWE} cannot be solved by a quantum computer in sub-exponential time.
    Then, any quantum algorithm that learns the unitary to within $\epsilon$ average-case distance $\davg$ must use $\exp(\Omega(\min\{G, n\}))$ time.
    Meanwhile, for $G=\mathcal{O}(\log n)$, the learning task can be solved in polynomial time.
\end{theorem}

{\renewcommand\addcontentsline[3]{} 
\subsection{Learning \ins{with} physical functions}}

Apart from learning quantum states and dynamics themselves, a more classically minded learner may care more about learning classical functions resulting from quantum processes.
We define these \textit{physical functions} in \Cref{app:learning-function} as functions $f(x, \{U_i\}_{i=1}^G, a)$ mapping $x\in[0, 1]^\nu$ to $\mathbb{R}$ resulting from a physical experiment consisting of three steps: (1) a fixed state preparation procedure that can depend on $x$; (2) a unitary evolution consisting of $G$ tunable two-qubit gates $\{U_i\}_{i=1}^G$ and arbitrary fixed unitaries that can depend on $x$, arranged in a circuit architecture $a$; (3) the measurement of a fixed observable, whose expectation is the function output.
By tuning the local gates $\{U_i\}_{i=1}^G$ and potentially changing architecture $a$, we obtain a resulting class of functions that can be implemented in this general experimental setting.
Despite the generality of this setup, we find that certain well-behaved functions are actually not physical in this sense: they cannot be efficiently approximated or learned via physical functions.

\vspace*{1mm}
\begin{theorem}[Approximating and learning with physical functions]
    To approximate and learn arbitrary $1$-bounded and $1$-Lipschitz $\mathbb{R}$-valued functions on $[0, 1]^\nu$ to accuracy $\epsilon$ in $\|\cdot\|_\infty$ with high probability, using physical functions with $G$ gates and variable circuit structures, we must use $G\geq \tilde{\Omega}(\nicefrac{1}{\epsilon^{\nu/2}})$ gates and collect at least $\Omega(\nicefrac{1}{\epsilon^{\nu}})$ samples.
    \ins{If the circuit structure is fixed, we require $G \geq \tilde{\Omega}(1/\epsilon^\nu)$ gates.}
    \label{thm:approx-functions}
\end{theorem}

We prove this in \Cref{app:learning-function} by noting that to approximate arbitrary $1$-bounded and $1$-Lipschitz functions well, the complexity of experimentally implementable functions cannot be too small, as measured by pseudo-dimension \cite{pollard1984convergence} or fat-shattering dimension \cite{kearns1994efficient}. 
Then the gate complexity lower bound follows because the function class complexity is limited by the circuit complexity \cite{caro2020pseudo}, and we can appeal to results in classical learning theory \cite{anthony1999neural} to obtain our sample complexity lower bound.

It has been established that a classical neural network can learn to approximate any $1$-bounded and $1$-Lipschitz functions to accuracy $\epsilon$ in $\|\cdot\|_\infty$ with $\tilde{\Theta}(1/\epsilon^\nu)$ parameters, exponential in the number of variables $\nu$, known as the curse of dimensionality~\cite{grohs2022mathematical}. 
Our results show that quantum neural networks can do no better.
This result not only is relevant to the practical implementation of quantum machine learning, complementing existing results on the universal approximation of quantum neural networks~\cite{gonon2023universal, perez2021one, schuld2021effect, manzano2023parametrized}, but also has deep implications to the physicality of the function class at consideration.
It means that there are some many-variable $1$-bounded and $1$-Lipschitz functions that cannot be implemented in Nature efficiently.
On the other hand, certain more restricted function classes can be approximated using only $\mathcal{O}(1/\epsilon^2)$ parameters with both classical~\cite{grohs2022mathematical} and quantum neural networks~\cite{gonon2023universal}, independent of the number of variables.
This reveals a fundamental limitation on the functional expressivity of Nature, practical QML models, and quantum signal processing algorithms \cite{martyn2021grand,gilyen2019quantum}.
\ins{We remark that while we prove this no-go result, achieving a quantum advantage may still be possible for other function classes~\cite{yu2023provable}.}

\ins{
{\renewcommand\addcontentsline[3]{}
\section{Numerical Experiments}
\label{sec:numerics}}
}

\begin{figure}
    \centering
    \includegraphics[width=0.9\linewidth]{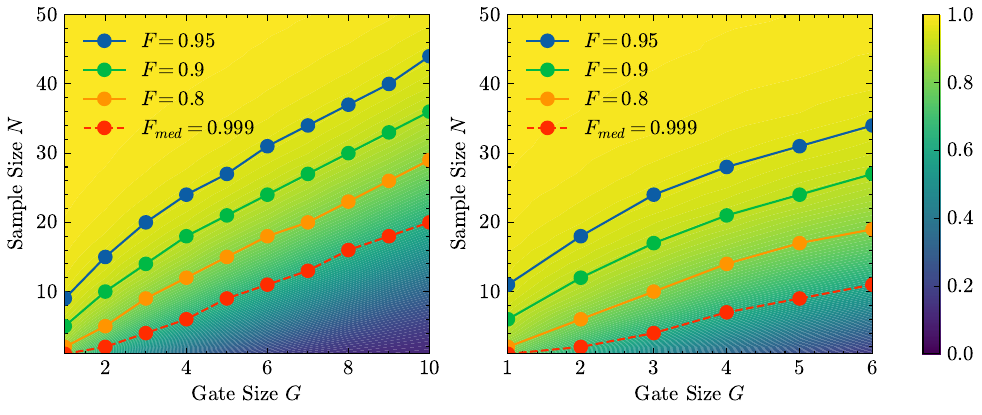}
    \caption{\ins{\textbf{Sample complexity $N$ of the learning algorithm with different gate numbers $G$ and reconstruction fidelity $F$.} 
    The unknown target states are pure states on $n=10000$ qubits generated from $G$ gates, either concentrated on the first $4$ qubits \textbf{(left)} or randomly placed \textbf{(right)}.
    The contour plot represents the fidelity for different $G$ and $N$ averaged over many random instances.
    Sample complexities with average fidelity $F$ and median fidelity $F_\mathrm{med}$ are plotted in solid and dashed lines, respectively.}}
    \label{fig:numerics}
\end{figure}

\ins{
To support our theoretical findings, we conduct numerical experiments using our learning algorithm applied to pure states generated from $G$ two-qubit gates.
The results reflect the linear-in-$G$ scaling of the sample complexity $N$ from \Cref{thm:state-learning}.
We consider a large system size $n=10000$ to illustrate that the sample complexity is independent of $n$.
We study two families of unknown target states with different gate configurations: (a) the $G$ gates are concentrated on $4$ qubits; and (b) the $G$ gates are randomly placed.
We note that Case (a) corresponds to the hard-to-learn states that we construct to prove the $\tilde{\Omega}(G)$ lower bound in \Cref{thm:state-learning}, while in Case (b) the gates are expected to spread out over the large system and form shallow circuits \cite{huang2024learning}.
Due to the exponential-in-$G$ computation time proved in \Cref{thm:state-comp-complex}, we restrict the gate size to $G=10$ in Case (a) and $G=6$ in Case (b).
We perform the simulations by implementing the algorithm from \Cref{app:states-upper} using shallow Clifford classical shadows~\cite{bertoni2024shallow,schuster2024random}.
The details of the numerical implementation are provided in \Cref{app:numerics}, and the code can be found at \url{https://github.com/haimengzhao/bounded-gate-tomography}.
}

\ins{
The performance of our learning algorithm is shown in \Cref{fig:numerics}.
\Cref{fig:numerics} (Left) corresponds to Case (a) described in the previous paragraph, and \Cref{fig:numerics} (Right) corresponds to Case (b).
We provide contour plots of the average fidelity $F$ of the reconstruction with different gate sizes $G$ and sample sizes $N$.
The sample complexity $N$ with different gate sizes $G$ is plotted in solid lines for different average fidelities $F$ and in dashed lines for the  median fidelity $F_\mathrm{med}$.
We see a linear dependence of sample complexity $N$ on gate size $G$, in accordance with our theoretical bound $N=\tilde{\Theta}(G)$ from \Cref{thm:state-learning}.
Moreover, we note that a relatively small sample size $N\sim 50$ suffices to learn states with $G\sim 10$ gates on very large system size $n=10000$, due to the fact that our sample complexity is independent of $n$.
}

\vspace{2em}
{\renewcommand\addcontentsline[3]{}
\section{Discussion}}
Our work provides a new, more fine-grained perspective on the fundamental problems of state and process tomography by analyzing them for the broad and physically relevant class of bounded-complexity states and unitaries.
It complements existing literature on learning restricted classes of states/unitaries or their properties. 
Examples include stabilizer circuits and states~\cite{aaronson2004improved,montanaro2017learning,low2009learning,rocchetto2018stabiliser}, Clifford circuits with few non-Clifford T gates and their output states~\cite{lai2022learning,grewal2023efficient,grewal2023efficient2, hangleiter2023bell, leone2022learning}, matrix product operators~\cite{torlai2023quantum} and states~\cite{cramer2010efficient,landon2010efficient,guo2023scalable}, phase states~\cite{bernstein1993quantum,montanaro2017learning,rotteler2009quantum,arunachalam2022optimal}, permutationally invariant states~\cite{toth2010permutationally,moroder2012permutationally,schwemmer2014experimental}, outputs of shallow quantum circuits \cite{rouze2021learning}, PAC learning quantum states \cite{aaronson2007learnability} and circuits \cite{cai2022sample}, shadow tomography~\cite{aaronson2018shadow}, classical shadow formalism~\cite{huang2020predicting, levy2021classical, elben2022randomized, kunjummen2023shadow}, and property prediction of the outputs of quantum processes~\cite{huang2022quantum, huang2022learning, caro2022learning}.
It also raises many interesting questions for future research.

Firstly, to account for decoherence and imperfections in realistic experiments, it is natural to generalize our results to mixed states and channels. 
As our learning algorithms based on hypothesis selection and classical shadows rely on the purity/unitarity of the unknown state/process, it seems that different algorithmic approaches would be needed to go beyond states of constant rank.
Moreover, while our results show that learners using only single-copy measurements and no coherent quantum processing can achieve  optimal sample/query complexity (in $G$) for pure state/unitary learning (in line with the state tomography protocol in \cite{o2016efficient}, which uses at most $\mathrm{rank}(\rho)$ copies at a time for the tomography of general state $\rho$), quantum-enhanced learners, using multi-copy measurements and coherent processing, may have an advantage in the case of mixed states and channels. Such a quantum advantage is known for general mixed state tomography \cite{chen2022tight, chen2023when-does} and in certain channel learning scenarios \cite{chen2022exponential, huang2022quantum, caro2022learning, chen2023unitarity, fawzi2023lower, fawzi2023quantum-channel-certification, oufkir2023sample}, however, to our knowledge not yet under assumptions of bounded complexity.

Secondly, there are several regimes of interest in which our results may be further extended. For instance, while we establish \ins{a} computational efficiency transition for state and unitary learning at logarithmic circuit complexity, we leave open the question of computationally efficient learning with constraints beyond circuit complexity (e.g., constant-depth circuits where the gates are spread out).
Another potential improvement related to the computational complexity is in regards to average-case computational hardness. While our computational lower bounds hold in the worst-case, this does not tell us if most states/unitaries of bounded gate complexity are computationally hard to learn. \edit{Are}{Is} there a worst-case to average-case reduction for this problem? Or perhaps is there an average-case notion of pseudorandomness that one could leverage here?
An additional regime where our work can be extended is as follows. Our adaptation of the bootstrap strategy from \cite{haah2023query} to average-case unitary learning achieves Heisenberg scaling only at the cost of a dimension-dependent factor. 
Given recent work in state shadow tomography \cite{van2021quantum, huggins2022nearly, van2023quantum}, it may not be possible to find a learner free from this dimensional factor while achieving the $\epsilon^{-1}$ scaling.
Finding such a learner or disproving its existence could serve as an important contribution to recent progress on Heisenberg-limited learning in different scenarios \cite{huang2022learningmanybodyhamiltonians, dutkiewicz2023advantage, li2023heisenberglimited}.

Thirdly, can we make learning even more efficient if the circuit structure is fixed and known in advance? Our upper bound already implies an algorithm with $\tilde{\bigo}(G)$ sample complexity for fixed circuit structure, but the lower bound proof crucially relies on the ability to place gates freely in the construction of the packing net.
A particular fixed circuit structure of physical relevance is the brickwork circuit \cite{fisher2023random}. In \Cref{app:brickwork}, we give preliminary results showing that if an $n$-qubit $G$-gate brickwork circuit suffices to implement an approximate unitary $t$-design \cite{brandao2016local}, then the metric entropy of this unitary class with respect to $\davg$ is lower bounded by $\Omega(tn)$. 
Considering the known lower bound of $G\geq \tilde{\Omega}(tn)$ on the size of brickwork circuits implementing $t$-designs \cite{brandao2016local}, whose tightness is still an open problem \cite{haferkamp2022random}, this may hint at a similar $\tilde{\Theta}(G)$ sample complexity of learning brickwork circuits.

Lastly, we outline a potential connection to the Brown-Susskind conjecture \cite{brown2018second, susskind2018black} originating from the wormhole-growth paradox in holographic duality \cite{susskind2016computational, bouland2019computational, stanford2014complexity, brown2016complexity}. Informally, the conjecture states that the complexity of a generic local quantum circuit grows linearly with the number of $2$-qubit gates for an exponentially long time, dual to the steady growth of a wormhole's volume in the bulk theory. 
With ``complexity'' understood as ``circuit complexity'' \cite{stanford2014complexity}, this conjecture has recently been confirmed for exact circuit complexity~\cite{haferkamp2022linear, li2022short} while the case of approximate circuit complexity is only partially resolved~\cite{oszmaniec2022saturation,haferkamp2023moments}.
Our work suggests an alternative approach to the Brown-Susskind conjecture. Namely, we have demonstrated that the complexity of learning quantum circuits grows linearly with the number of local gates in the worst case. If our bounds were extended to hold with high probability over random circuits with $G$ gates, this would yield a sample complexity version of the Brown-Susskind conjecture, suggesting the complexity of learning as a dual of the wormhole volume.

Via these open questions, tomography problems dating back to the early days of quantum computation and information connect closely to different avenues of current research in the field. Consequently, answering these questions will shed new light on fundamental quantum physics as well as on the frontiers of quantum complexity and quantum learning.

\vspace{2em}
{\renewcommand\addcontentsline[3]{} 
\section{Methods}}

In this section, we discuss the main ideas behind the proof of our results on the sample complexity of learning states (\Cref{thm:state-learning}) and unitaries (\Cref{thm:avg-case-unitary}), along with the computational complexity (\Cref{thm:state-comp-complex,thm:unitary-comp-complex}).

\vspace{1.0em}
{\renewcommand\addcontentsline[3]{} 
\subsection{Sample complexity upper bounds}}

We prove the upper bounds in \Cref{thm:state-learning,thm:avg-case-unitary} using a hypothesis selection protocol similar to \cite{bădescu2023improved}, but now based on classical shadow tomography \cite{huang2020predicting} that enables a linear-in-$G$ scaling.

\paragraph{State learning}
For state learning, we first take a minimal covering net $\mathcal{N}$ over the set of states with bounded circuit complexity $G$ such that for any such state $\ket{\psi}$, there exists a state in the covering net that is $\epsilon$-close to $\ket{\psi}$ in trace distance.
This net then serves as a set of candidate states from which the learning algorithm will select one.
Importantly, we prove that the cardinality of $\mathcal{N}$ can be upper bounded by $|\mathcal{N}| \leq e^{\tilde{\bigo}(G)}$.
Here, note that the tilde hides a logarithmic factor in terms of system size, which we remove using a more detailed analysis with ideas from junta learning~\cite{chen2023testing}.

Next, we use classical shadows created via random Clifford measurements~\cite{huang2020predicting} to estimate the trace distance between the unknown state and each of the candidates in $\mathcal{N}$.
This is achieved by estimating the expectation value of the Helstrom measurement \cite{Helstrom1969}, which is closely related to the trace distance between two states.
As the rank of Helstrom measurements between pure states is at most $2$, Clifford classical shadows can efficiently estimate all $\binom{|\mathcal{N}|}{2}$ of them simultaneously to $\epsilon$ error using $\bigo(\log|\mathcal{N}|/\epsilon^2)\leq \tilde{\bigo}(G/\epsilon^2)$ copies of $\ket{\psi}$.
Then we select the candidate that has the smallest trace distance from $\ket{\psi}$ as the output.

The above strategy leads to a sample complexity upper bound that depends logarithmically on the number of qubits $n$.
This is undesirable when the circuit complexity $G$ is smaller than $n/2$ (i.e., when some of the qubits are in fact never influenced by the circuit).
We improve our algorithm in this small-size regime by first performing a junta learning step \cite{chen2023testing} to identify which of the qubits are acted on non-trivially.
After that, we enhance our protocol with a measure-and-postselect step.
This allows us to construct a covering net only over the qubits acted upon non-trivially whose cardinality no longer depends on $n$. 
We then perform the hypothesis selection as before.
In this way, we are able to achieve a sample complexity independent of system size.

\paragraph{Unitary learning}
The algorithm for unitary learning is similar to the state learning protocol.
When allowing the use of an auxiliary system, we utilize the fact that the average-case distance between unitaries is equivalent to the trace distance between their Choi states. This way, we can reduce the problem to state learning of the Choi states and achieve the $\tilde{\bigo}(G/\epsilon^2)$ sample complexity.
Without auxiliary systems, we can sample random input states and perform one-shot Clifford shadows on the outputs to estimate the squared average-case distance, resulting in an $\tilde{\bigo}(G/\epsilon^4)$ sample complexity with a sub-optimal $\epsilon$-dependence.

Furthermore, we improve the $\epsilon$ dependence in unitary learning to the Heisenberg scaling $\tilde{\bigo}(1/\epsilon)$ via a bootstrap method similar to \cite{haah2023query}, using the above learning algorithm as a sub-routine.
Specifically, we iteratively refine our learning outcome $\hat{U}$ by performing hypothesis selection over a covering net of $(U\hat{U}^\dagger)^{p}$, with $p$ increasing exponentially as the iteration proceeds.
Although the circuit complexity of $(U\hat{U}^\dagger)^{p}$ grows with $p$, a covering net with $p$-independent cardinality can be constructed based on the one-to-one correspondence to $U$.
However, unlike the diamond distance learner considered in \cite{haah2023query}, which has fine control over every eigenvalue of the unitaries, our average-case learner only has control over the average of the eigenvalues.
Thus for the bootstrap to work (i.e., for the learning error to decrease with increasing $p$), the average-case learner has to work in an exponentially small error regime, which results in a dimensional factor in the final sample complexity $\tilde{\bigo}(\sqrt{2^n}G/\epsilon)$.

{\renewcommand\addcontentsline[3]{} 
\subsection{Sample complexity lower bounds}}

We prove the sample complexity lower bounds in \Cref{thm:state-learning,thm:avg-case-unitary} by reduction to distinguishing tasks. 
Specifically, if we can learn the state/unitary to within $\epsilon$ error, then we can use this learning algorithm to distinguish a set of states/unitaries that are $3\epsilon$ far apart from each other.
Hence a lower bound on the sample complexity of distinguishing states/unitaries from a packing net implies a lower bound for the learning task.

\paragraph{State learning}
For state learning, we construct a packing net $\mathcal{M}$ of the set of $(\log_2 G)$-qubit states, which we later tensor product with zero states on the remaining qubits.
These states have circuit complexity $\sim G$ because $\mathcal{O}(2^k)$ two-qubit gates can implement any pure $k$-qubit states \cite{shende2005synthesis}.
We prove that the cardinality of $\mathcal{M}$ can be lower bounded by $e^{\Omega(G)}$.
This means that to distinguish the states in $\mathcal{M}$, one has to gather $\Omega(\log|\mathcal{M}|)\geq \Omega(G)$ bits of information.
Meanwhile, Holevo's theorem \cite{holevo} asserts that the amount of information carried by each sample is upper bounded by $\tilde{\mathcal{O}}(\epsilon^2)$ \cite{haah2016sample}.
Hence, we need at least $\tilde{\Omega}(G/\epsilon^2)$ copies of the unknown state.

\paragraph{Unitary learning}
Similarly, for unitary learning, we construct a packing net by stacking all the gates into $\log_4 G$ qubits, using the fact that $\bigo(4^k)$ two-qubit gates suffice to implement any $k$-qubit unitaries~\cite{vartiainen2004efficient}.
Lacking an analogue of Holevo's theorem for unitary queries, we turn to a recently established bound on the success probability of unitary discrimination \cite{bavaresco2022unitary} and obtain an $\Omega(G)$ sample complexity lower bound for constant $\epsilon$.
To incorporate the $\epsilon$ dependence, we follow \cite{haah2023query} and map the problem into a fractional query problem.
We show that with $N$ queries, we can use the learning algorithm to simulate \cite{berry2015hamiltonian, cleve2009efficient} an $\bigo(\epsilon N)$ query algorithm that solves the above constant-accuracy distinguishing problem.
This gives us the desired $N\geq \Omega(G/\epsilon)$ lower bound.

{\renewcommand\addcontentsline[3]{} 
\subsection{Computational hardness}}

We prove the computational complexity lower bounds in \Cref{thm:state-comp-complex,thm:unitary-comp-complex} again by reduction to distinguishing tasks, whose hardness relies on cryptographic primitives in this case.
In particular, we show that if we can learn the state/unitary in polynomial time, then we can use this learning algorithm to efficiently distinguish between pseudorandom states/functions~\cite{ji2018pseudorandom,goldreich1986construct} and truly random states/functions.
We note that similar ideas have been used to establish a cryptographic no-cloning theorem \cite{ji2018pseudorandom} for PRS, but without gate complexity dependence and the unitary counterpart.
The \textsf{RingLWE} hardness assumption here may also be relaxed to the existence of appropriate quantum-secure PRS/PRF constructions that have the same gate complexity discussed below.

Our proofs rely on the construction of quantum-secure pseudorandom functions (PRFs) that can be implemented using $\mathsf{TC}^0$ circuits, subject to the assumption that Ring Learning with Errors (\textsf{RingLWE}) cannot be solved by a quantum computer in sub-exponential time~\cite{arunachalam2021quantum}.
We show that the circuit construction of ~\cite{arunachalam2021quantum} can be implemented quantumly using $G = \mathcal{O}(n\polylog(n))$ gates by converting this $\mathsf{TC}^0$ circuit into a quantum circuit that computes the same function.
With this construction, we can prove the computational hardness of learning when $G = \mathcal{O}(n\polylog(n))$ as follows.

\paragraph{State learning}
For state learning, we utilize these quantum-secure PRFs to construct pseudorandom quantum states (PRS), in particular binary phase states from~\cite{ji2018pseudorandom,brakerski2019pseudo}, with $G = \mathcal{O}(n\polylog(n))$ gates.
Given copies of some unknown quantum state that is promised to either be a PRS or a Haar-random state, we design a procedure that can distinguish these two cases.
The distinguisher uses our algorithm for learning states along with the SWAP test applied to the learned state and the given state~\cite{barenco1997stabilization,buhrman2001quantum}.
Thus, we show that if our learning algorithm was able to computationally efficiently learn PRS, then we would have an efficient distinguisher between PRS and Haar-random states, contradicting the definition of a PRS~\cite{ji2018pseudorandom}.

\paragraph{Unitary learning}
The proof idea in the unitary setting is similar.
In this case, we consider PRFs directly rather than the PRS construction.
Given query access to some unknown unitary that is promised to be the unitary oracle of either a PRF or a uniformly random Boolean function, we design a procedure that can distinguish these two cases. 
The distinguisher uses our algorithm for learning unitaries along with the SWAP test~\cite{barenco1997stabilization,buhrman2001quantum}.
Here, we query the given/learned unitaries on a random tensor product of single-qubit stabilizer states and conduct the SWAP test between the output states.
This way, we show that if our learning algorithm was able to computationally efficiently learn a unitary implementing a PRF, then we would have an efficient distinguisher between PRFs and uniformly random functions, which contradicts the definition of a PRF~\cite{goldreich1986construct}.

We then go one step further and show computational hardness for circuit size $\tilde{\mathcal{O}}(G)$. To do this we rely critically on the assumption that $\mathsf{RingLWE}$ is hard not just to polynomial-time quantum algorithms, but even to quantum algorithms that run for longer (sub-exponential) time. This allows us to take a much smaller input size to the PRS/PRF in our previous constructions (i.e., over $\mathcal{O}(G)$ qubits which can be implemented with $\tilde{\bigo}(G)$ gates).
The sub-exponential computational hardness of $\mathsf{RingLWE}$ then implies that solving the learning tasks requires time exponential in $G$. 

Meanwhile, for $G=\mathcal{O}(\log n)$, the learning tasks can be solved efficiently by junta learning and standard tomography methods.
This establishes $\log n$ circuit complexity as a transition point of computational efficiency.
This also implies that the circuit complexity of the PRS/PRF constructions in \cite{arunachalam2021quantum, brakerski2019pseudo} is optimal up to logarithmic factors, otherwise it would contradict efficient tomography of $\bigo(\log n)$-complexity states/unitaries.
Finally, we note that the PRS/PRF we consider can be implemented with a similar number of Clifford and T gates, extending our results to Clifford+T circuits.

{\renewcommand\addcontentsline[3]{} 
\section*{Acknowledgments}}
The authors thank John Preskill for valuable feedback and continuous support throughout this project.
The authors also thank Ryan Babbush, Yiyi Cai, Charles Cao, Alexandru Gheorghiu, András Gilyén, Alex B. Grilo, Jerry Huang, \ins{Minghao Liu}, Nadine Meister, Chris Pattison, Alexander Poremba, Xiao-Liang Qi, Ruohan Shen, Mehdi Soleimanifar, Yu Tong, Tzu-Chieh Wei, Tai-Hsuan Yang, Nengkun Yu for fruitful discussions.
Finally, the authors thank Antonio Anna Mele and the anonymous reviewers \del{at QIP 2024} for feedback on an earlier version of this paper.
HZ was supported by a Caltech Summer Undergraduate Research Fellowship (SURF).
LL was supported by a Mellon Mays Undergraduate Fellowship.
HH was supported by a Google PhD fellowship and a MediaTek Research Young Scholarship.
HH acknowledges the visiting associate position at the Massachusetts Institute of Technology.
This work is supported by a collaboration between the US DOE and other Agencies. YQ acknowledges financial support from the U.S. Department of Energy, Office of Science, National Quantum Information Science Research Centers, Quantum Systems Accelerator.
MCC was supported by a DAAD PRIME fellowship.
This work was done (in part) while a subset of the authors were visiting the Simons Institute for the Theory of Computing.
The Institute for Quantum Information and Matter is an NSF Physics Frontiers Center.

\vspace{2em}
{\renewcommand\addcontentsline[3]{} 
\section*{Data Availability}}
\edit{No data are generated or analyzed in this theoretical work}{The code that generates the data presented in the figures and that support the other findings of this study is available at \url{https://github.com/haimengzhao/bounded-gate-tomography}.}

\vspace{2em}
{\renewcommand\addcontentsline[3]{} 
\section*{Author Contributions}}
Y.Q., H.H., and M.C.C. conceived the project. 
H.Z., L.L., and I.K. led the development of the theory and the analytic calculations. 
\ins{H.Z., I.K., and H.H. designed and conducted the numerical experiments.}
All authors contributed to the mathematical aspects of this work.
H.Z. and L.L. wrote the manuscript with input from all authors.

\vspace{2em}
{\renewcommand\addcontentsline[3]{} 
\section*{Competing Interests}}
The authors declare no competing interests.

\newpage
\bibliographystyle{unsrt}
\let\oldaddcontentsline\addcontentsline
\renewcommand{\addcontentsline}[3]{}
\bibliography{refs}
\let\addcontentsline\oldaddcontentsline

\newpage
\appendix
\appendixpage

\tableofcontents

\section{Preliminaries}

Throughout the appendices, we use $d=2^n$ to denote the dimension of the $n$-qubit Hilbert space unless otherwise stated.

\subsection{Distance metrics}

Here we review some distance metrics and their properties used throughout our proofs.
In the main text we have already introduced the trace distance 
\begin{equation}
    \dtr(\ket{\psi}, \ket{\phi}) 
    =\frac{1}{2} \|\ketbra{\psi} - \ketbra{\phi}\|_1 \, ,
\end{equation}
which is analogously defined for density matrices as $\dtr(\rho, \sigma) = \norm{\rho - \sigma}_1/2$, the diamond distance
\begin{equation}
    \dworst(U, V) = \max_{\rho} \|(U\otimes I)\rho(U\otimes I)^\dagger - (V\otimes I)\rho(V\otimes I)^\dagger\|_1\, ,
\end{equation} 
and the root mean squared trace distance
\begin{equation}
    \davg(U, V) = \sqrt{\E_{\ket{\psi}}[\dtr(U\ket{\psi}, V\ket{\psi})^2]}
\end{equation}
where the expectation is taken over Haar measure\footnote{Due to the concentration of Lipschitz functions on inputs drawn from the Haar measure, controlling this root mean squared distance also leads to error bounds that hold with high probability over random input states.}.

Apart from these, we also use the following auxiliary distance metrics.
We define the quotient spectral distance 
\begin{equation}
    d_2'(U, V) 
    = \min_{e^{i\phi}\in U(1)}\|U-e^{i\phi}V\|
\end{equation}
to be the spectral distance $d_2(U, V) = \|U-V\|$ up to a global phase.
Similarly, we define the quotient normalized Frobenius distance 
\begin{equation}
    d_F'(U, V) = \min_{e^{i\phi}\in U(1)}\frac{1}{\sqrt{d}}\|U-e^{i\phi}V\|_F
\end{equation}
as the normalized Frobenius norm distance $d_F(U, V) = \tfrac{1}{\sqrt{d}}\|U-V\|_F$ up to a global phase.

The following lemma shows that (quotient) spectral distance and diamond distance are equivalent.
\begin{lemma}[Spectral and diamond distance of unitaries, variant of {\cite[Lemma B.5]{caro2022generalization}}]
    For any two $d$-dimensional unitaries $U$ and $V$, we have
    \begin{equation}
        \frac{1}{\sqrt{2}}d_2'(U, V)\leq \frac{1}{2}\dworst(U, V)\leq d_2'(U, V)\leq \|U-V\|.
    \end{equation}
\label{lem:dist-spectral-diamond}
\end{lemma}
\begin{proof}
    Since stabilization is not necessary for computing the diamond distance of two unitary channels~\cite{wildebook}, we have
    \begin{equation}
    \begin{split}
        &\frac{1}{2}\dworst(U, V) = \max_{\ket{\psi}} \frac{1}{2}\|U\ketbra{\psi}U^\dagger - V\ketbra{\psi}V^\dagger\|_1 
        = \max_{\ket{\psi}} \sqrt{1-|\braket{\psi}{U^\dagger V|\psi}|^2} \\
        &=\max_{\ket{\psi}} \sqrt{(1+|\braket{\psi}{U^\dagger V|\psi}|)(1-|\braket{\psi}{U^\dagger V|\psi}|)}
        \geq \max_{\ket{\psi}}\frac{1}{\sqrt{2}} \sqrt{2(1-|\braket{\psi}{U^\dagger V|\psi}|)} \\
        &=\frac{1}{\sqrt{2}}\min_{e^{i\phi}\in U(1)} \max_{\ket{\psi}} \|U\ket{\psi} - e^{i\phi}V\ket{\psi}\|_2 
        = \frac{1}{\sqrt{2}}\min_{e^{i\phi}\in U(1)}\|U-e^{i\phi}V\|
        = \frac{1}{\sqrt{2}}d_2'(U, V),
    \end{split}
    \end{equation}
    where we have used $\left|\braket{\psi}{U^\dagger V|\psi}\right|\geq 0$ and the standard conversion between trace distance and fidelity. 
    This proves the first inequality. 
    Similarly, we have
    \begin{equation}
    \begin{split}
        &\frac{1}{2}\dworst(U, V) = \max_{\ket{\psi}} \frac{1}{2}\|U\ketbra{\psi}U^\dagger - V\ketbra{\psi}V^\dagger\|_1 
        = \max_{\ket{\psi}} \sqrt{1-|\braket{\psi}{U^\dagger V|\psi}|^2} \\
        &=\max_{\ket{\psi}} \sqrt{(1+|\braket{\psi}{U^\dagger V|\psi}|)(1-|\braket{\psi}{U^\dagger V|\psi}|)}
        \leq \max_{\ket{\psi}} \sqrt{2(1-|\braket{\psi}{U^\dagger V|\psi}|)} \\
        &=\min_{e^{i\phi}\in U(1)} \max_{\ket{\psi}} \|U\ket{\psi} - e^{i\phi}V\ket{\psi}\|_2 
        = \min_{e^{i\phi}\in U(1)}\|U-e^{i\phi}V\|
        = d_2'(U, V),
    \end{split}
    \end{equation}
    where we have used $\left|\braket{\psi}{U^\dagger V|\psi}\right|\leq 1$, proving the second inequality. 
    The third inequality follows immediately from $d_2'(U, V)=\min_{e^{i\phi}\in U(1)}\|U-e^{i\phi}V\|\leq \|U-V\|$.
\end{proof}

We will also utilize the subadditivity of the diamond distance.

\begin{lemma}[Subadditivity of diamond distance {\cite[Prop.~3.48]{watrous2018book}}]
    \label{lem:subadd-diamond}
    For any $d$-dimensional unitaries $U_1, U_2, V_1, V_2$, we have the following inequality:
    \begin{equation}
        \dworst(U_2U_1, V_2V_1) \leq \dworst(U_2, V_2) + \dworst(U_1, V_1).
    \end{equation}
\end{lemma}

From the standard relationship between different $p$-norms, we have the following relation between $d_2'$ and $d_F'$.

\begin{lemma}[Norm conversion between quotient spectral and normalized Frobenius distance]
    For any two $d$-dimensional unitaries $U$ and $V$, we have
    \begin{equation}
        \frac{1}{\sqrt{d}}d_2'(U, V)\leq d_F'(U, V)\leq d_2'(U, V).
    \end{equation}
\label{lem:dist-norm-conversion}
\end{lemma}
\begin{proof}
    For any $e^{i\phi}\in U(1)$, the standard relation between matrix norms gives us
    \begin{equation}
        \|U-Ve^{i\phi}\| \leq \|U-Ve^{i\phi}\|_F \leq \sqrt{d}\|U-Ve^{i\phi}\|.
    \end{equation}
    Taking the minimum of $\|U-Ve^{i\phi}\|_F$ over $e^{i\phi}$ in the first inequality and dividing by $\sqrt{d}$, we obtain 
    \begin{equation}
        \frac{1}{\sqrt{d}}d_2'(U, V)\leq \frac{1}{\sqrt{d}}\|U-Ve^{i\phi}\| \leq d_F'(U, V).
    \end{equation}
    Similarly, taking the minimum of $\|U-Ve^{i\phi}\|$ over $e^{i\phi}$ in the second inequality and dividing by $\sqrt{d}$ yields \begin{equation}
        d_F'(U, V)\leq \frac{1}{\sqrt{d}}\|U-Ve^{i\phi}\|_F \leq d_2'(U, V).
    \end{equation}
    Thus we have the desired results.
\end{proof}

The following lemma collects some useful properties of $d_F'$ and in particular shows that $d_F'$ and $\davg$ are equivalent.

\begin{lemma}[Properties of quotient normalized Frobenius distance]
For any two $d$-dimensional unitaries $U$ and $V$, we have:
\begin{enumerate}
    \item $\frac{1}{2} d_F'(U, V)\leq \davg(U, V) \leq d_F'(U, V)$. 
    \item For any integer $p\geq 1$, $d_F'(U^p, V^p)\leq p d_F'(U, V)$.
    \item For any integer $p\geq 1$, if $d_F'(U, I), d_F'(V, I) \leq \frac{4/(25\pi)}{\sqrt{d}}$, then $d_F'(U^{1/p}, V^{1/p})\leq \frac{2}{p} d_F'(U, V)$.
\end{enumerate}
\label{lem:dist-df'}
\end{lemma}

Item 3 can be viewed as a version of \cite[Lemma 3.1]{haah2023query}.

\begin{proof}
    Item 1: From properties of the Haar integral (see e.g., \cite[Example 50]{mele2023introduction}), we have
    \begin{equation}
    \label{eq:davg}
        \davg(U, V)^2 = 1 - \frac{d+|\tr(U^\dagger V)|^2}{d(d+1)}.
    \end{equation}
    On the other hand, we have
    \begin{equation}
        d_F'^2(U, V) 
        = \min_{e^{i\phi}\in U(1)}\frac{1}{d}\|U-V e^{i\phi}\|_F^2 
        = \min_{e^{i\phi}\in U(1)}2 - \frac{2}{d}\mathrm{Re}[\tr(U^\dagger V e^{i\phi})] 
        = 2 - \frac{2}{d}|\tr(U^\dagger V)|.
    \end{equation}
    Combining them, we get
    \begin{equation}
        \davg(U, V)^2 = \frac{d}{d+1}d_F'^2(U, V)\left(1-\frac{d_F'^2(U, V)}{4}\right)\in \left[\frac{1}{4}d_F'^2(U, V), d_F'^2(U, V)\right],
    \end{equation}
    because $d_F'^2(U, V) \in [0, 2]$. Thus we have established Item 1.

    Item 2:
    From triangle inequality, we have
    \begin{equation}
        d_F'(U^p, V^p) \leq \sum_{k=1}^p d_F'(U^{p+1-k}V^{k-1}, U^{p-k}V^{k}) = \sum_{k=1}^p d_F'(U, V) = p d_F'(U, V)\, ,
    \end{equation}
    where we have used the unitary invariance of $d_F'$. This proves Item 2.

    Item 3:
    We first prove the following modified version without the global phase: ``If $d_F(U, I), d_F(V, I)\leq \frac{4/(5\pi)}{\sqrt{d}}$, then $d_F(U^{1/p}, V^{1/p})\leq \frac{2}{p}d_F(U, V)$.''
    Let $U=e^X, V=e^Y$ with $\|X\|, \|Y\|\leq \pi$. We can refine the bound on $\|X\|, \|Y\|$ by noting the following
    \begin{equation}
        \|X\|\leq \frac{\pi}{2}\|e^X-I\|\leq \frac{\pi}{2}\|U-I\|_F
        =\frac{\pi \sqrt{d}}{2}d_F(U, I)\leq \frac{2}{5},
    \end{equation}
    where the first inequality can be seen from eigenvalue analysis as follows: 
    Let $i\theta_k$ be the eigenvalues of $X$ with $|\theta_k|\leq \pi$. Then we have
    \begin{equation}
        \|X\|=\max_k |\theta_k| \leq \pi \max_k \left|\sin\frac{\theta_k}{2}\right| = \frac{\pi}{2} \max_k \left|e^{i\theta_k}-1\right| = \frac{\pi}{2}\|e^X-I\|.
    \end{equation}
    Similarly, we have $\|Y\|\leq 2/5$.

    Next, we prove the following inequality when $\|X\|, \|Y\|\leq 2/5$ (similar to~\cite[Appendix D]{barthel2018fundamental}):
    \begin{equation}\label{eq:frobenius-norm-generator-exponential-1}
        \frac{1}{2}\|X-Y\|_F \leq \|e^X-e^Y\|_F \leq \|X-Y\|_F.
    \end{equation}
    For the upper bound, we use the triangle inequality and a telescoping sum representation: For any $m\in\mathbb{N}$,
    \begin{equation}
        \|e^X-e^Y\|_F \leq \sum_{k=1}^m \|e^{(k-1)X/m}(e^{X/m}-e^{Y/m})e^{(m-k)Y/m}\|_F = m \|e^{X/m}-e^{Y/m}\|_F\, ,
    \end{equation}
    and by taking $m\to\infty$ we arrive at the upper bound.
    For the lower bound, note that by triangle inequality, we have
    \begin{equation}
        \|e^X-e^Y\|_F = \left\|\sum_{k=1}^{\infty} \frac{1}{k!}(X^k-Y^k)\right\|_F \geq \|X-Y\|_F - \left\|\sum_{k=2}^{\infty} \frac{1}{k!}(X^k-Y^k)\right\|_F.
    \end{equation}
    The second term can be upper bounded by
    \begin{equation}
    \begin{split}
        \left\|\sum_{k=2}^{\infty} \frac{1}{k!}(X^k-Y^k)\right\|_F &= \left\|\sum_{k=2}^{\infty} \sum_{l=1}^k \frac{1}{k!}X^{l-1}(X-Y)Y^{k-l}\right\|_F \\
        &\leq \sum_{k=2}^{\infty}\frac{k}{k!}\left(\frac{2}{5}\right)^k\|X-Y\|_F \\
        &= (e^{2/5}-1)\|X-Y\|_F\, ,
    \end{split}
    \end{equation}
    where we have used $\|AB\|_F\leq \|A\|\cdot \|B\|_F$ and $\|X\|, \|Y\|\leq 2/5$. 
    Plugging this bound back in, we arrive at the lower bound
    \begin{equation}\label{eq:frobenius-norm-generator-exponential-2}
        \|e^X-e^Y\|_F \geq (2-e^{2/5})\|X-Y\|_F\geq \frac{1}{2}\|X-Y\|_F.
    \end{equation}\Cref{eq:frobenius-norm-generator-exponential-1} in particular implies
    \begin{equation}
        d_F(U^{1/p}, V^{1/p})
        \leq\frac{1}{p\sqrt{d}}\|X-Y\|_F
        \leq\frac{2}{p}d_F(U, V)\, ,
    \end{equation}
    and thus the modified version of our claim.

    Finally, we deal with the global phase and prove the $d_F'$ version, where we assume $d_F'(U, I), d_F'(V, I)\leq\frac{4/(25\pi)}{\sqrt{d}}$. 
    Let $e^{i\phi_U}, e^{i\phi_V}, e^{i\phi}\in U(1)$ denote the global phases that minimize $d_F(U, Ie^{i\phi_U}), d_F(V, Ie^{i\phi_V})$ and $d_F(Ue^{-i\phi_U}, Ve^{-i\phi_V} e^{i\phi})$, respectively. 
    Then $d_F(U, Ie^{i\phi_U}), d_F(V, Ie^{i\phi_V})\leq \frac{4/(25\pi)}{\sqrt{d}}$ by assumption, and $d_F(Ue^{-i\phi_U}, Ve^{-i\phi_V})\leq d_F(U, Ie^{i\phi_U}) + d_F(V, Ie^{i\phi_V}) \leq \frac{8/(25\pi)}{\sqrt{d}}$. 
    Therefore,
    \begin{equation}
    \begin{split}
        d_F(e^{i\phi}, I) 
        &\leq d_F(e^{i\phi}, (Ve^{-i\phi_V})^\dagger (Ue^{-i\phi_U})) + d_F((Ve^{-i\phi_V})^\dagger (Ue^{-i\phi_U}), I) \\
        &= d_F(Ue^{-i\phi_U}, Ve^{-i\phi_V}e^{i\phi}) + d_F(Ue^{-i\phi_U}, Ve^{-i\phi_V}) 
        \\
        &\leq 2d_F(Ue^{-i\phi_U}, Ve^{-i\phi_V}) 
        \\
        &\leq \frac{16/(25\pi)}{\sqrt{d}}\, .
    \end{split}
    \end{equation}
    This means that $d_F(Ue^{-i\phi_U} e^{-i\phi}, I)\leq d_F(U, Ie^{i\phi_U}) + d_F(e^{i\phi}, I) \leq \frac{(4+16)/(25\pi)}{\sqrt{d}}=\frac{4/(5\pi)}{\sqrt{d}}.$ 
    We also know that $d_F(Ve^{-i\phi_V}, I)\leq \frac{4/(25\pi)}{\sqrt{d}}\leq \frac{4/(5\pi)}{\sqrt{d}}$. 
    Thus the two matrices $Ue^{-i\phi_U}e^{-i\phi}$ and $Ve^{-i\phi_V}$ satisfy the condition of the modified version without global phase, and we thus have
    \begin{equation}
    \begin{split}
        d_F'(U^{1/p}, V^{1/p})
        &\leq d_F'(U^{1/p}, V^{1/p}(e^{-i\phi_V})^{1/p}(e^{i\phi_U})^{1/p} (e^{i\phi})^{1/p})\\
        &=d_F((Ue^{-i\phi_U} e^{-i\phi})^{1/p}, (Ve^{-i\phi_V})^{1/p})
        \\
        &\leq\frac{2}{p}d_F(Ue^{-i\phi_U} e^{-i\phi}, Ve^{-i\phi_V})
        =d_F'(U, V). 
    \end{split}
    \end{equation}
    This concludes the proof of Item 3.
\end{proof}

Haar-random states are in general hard to generate. 
One may want to use other ensembles of input states and the associated distance metric for average-case learning.
A class of ensembles of physical interest is that of locally scrambled ensembles \cite{kuo2020markovian, hu2023classical} defined as follows:

\begin{definition}[Locally scrambled ensembles up to the second moment]
\label{def:local-scram-ensem}
An ensemble $\mathcal{S}$ of (i.e., a distribution over) $n$-qubit states is called a locally scrambled ensemble up to the second moment if it is of the form $\mathcal{S}=\mathcal{U}\ket{0}^{\otimes n}$, where $\mathcal{U}$ is an ensemble of unitaries that is locally scrambled up to the second moment. 
That is, there exists another unitary ensemble $\mathcal{U}'$, such that: (1) for any $U'$ randomly sampled from $\mathcal{U}'$ and for any tensor product of single-qubit unitaries $\otimes_{i=1}^n U_i$, $U'\otimes_{i=1}^n U_i$ follows the same distribution of $\mathcal{U}'$; and (2) for any $2n$-qubit density matrices $\rho$, we have $\mathbb{E}_{U\sim \mathcal{U}}[U^{\otimes 2}\rho (U^\dagger)^{\otimes 2}] = \mathbb{E}_{U'\sim \mathcal{U}'}[U'^{\otimes 2}\rho (U'^\dagger)^{\otimes 2}]$.  
We use $\mathbb{S}_\mathrm{LS}^{(2)}$ to denote the set of all such state ensembles.
\end{definition}

Notable examples of these ensembles include $n$-qubit Haar-random states, products of Haar-random single-qubit states, products of random single-qubit stabilizer states, $2$-designs on $n$-qubit states, and output states of random local quantum circuits with any fixed architecture.
The following lemma from the study of out-of-distribution generalization \cite{caro2022out} shows that these ensembles lead to mutually equivalent average-case distance metrics.

\begin{lemma}[Equivalence of locally scrambled average-case distances {\cite[Theorem 1]{caro2022out}}]
We denote by $d_P (U,V) = \sqrt{\E_{\ket{\psi}\sim P}[\dtr(U\ket{\psi}, V\ket{\psi})^2]}$ the root mean squared trace distance with respect to an ensemble $P$. For any $P, Q\in \mathbb{S}_\mathrm{LS}^{(2)}$ and for any unitaries $U, V$, we have
\begin{equation}
    \frac{1}{\sqrt{2}}d_Q(U, V)\leq d_P(U, V) \leq \sqrt{2}d_Q(U, V).
\end{equation}
\label{lem:equi_local_scram}
\end{lemma}

The following lemma shows that the triangle inequality holds for $d_P$ (and in particular, $\davg$).

\begin{lemma}[Triangle inequality for average-case distance]
\label{lem:tri-ineq-avg}
Let $d_P (U,V) = \sqrt{\E_{\ket{\psi}\sim P}[\dtr(U\ket{\psi}, V\ket{\psi})^2]}$ be the root mean squared trace distance with respect to an ensemble $P$.
For any three unitaries $U, V$ and $W$, we have the triangle inequality
\begin{equation}
    d_P(U, V)\leq d_P(U, W) + d_P(W, V).
\end{equation}
\end{lemma}
\begin{proof}
    Note that
    \begin{equation}
    \begin{split}
        &d_{P}^2(U, V) 
        = \mathbb{E}_{\ket{\psi}\sim P}[\dtr(U\ket{\psi}, V\ket{\psi})^2]
        \leq \mathbb{E}_{\ket{\psi}\sim P}[\left(\dtr(U\ket{\psi}, W\ket{\psi}) + \dtr(W\ket{\psi}, V\ket{\psi})\right)^2] \\
        &= d_P^2(U, W) +  d_P^2(W, V) + 2 \mathbb{E}_{\ket{\psi}\sim P}[\dtr(U\ket{\psi}, W\ket{\psi}) \dtr(W\ket{\psi}, V\ket{\psi}) ] \\
        &\leq d_P^2(U, W) +  d_P^2(W, V) + 2 \sqrt{\mathbb{E}_{\ket{\psi}\sim P}[\dtr(U\ket{\psi}, W\ket{\psi})^2]}\cdot \sqrt{\mathbb{E}_{\ket{\psi}\sim P}[\dtr(W\ket{\psi}, V\ket{\psi})^2]} \\
        &= \left( d_{P}(U, W) + d_{P}(W, V)\right)^2,
    \end{split}
    \end{equation}
    where we have used the triangle inequality for $\dtr$ and the Cauchy-Schwartz inequality.
    Taking the square root gives us the desired result.
\end{proof}

\subsection{Covering and packing nets}
Our results in state and unitary learning utilize a tool from high-dimensional probability theory, namely covering and packing nets.
We employ covering nets in our proofs of the sample complexity upper bounds and packing nets in our proofs of sample complexity lower bounds.
Intuitively, covering and packing nets characterize the complexity of a space by discretizing it with small balls of a given resolution.
We formally define these concepts below.

\begin{definition}[Covering net/number and metric entropy]
    Let $(X,d)$ be a metric space. Let $K \subseteq X$ be a subset and $\epsilon > 0$. Then, define the following.
    \begin{itemize}
        \item $N \subseteq K$ is an \emph{$\epsilon$-covering net} of $K$ if for any $x \in K$, there exists a $y \in N$ such that $d(x,y) \leq \epsilon$.
        \item The \emph{covering number} $\mathcal{N}(K, d, \epsilon)$ of $K$ is the smallest possible cardinality of an $\epsilon$-covering net of $K$.
        \item The \emph{metric entropy} is $\log \mathcal{N}(K,d,\epsilon)$.
    \end{itemize}
\end{definition}

We can similarly define a packing net.

\begin{definition}[Packing net/number]
    Let $(X,d)$ be a metric space. Let $K \subseteq X$ be a subset and $\epsilon > 0$. Then, define the following.
    \begin{itemize}
        \item $N \subseteq K$ is an \emph{$\epsilon$-packing net} of $K$ if for any $x,y \in N$, $d(x,y) > \epsilon$.
        \item The \emph{packing number} $\mathcal{M}(K, d, \epsilon)$ of $K$ is the largest possible cardinality of an $\epsilon$-packing net of $K$.
    \end{itemize}
\end{definition}

The following equivalence between covering and packing numbers is often useful.

\begin{lemma}[Covering and packing are equivalent, {\cite[Section 4.2]{vershynin2018high}}]
\label{lem:equi-cov-pack}
Let $(X, d)$ be a metric space. Let $K\subseteq X$ and $\epsilon>0$. We have
\begin{equation}
    \mathcal{N}(K, d, \epsilon/2) \geq \mathcal{M}(K, d, \epsilon) \geq \mathcal{N}(K, d, \epsilon).
\end{equation}
\end{lemma}

Covering numbers also have the following monotonicity property.

\begin{lemma}[Monotonicity of covering number, {\cite[Section 4.2]{vershynin2018high}}]
    Let $(K, d)$ be a metric space. If $L\subseteq K$, then $\mathcal{N}(L, d, \epsilon) \leq \mathcal{N}(K, d, \epsilon/2)$.
\end{lemma}

For our purposes, we need the following upper and lower bounds on the covering number of the unitary group.
Since the states that we consider can be generated by unitaries applied to a fixed input state, a covering number upper bound for unitaries with respect to the diamond distance implies a corresponding covering number upper bound for states with respect to the trace distance.

\begin{lemma}[Covering number of the unitary group, {\cite[Proposition 7]{szarek1982nets}}, {\cite[Lemma 1]{barthel2018fundamental}} and {\cite[Lemma C.1]{caro2022generalization}}]
\label{lem:2q-unitaries}
    Let $\norm{\cdot}'$ be any unitarily invariant norm. there exist universal constants $c_1, c_2>0$ such that for any $\epsilon\in(0, 2]$, the covering number of the $d$-dimensional unitary group $U(d)$ with respect to the norm $\norm{\cdot}'$ satisfies:
    \begin{equation}
        \left(\frac{c_1}{\epsilon}\right)^{d^2}
        \leq 
        \mathcal{N}(U(d), \norm{\cdot}', \norm{I}'\epsilon) 
        \leq 
        \left(\frac{c_2}{\epsilon}\right)^{d^2}.
    \end{equation}
    In particular, for the spectral norm $\norm{\cdot}$, we have the upper bound $\mathcal{N}(U(d), \norm{\cdot}, \epsilon) \leq \left(6/\epsilon\right)^{2d^2}$.
    For the Frobenius norm $\norm{\cdot}_F$, we have $\left(c_1/\epsilon\right)^{d^2} \leq \mathcal{N}(U(d), \norm{\cdot}_F, \sqrt{d}\epsilon) \leq \left(c_2/\epsilon\right)^{d^2}$.
\end{lemma}

We can use this result to bound the covering number for $n$-qubit unitaries consisting of $G$ two-qubit gates.

\begin{theorem}[Covering number of $G$-gate unitaries]
\label{thm:covering}
    Let $U^G \subseteq U(2^n)$ be the set of $n$-qubit unitaries that can be implemented by $G$ two-qubit gates. 
    Then for any $\epsilon \in (0, 1]$, there exist universal constants $c_1, c_2, C>0$ such that for $1\leq G/C\leq 4^{n+1}$, the metric entropy of $U^G$ with respect to the normalized Frobenius distance $d_F$ can be bounded as
    \begin{equation}
        \frac{G}{4C}\log(\frac{c_1}{\epsilon})\leq \log\mathcal{N}(U^G, d_F, \epsilon) \leq 16G\log(\frac{c_2G}{\epsilon}) + 2G\log n.
    \end{equation}
    Moreover, the metric entropy with respect to diamond distance $\dworst$ can be explicitly upper bounded by
    \begin{equation}
        \log\mathcal{N}(U^G, \dworst, \epsilon) \leq 32G\log(\frac{12G}{\epsilon}) + 2G\log n.
    \end{equation}
\end{theorem}
\begin{proof}
    The proof of the upper bounds is similar to the proof of Theorem C.1 in \cite{caro2022generalization}.
    We first prove the upper bound for diamond distance.
    
    Let $\epsilon\in (0, 1]$, and define $\epsilon' = \epsilon/2G$.
    Then by Lemma~\ref{lem:2q-unitaries}, there exists an $\epsilon'$-covering net $\tilde{\mathcal{N}}_{\epsilon'}$ of the set of two-qubit unitaries $U(2^2)$ with respect to the spectral norm $\norm{\cdot}$ of size 
    \begin{equation}
        |\tilde{\mathcal{N}}_{\epsilon'}| \leq \left(\frac{6}{\epsilon'}\right)^{32} = \left(\frac{12G}{\epsilon}\right)^{32}.
    \end{equation}
    This bound applies when the two-qubit unitary acts on a fixed set of two qubits.
    We can consider two-qubit unitaries that act on any of the $n$ qubits.
    Let $U^{2q} \subset U(2^n)$ denote this set of two-qubit unitaries that can act on any pair of the $n$ qubits of the system.
    Because there are $\binom{n}{2}$ pairs of qubits that the unitary could act on, the size of the covering net $\tilde{\mathcal{N}}_{\epsilon',n}$ of $U^{2q}$ is bounded by
    \begin{equation}
        |\tilde{\mathcal{N}}_{\epsilon',n}| \leq \binom{n}{2}\left(\frac{12G}{\epsilon}\right)^{32}.
    \end{equation}
    
    Recall that we want to find a covering net for the set $U^G$ of $n$-qubit unitaries consisting of $G$ two-qubit gates.
    Any unitary $U \in U^G$ can be written as $U_GU_{G-1}...U_1$ for $U_i \in U^{2q}$, where we suppress the tensor product with identity for readability.
    We consider the set of unitaries obtained by multiplying elements of the covering net $\tilde{\mathcal{N}}_{\epsilon',n}$ of $U^{2q}$. Namely, we define
    \begin{equation}
        \mathcal{N}_{\epsilon} \triangleq \{U_GU_{G-1}...U_1|U_i \in \tilde{\mathcal{N}}_{\epsilon',n}, 1\leq i \leq G\}\, .
    \end{equation}
    Let $U \in U^G$ be any arbitrary unitary that can be implemented by $G$ two-qubit gates, i.e., it can be written as $U = U_GU_{G-1}...U_1$ for $U_i \in U^{2q}$.
    As $\tilde{\mathcal{N}}_{\epsilon',n}$ is an $\epsilon'$-covering net of the set $U^{2q}$ of two-qubit unitaries, for each $U_i$ comprising the circuit $U$, we can find a $\Tilde{U}_i\in \Tilde{\mathcal{N}}_{\epsilon',n}$ such that $\lVert U_i - \Tilde{U}_i \rVert \leq \epsilon'$ for all $1 \leq i \leq G$, where $\norm{\cdot}$ denotes the spectral norm. Then, the unitary $\Tilde{U} \triangleq \Tilde{U}_G\Tilde{U}_{G-1}...\Tilde{U}_1 \in \mathcal{N}_{\epsilon_1}$ satisfies
    \begin{equation}
        \dworst(U, \tilde{U}) \leq \sum_{i=1}^G \dworst(U_i, \tilde{U}_i) \leq 2\sum_{i=1}^G \norm{U_i - \tilde{U}_i} \leq 2G \epsilon'=\epsilon,
    \end{equation}
    where we have employed the subadditivity of the diamond distance (Lemma~\ref{lem:subadd-diamond}) in the first inequality and then used the relationship between the diamond norm and spectral norm in the second inequality (Lemma~\ref{lem:dist-spectral-diamond}). In the last inequality, we used that $\norm{U_i - \tilde{U}_i} \leq \epsilon'$ and $\epsilon' = \nicefrac{\epsilon}{2G}$.

    Thus, $\mathcal{N}_{\epsilon}$ is an $\epsilon$-covering net of the set $U^G$ of $n$-qubit unitaries that can be implemented by $G$ two-qubit gates with respect to the diamond distance.
    By definition of $\mathcal{N}_{\epsilon}$, we have $|\mathcal{N}_{\epsilon}| = |\Tilde{\mathcal{N}}_{\epsilon',n}|^G$, since each unitary in the length $G$ strings of unitaries comprising elements of $\mathcal{N}_{\epsilon}$ are chosen from $\Tilde{\mathcal{N}}_{\epsilon',n}$. Then 
    \begin{equation}
        |\mathcal{N}_{\epsilon}|\leq \binom{n}{2}^G\left(\frac{12G}{\epsilon}\right)^{32G} \leq n^{2G}\left(\frac{12G}{\epsilon}\right)^{32G}.
    \end{equation}
    Taking the logarithm gives the desired result for diamond distance.

    We can argue similarly for the normalized Frobenius distance $d_F$.
    Specifically, we make use of the subadditivity of $\|\cdot\|_F$: $\forall U_1, V_1, U_2, V_2\in U(2^n)$, we have
    \begin{equation}
        \|U_2U_1 - V_2V_1\|_F \leq \|U_2U_1 - U_2V_1\|_F + \|U_2V_1 - V_2V_1\|_F = \|U_1 - V_1\|_F + \|U_2 - V_2\|_F,
    \end{equation}
    where we have used triangle inequality and $\|\cdot\|_F$ being unitary invariant.

    Consider any $U\in U^G, U = U_G\cdots U_1$, where $U_i, 1\leq i\leq T$ are 2-qubit unitaries acting on some pair of qubits. 
    Take $\epsilon'=\epsilon/G$ and let $\mathcal{N}_{\epsilon'}$ be an $\epsilon'$-covering net of $U(2^2)$ with respect to $\|\cdot\|_F$. Then there exist $V_i\in \mathcal{N}, 1\leq i\leq G$, such that $\|U_i - V_i\|\leq \epsilon/G$ when the $V_i$ are placed on the corresponding qubits. Let $V=V_G\cdots V_1$. By sub-additivity, we have
    \begin{equation}
        \|U-V\|_F \leq \sum_{i=1}^G \sqrt{2^{n-2}}\|U_i - V_i\|_F \leq \sqrt{2^{n-2}}G\epsilon' = \sqrt{2^{n-2}}\epsilon,
    \end{equation}
    where we have used the facts that the Frobenius norm is multiplicative w.r.t.~tensor products and that an ($n-2$)-qubit identity has Frobenius norm equal to $\sqrt{2^{n-2}}$.
    Therefore, the set of $V=V_G\cdots V_1$, where $V_i\in U(2^2)$ and acting on all possible pair of qubits is a $(2^{n-2}\epsilon)$-covering net of $U^G$. Since the number of choices for qubits to act on is $\binom{n}{2}$ for each $V_i$, we have
    \begin{equation}
        \mathcal{N}(U^G, \|\cdot\|_F, \sqrt{2^{n-2}}\epsilon) \leq \left[\binom{n}{2} \mathcal{N}(U(2^2), \|\cdot\|_F, \epsilon/G)\right]^G \leq n^{2G}\left(\frac{c_2 G \sqrt{2^2}}{\epsilon}\right)^{16G},
    \end{equation}
    where we have used \Cref{lem:2q-unitaries}. Redefining $\epsilon$ to be $\epsilon/\sqrt{2^2}$ and switching to the normalized $d_F$, we obtain
    \begin{equation}
        \log\mathcal{N}(U^G, d_F, \epsilon) \leq 16G\log \left(\frac{c_2G}{\epsilon}\right) + 2G\log n.
    \end{equation}

    Finally, we prove the lower bound.
    For this, we consider a particular set of circuit structures where all the $G$ gates are placed on the first $k\leq n$ qubits. 
    The set of unitaries that can be implemented by such circuits is denoted by $U^{\leq k}_G \subseteq U^G$. From the theory of universal quantum gates (see \cite{vartiainen2004efficient}), we know that to implement an arbitrary $k$-qubit unitary, we only need $G_k=\bigo(4^k)$ two-qubit gates that can implement single-qubit gates and CNOT. 
    That is, there exists a universal constant $C>0$, such that $C4^k\geq G_k$. 
    Therefore, for any integer $k\leq n$ satisfying $C4^k \leq G$, we have $G \geq G_k$. 
    Then all possible $k$-qubit unitaries can be implemented with these $G$ gates: $U^n(2^k)=\{U\otimes I_{2^{n-k}}: U\in U(2^k)\} \subseteq U^{\leq k}_G \subseteq U^G$, where $U^n(2^k)$ denotes the set obtained by embedding the $k$-qubit unitaries into the $n$-qubit unitaries via tensor-multiplication with the identity. 
    Thus $\mathcal{N}(U^G, \|\cdot\|_F, \epsilon) \geq \mathcal{N}(U^n(2^k), \|\cdot\|_F, 2\epsilon)$ by monotonicity.

    Next, we prove that $\mathcal{N}(U^n(2^k), \|\cdot\|_F, 2\epsilon)\geq \mathcal{N}(U(2^k), \|\cdot\|_F, 2\epsilon/\sqrt{2^{n-k})}$. 
    To do this, we take a minimal $2\epsilon$-covering net $\mathcal{N}$ of $U^n(2^k)$ with $|\mathcal{N}|=\mathcal{N}(U^n(2^k), \|\cdot\|_F, 2\epsilon)$. 
    Hence $\forall U\in U(2^k)$, $U\otimes I_{2^{n-k}}\in U^n(2^k), \exists V\otimes I_{2^{n-k}} \in \mathcal{N}$, such that $\|U-V\|_F =  \|U\otimes I_{2^{n-k}}-V\otimes I_{2^{n-k}}\|_F /\sqrt{2^{n-k}} \leq {2\epsilon}/\sqrt{2^{n-k}}$. 
    Therefore, $\{V: V\otimes I_{2^{n-k}}\in \mathcal{N}\}$ forms a $2\epsilon/\sqrt{2^{n-k}}$-covering net of $U(2^k)$, and we have $\mathcal{N}(U^n(2^k), \|\cdot\|_F, 2\epsilon)\geq\mathcal{N}(U(2^k), \|\cdot\|_F, 2\epsilon/\sqrt{2^{n-k}})$.

    Combining the above inequalities, we have
    \begin{equation}
        \log\mathcal{N}(U^G, \|\cdot\|_F, \epsilon) \geq \log\mathcal{N}(U^n(2^k), \|\cdot\|_F, 2\epsilon) \geq \log\mathcal{N}(U(2^k), \|\cdot\|_F, 2\epsilon/\sqrt{2^{n-k}}) \geq 2^{2k}\log\frac{c_1 \sqrt{2^{n}}}{2\epsilon},
    \end{equation}
    where the last inequalities follow from \Cref{lem:2q-unitaries}. 
    The largest possible $k$ is given by $k=\floor{\log_4(G/C)}\geq \log_4{G/(4C)}$. Thus, by redefining $\epsilon$ to be $\epsilon/\sqrt{2^n}$ and switching to $d_F$, we arrive at
    \begin{equation}
        \log\mathcal{N}(U^G, d_F, \epsilon) \geq \frac{G}{4C}\log\frac{c_1}{2\epsilon}.
    \end{equation}
    This completes the proof of \Cref{thm:covering}.
\end{proof}

The $d_F$ covering number bounds in \Cref{thm:covering} do not yet properly take into account the global $U(1)$ phase.
To obtain covering number for the average-case distance $\davg$, which is equivalent to the quotient normalized Frobenius distance $d_F'$ (\Cref{lem:dist-df'} Item 1), we need to quotient out the global phase.
This is formalized in the following lemma.

\begin{lemma}[Packing number of quotient distance metric, variant of {\cite[Lemma 4]{barthel2018fundamental}}]
\label{lem:quotient-packing}
    For any $d$-dimensional unitaries $U$ and $V$, let $d_F(U, V)=\norm{U-V}_F/\sqrt{d}$ be the normalized Frobenius distance, and $d_F'(U, V) = \min_{W\in U(1)}d_F(U, VW)$ be the corresponding quotient distance. Then there exists a universal constant $c_2>0$ such that the packing number of any set $\mathcal{U}\subseteq U(d)$ with respect to $d_F$ and $d_F'$ satisfies
    \begin{equation}
        \log\mathcal{M}(\mathcal{U}, d_F, 4\epsilon) - \log(c_2/\epsilon) 
        \leq 
        \log\mathcal{M}(\mathcal{U}, d_F', \epsilon) 
        \leq 
        \log\mathcal{M}(\mathcal{U}, d_F, \epsilon).
    \end{equation}
\end{lemma}
\begin{proof}
    We focus on the lower bound first.
    Take a minimal $\epsilon$-covering $\mathcal{N}_1$ of $\mathcal{U}$ with respect to $d_F'$ and a minimal $\epsilon$-covering $\mathcal{N}_2$ of $U(1)$ with respect to the absolute value distance $d_A(e^{i\phi}, e^{-i\phi'})=|e^{i\phi}-e^{-i\phi'}|$. Then, for any $U\in\mathcal{U}$, there exists $V\in \mathcal{N}_1$ such that $d_F'(U, V)\leq \epsilon$. Let $e^{i\phi^\star} = \argmin_{e^{i\phi}} d_F(U, Ve^{i\phi})$. Then $d_F(U, Ve^{i\phi^\star})\leq \epsilon$, and there exists $e^{i\phi'} \in \mathcal{N}_2$ such that $d_A(e^{i\phi^\star}, e^{i\phi'})\leq \epsilon$. Therefore,
    \begin{equation}
        d_F(U, V e^{i\phi'}) \leq d_F(U, Ve^{i\phi^\star}) + d_F(Ve^{i\phi^\star}, Ve^{i\phi'}) = d_F(U, Ve^{i\phi^\star}) + d_F(Ie^{i\phi^\star}, Ie^{i\phi'}) \leq 2\epsilon,
    \end{equation}
    where we have used the triangle inequality, $d_F$ being unitary invariant, and $d_F(Ie^{i\phi^\star}, Ie^{i\phi'}) = \frac{1}{\sqrt{d}}\|Ie^{i\phi^\star}-Ie^{i\phi'}\|_F=\frac{\|I\|_F}{\sqrt{d}}|e^{i\phi^\star}-e^{i\phi'}|=d_A(e^{i\phi^\star}, e^{i\phi'})\leq \epsilon$.
    Hence, the set $\{Ve^{i\phi'}: V\in \mathcal{N}_1, e^{i\phi'}\in \mathcal{N}_2\}$ is a $(2\epsilon)$-covering of $\mathcal{U}$ with respect to $d_F$. Then 
    \begin{equation}
        \mathcal{N}(\mathcal{U}, d_F, 2\epsilon) \leq \mathcal{N}(\mathcal{U}, d_F', \epsilon)\mathcal{N}(U(1), d_A, \epsilon).
    \end{equation}
    Therefore, using the equivalence of covering and packing (\Cref{lem:equi-cov-pack}) and the covering number bound for $U(1)$ (\Cref{lem:2q-unitaries}), we arrive at
    \begin{equation}
        \log \mathcal{M}(\mathcal{U}, d_F', \epsilon) \geq \log \mathcal{M}(\mathcal{U}, d_F', 4\epsilon) - \log(c_2/\epsilon).
    \end{equation}

    For the upper bound, note that $\forall U, V \in \mathcal{U}$, we have
    \begin{equation}
        d_F'(U, V) = \min_{e^{i\phi}\in U(1)}d_F(U, Ve^{i\phi}) \leq d_F(U, V).
    \end{equation}
    Therefore, a maximal $\epsilon$-packing net with respect to $d_F'$ is an $\epsilon$-packing net with respect to $d_F$. Therefore,
    \begin{equation}
        \mathcal{M}(\mathcal{U}, d_F', \epsilon) \leq \mathcal{M}(\mathcal{U}, d_F, \epsilon).
    \end{equation}
    This concludes the proof of \Cref{lem:quotient-packing}.
\end{proof}

With \Cref{lem:quotient-packing}, we can obtain the covering number of $G$-gate unitaries with respect to the average-case distance.

\begin{corollary}[Covering number with average-case distance]
\label{cor:covering-avg}
    Let $U^G \subseteq U(2^n)$ be the set of $n$-qubit unitaries that can be implemented by $G$ two-qubit gates. 
    Then for any $\epsilon \in (0, 1]$, there exist universal constants $c_1, c_2, C>0$ such that for $1\leq G/C\leq 4^{n+1}$, the metric entropy of $U^G$ with respect to the average-case distance $\davg(U, V) = \sqrt{\E_{\ket{\psi}}[\dtr(U\ket{\psi}, V\ket{\psi})^2]}$, where the expectation value is over Haar measure, can be bounded as
    \begin{equation}
        \frac{G}{4C}\log(\frac{c_1}{8\epsilon}) - \log(\frac{c_2}{2\epsilon})\leq \log\mathcal{N}(U^G, \davg, \epsilon) \leq 16G\log(\frac{c_2G}{\epsilon}) + 2G\log n.
    \end{equation}
\end{corollary}
\begin{proof}
    The corollary follows directly from \Cref{thm:covering}, \Cref{lem:quotient-packing}, and the equivalence of $d_F'$ and $\davg$ (\Cref{lem:dist-df'} Item 1).
\end{proof}

\subsection{Classical shadows and hypothesis selection}
\label{app:hypothesis}

Our proofs of the sample complexity upper bounds crucially rely on a known algorithm for quantum hypothesis selection~\cite{badescu2020improved}.
The high-level idea is to find a covering net over all unitaries consisting of only $G$ two-qubit gates and to then use quantum hypothesis selection to identify a candidate in the covering net close to the unknown target state/unitary.
A similar idea has previously appeared in~\cite{huang2022quantum}.
In this section, we discuss the quantum hypothesis selection algorithm from~\cite{badescu2020improved} and prove a performance guarantee when basing it on classical shadow tomography~\cite{huang2020predicting}.

The quantum hypothesis selection algorithm takes as input (classical descriptions of) a set of hypothesis states $\sigma_1,\dots,\sigma_m$ and quantum copies of an unknown state $\rho$.
Using these copies, the algorithm identifies a hypothesis state $\sigma_k$ that is close to the unknown state $\rho$ in trace distance.
Importantly, quantum hypothesis selection black-box reduces to shadow tomography~\cite{aaronson2018shadow}, i.e., one can use the shadow tomography protocol as a black-box to solve quantum hypothesis selection.
To obtain a better sample complexity scaling, we instead utilize classical shadow tomography~\cite{huang2020predicting}.

Recall that a classical shadow is a succinct classical description of a quantum state that allows us to predict many expectation values accurately.
One can construct this classical shadow description by applying a random unitary to the quantum state and measuring in the computational basis.
The most prevalent examples are random Clifford measurements, where the random unitary is chosen to be a random Clifford circuit, or random Pauli measurements, where the random unitary is chosen to be a tensor product of random Pauli gates.
Moreover, we have the following rigorous guarantee for using classical shadows to predict expectation values.

\begin{theorem}[Theorem 1 in~\cite{huang2020predicting}]
    \label{thm:classical-shadow}
    Let $O_1,\dots,O_M$ be Hermitian $2^n \times 2^n$ matrices, and let $\epsilon, \delta \in [0,1]$. Then,
    \begin{equation}
        N = \mathcal{O}\left(\frac{\log(M/\delta)}{\epsilon^2}\max_{1\leq i\leq M} \norm{O_i - \frac{\tr(O_i)}{2^n}\mathbb{I}}^2_{\mathrm{shadow}}\right)
    \end{equation}
    copies of an unknown quantum state $\rho$ suffice to predict $\hat{o}_i$ such that
    \begin{equation}
        |\hat{o}_i - \tr(O_i \rho)| \leq \epsilon
    \end{equation}
    for all $1 \leq i \leq M$, with probability at least $1-\delta$.
\end{theorem}
Here, $\norm{\cdot}_{\mathrm{shadow}}$ denotes the shadow norm, which depends on the ensemble of unitary transformations used to create the classical shadow.
For instance, in the case of random Cliffords, the shadow norm can be controlled via the (unnormalized) Frobenius norm, compare \cite[Proposition S1]{huang2020predicting}.

Now, we can prove a new guarantee for the quantum hypothesis selection by replacing shadow tomography with classical shadow in the proof in~\cite{badescu2020improved}.

\begin{prop}[Proposition 5.3 in \cite{badescu2020improved}; Classical Shadow Version]
\label{prop:hypothesis}
    Let $0 < \epsilon, \delta < 1/2$. Given access to unentangled copies of a pure quantum state $\rho$ and classical descriptions of $m$ fixed pure hypothesis states $\sigma_1,\dots, \sigma_m$, there exists a quantum algorithm that selects $\sigma_k$ such that $\dtr(\rho, \sigma_k) \leq 3\eta + \epsilon$ with probability at least $1-\delta$, where $\eta = \min_i \dtr(\rho,  \sigma_i)$. Moreover, this algorithm uses
    \begin{equation}
        N = \mathcal{O}\left(\frac{\log (m/\delta)}{\epsilon^2}\right)
    \end{equation}
    copies of the quantum state $\rho$.
\end{prop}

In~\cite{badescu2020improved}, they prove the guarantee on the quantum hypothesis selection algorithm using Helstrom's Theorem.
We follow a similar proof.
Thus, we first state Helstrom's Theorem and recall a corollary of it, which will be useful in the proof of Prop.~\ref{prop:hypothesis}.

\begin{theorem}[Helstrom's Theorem~\cite{Helstrom1969}]
\label{theorem:helstrom}
    Consider two $d$-dimensional quantum states $\rho$ and $\sigma$. Then, the trace distance between $\rho$ and $\sigma$ can be written as
    \begin{equation}
        \frac{1}{2}\norm{\rho - \sigma}_1 = \max_{\norm{O}_\infty \leq 1} |\tr(O\rho) - \tr(O\sigma)|,
    \end{equation}
    where the maximum is taken over all observables $O \in \mathbb{C}^{d\times d}$.
\end{theorem}

\begin{corollary}
    \label{coro:helstrom}
    Consider two $d$-dimensional quantum states $\rho$ and $\sigma$. Then, there exists an observable $A$ achieving the maximum such that
    \begin{equation}
        \tr(A\rho) - \tr(A\sigma) = \frac{1}{2}\norm{\rho - \sigma}_1.
    \end{equation}
    
\end{corollary}

\begin{proof}[Proof of Corollary~\ref{coro:helstrom}]
We will construct an observable $A$ that maximizes $\tr(O(\rho - \sigma))$ over all observables $O$ with $\norm{O}_\infty \leq 1$. Choose a representation of $\rho - \sigma$ in terms of eigenstates $\ket{v}$. Suppose the eigenvalues are discrete:
\begin{equation}
    (\rho - \sigma)\ket{v} = \lambda_v \ket{v}.
\end{equation}
Then, we can write the quantity we wish to maximize as
\begin{equation}
    \tr(O(\rho - \sigma)) = \sum_v \lambda_v \expval{O}{v}.
\end{equation}
We can maximize this by choosing $A$ such that
\begin{equation}
    \expval{A}{v} = \begin{cases}
        1 & \text{if } \lambda_v > 0\\
        0 & \text{if } \lambda_v \leq 0.
    \end{cases}
\end{equation}
In this way, we can write $A$ as a sum of projectors
\begin{equation}
    A = \sum_{v : \lambda_v > 0} \ketbra{v}.
\end{equation}
This maximizes $\tr(O(\rho-\sigma))$, so the corollary has been proven.
\end{proof}

\noindent With this, we can now prove Prop.~\ref{prop:hypothesis}.

\begin{proof}[Proof of Prop.~\ref{prop:hypothesis}]
The proof of Proposition 5.3 in~\cite{badescu2020improved} uses shadow tomography as a black box.
We follow the same strategy but use classical shadow tomography~\cite{huang2020predicting} instead of shadow tomography. Recall that in~\cite{badescu2020improved}, they run the shadow tomography algorithm from~\cite{aaronson2018shadow} with observables given by Helstrom's Theorem~\cite{Helstrom1969}.
This is the key step that uses samples of the unknown quantum state $\rho$, so we need to analyze it when using classical shadow instead of shadow tomography.
In our setting, Corollary~\ref{coro:helstrom} states that for any $i \neq j$, there exists an observable $A_{ij}$ such that
\begin{equation}
    \tr(A_{ij}\sigma_i) - \tr(A_{ij}\sigma_j) = \frac{1}{2}\norm{\sigma_i - \sigma_j}_1.
\end{equation}
Thus, the algorithm in~\cite{badescu2020improved} uses $M = \binom{m}{2} = \mathcal{O}(m^2)$ observables $\{A_{ij}\}$ to select the hypothesis state, where $m$ is the size of the hypothesis set. Using classical shadow instead of shadow tomography requires
\begin{equation}
    N = \tilde{\mathcal{O}}\left(\frac{\log(M/\delta)}{\epsilon^2}\max_{i,j}\norm{A_{ij} - \frac{\tr(A_{ij})}{2^n}\mathbb{I}}_{\mathrm{shadow}}^2\right)
\end{equation}
copies of $\rho$ by Theorem~\ref{thm:classical-shadow}, where $M$ is the number of observables $A_{ij}$ that we want to predict. Here, $M = \mathcal{O}(m^2)$ so that we require
\begin{equation}
    N = \tilde{\mathcal{O}}\left(\frac{\log(m/\delta)}{\epsilon^2}\max_{i,j}\norm{A_{ij} - \frac{\tr(A_{ij})}{2^n}\mathbb{I}}_{\mathrm{shadow}}^2\right)
\end{equation}
copies of $\rho$.
We claim that
\begin{equation}
    \max_{i,j}\norm{A_{ij} - \frac{\tr(A_{ij})}{2^n}\mathbb{I}}_{\mathrm{shadow}}^2 = \mathcal{O}(1).
\end{equation}
The lemma then follows from this claim.
We can prove this bound on the shadow norm using the construction of the observables $A_{ij}$ from Helstrom's Theorem, as seen in Corollary~\ref{coro:helstrom}. In our case, the states $\sigma_i$ are pure, and hence of rank $1$. Thus, the rank of $\sigma_i - \sigma_j$ is at most $2$ so that $A_{ij}$ is a projector of rank at most $2$.
Thus, the Frobenius norm of every $A_{ij}$ is $\mathcal{O}(1)$ and, by \cite[Proposition S1]{huang2020predicting}, the same holds for the shadow norm of the centered version of $A_{ij}$.
\end{proof}

\subsection{Characterizing the complexity of function classes}

In the proof of \Cref{thm:approx-functions} (in \Cref{app:learning-function}), we will need to characterize the complexity of certain function classes. 
The following definitions will be useful.
Throughout the work, we use $\mathcal{Y}^\mathcal{X}$ to denote the set of functions from $\mathcal{X}$ to $\mathcal{Y}$.

\begin{definition}[Growth function, {\cite{vapnik1971uniform}}]
Let $\mathcal{F}\subseteq \mathcal{Y}^\mathcal{X}$ be a class of functions with finite target space $\mathcal{Y}$.
For every subset $\Xi \subseteq X$, define the restriction of $\mathcal{F}$ to $\Xi$ as $\mathcal{F}|_\Xi = \{f\in \mathcal{Y}^\Xi: \exists F\in \mathcal{F}, \forall x\in \Xi, f(x)=F(x)\}$.
We define the growth function $\Gamma$ of $\mathcal{F}$ as
\begin{equation}
    \Gamma(\mu) = \max_{\Xi\subseteq \mathcal{X}: |\Xi|\leq \mu} |\mathcal{F}|_\Xi|
\end{equation}
for any $\mu\in\mathbb{N}$.
\end{definition}

The growth function characterizes the size of $\mathcal{F}$ when restricted to a domain of $\mu$ points. 
With growth function, we can define the VC dimension that characterizes the complexity of binary functions.

\begin{definition}[VC dimension, {\cite{vapnik1971uniform}}]
    The Vapnik-Chervonenkis (VC) dimension of a function class $\mathcal{F}\subseteq \{0, 1\}^\mathcal{X}$ is defined as
    \begin{equation}
        \mathrm{VCdim}(\mathcal{F}) = \max\{\mu\in \mathbb{N}: \Gamma(\mu)=2^\mu\},
    \end{equation}
    or $\infty$ if the maximum does not exist. Here $\Gamma(\mu)$ is the growth function of $\mathcal{F}$.
    Or equivalently, $\mathrm{VCdim}(\mathcal{F})$ is the largest $D\in \mathbb{N}\cup\{\infty\}$ such that there exists a set of points $\{x_i\}_{i=1}^D\subseteq \mathcal{X}$ that for all $C\subseteq[D]$, there is a function $f\in \mathcal{F}$ satisfying
    \begin{equation}
        f(x_i)=1 \iff i\in C.
    \end{equation}
    These points are said to be shattered by $\mathcal{F}$.
\end{definition}

To go beyond binary functions, we can use pseudo dimension defined below.

\begin{definition}[Pseudo dimension, {\cite{pollard1984convergence}}]
\label{def:pdim}
    The pseudo dimension of a real-valued function class $\mathcal{F}\subseteq \mathbb{R}^\mathcal{X}$ is defined as
    \begin{equation}
        \mathrm{Pdim}(\mathcal{F}) = \mathrm{VCdim}(\{\mathcal{X}\times \mathbb{R}\ni (x, y) \to \mathrm{sgn}[f(x)-y]: f\in\mathcal{F}\}).
    \end{equation}
    Or equivalently, $\mathrm{Pdim}(\mathcal{F})$ is the largest $D\in \mathbb{N}\cup\{\infty\}$ such that there exists a set of points $\{(x_i, y_i)\}_{i=1}^D\subseteq \mathcal{X}\times\mathbb{R}$ that for all $C\subseteq[D]$, there is a function $f\in \mathcal{F}$ satisfying
    \begin{equation}
        f(x_i)\geq y_i \iff i\in C.
    \end{equation}
    These points are said to be pseudo-shattered by $\mathcal{F}$.
\end{definition}

We will also use the fat-shattering dimension, a scale-sensitive variant of the pseudo-dimension.

\begin{definition}[Fat-shattering dimension, {\cite{kearns1994efficient}}]
\label{def:fat}
    Let $\alpha>0$. The $\alpha$-fat-shattering dimension $\mathrm{fat}(\mathcal{F}, \alpha)$ of a real-valued function class $\mathcal{F}\subseteq \mathbb{R}^\mathcal{X}$ is defined as the largest $D\in \mathbb{N}\cup\{\infty\}$ such that there exists a set of points $\{(x_i, y_i)\}_{i=1}^D\subseteq \mathcal{X}\times\mathbb{R}$ that for all $C\subseteq[D]$, there is a function $f\in \mathcal{F}$ satisfying
    \begin{equation}
    \begin{split}
        f(x_i)&\geq y_i + \alpha \quad \mathrm{if} \quad i\in C, \\
        f(x_i)&\leq y_i - \alpha \quad \mathrm{if} \quad i\notin C.
    \end{split}
    \end{equation}
    Such a set of points is said to be $\alpha$-fat-shattered by $\mathcal{F}$.
\end{definition}

\subsection{Cryptography}

Our computational complexity lower bounds rely on cryptographic primitives such as pseudorandom functions~\cite{goldreich1986construct} and pseudorandom quantum states~\cite{ji2018pseudorandom,brakerski2019pseudo}.
A family of pseudorandom functions is a set of functions such that sampling from this family is indistinguishable from a uniformly random function.
We present the formal definition below, following the presentation in~\cite{arunachalam2021quantum}.

\begin{definition}[Pseudorandom functions (PRFs)~\cite{goldreich1986construct}]
    \label{def:PRF}
    Let $\lambda$ denote the security parameter.
    Let $\mathcal{K} = \{\mathcal{K}_\lambda\}_{\lambda\in \mathbb{N}}$ be an efficiently sampleable key space.
    Let $\mathcal{X} = \{\mathcal{X}_\lambda\}_{\lambda\in \mathbb{N}}, \{\mathcal{Y}_\lambda\}_{\lambda\in\mathbb{N}}$ be collections of finite sets.
    Let $\mathcal{F} = \{f_\lambda\}_{\lambda\in\mathbb{N}}$ be a family of efficiently-computable keyed functions $f_\lambda: \mathcal{K}_\lambda \times \mathcal{X}_\lambda \to \mathcal{Y}_\lambda$.
    $\mathcal{F}$ is a \emph{pseudorandom function} if for every polynomial-time probabilistic algorithm $\mathsf{Adv}$, there exists a negligible function $\mathsf{negl}(\cdot)$ such that for every security parameter $\lambda\in \mathbb{N}$
    \begin{equation}
        \left|\Pr_{\mathbf{k} \leftarrow \mathcal{K}_\lambda}[\mathsf{Adv}^{f_\lambda (\mathbf{k}, \cdot)}(\cdot) = 1] - \Pr_{g \in \mathcal{U}_\lambda}[\mathsf{Adv}^{g}(\cdot) = 1]\right| \leq \mathsf{negl}(\lambda),
    \end{equation}
    where the key $\mathbf{k}$ is picked uniformly at random from the key space $\mathcal{K}_\lambda$ and $g$ is picked uniformly at random from $\mathcal{U}_\lambda$, the set of all functions from $\mathcal{X}_\lambda$ to $\mathcal{Y}_\lambda$.
    Here, $\mathsf{negl}(\lambda)$ denotes a negligible function, i.e., a function that grows more slowly than any inverse polynomial in $\lambda$.
\end{definition}

Concretely, it is common to take the input and output spaces to be $\mathcal{X}_\lambda = \{0,1\}^{m}$ and $\mathcal{Y}_\lambda = \{0,1\}$ for some input length $m = m(\lambda)$ that depends on the security parameter $\lambda$.
We consider this setting throughout the work.

\begin{definition}[Quantum secure PRFs~\cite{arunachalam2021quantum}]
    Let $\lambda$ denote the security parameter.
    A pseudorandom function is \emph{quantum secure against $t(\lambda)$ adversaries} if it satisfies \Cref{def:PRF} where $\mathsf{Adv}$ is a $t(\lambda)$-time quantum algorithm with quantum query access to $f_{\mathbf{k}}$ and $g$.
    When $t(\lambda)=\poly(\lambda)$, we say that the PRF is \emph{quantum secure}.
\end{definition}

There are several constructions for implementing PRFs with low-depth circuits~\cite{banerjee2012pseudorandom,naor1999synthesizers,arunachalam2021quantum}.
We will focus on the construction of~\cite{arunachalam2021quantum}, which relies on the assumption that the Ring Learning with Errors (\textsf{RingLWE}) problem~\cite{lyubashevsky2010ideal} is hard even for quantum computers.
Specifically, we assume that \textsf{RingLWE} cannot be solved by a quantum computer in sub-exponential time, which is a commonly believed cryptographic assumption~\cite{aggarwal2023lattice, arunachalam2021quantum, ananth2023revocable, diakonikolas2022cryptographic, regev2009lattices}.
Here, \textsf{RingLWE} is a variant of the more well-known Learning with Errors problem~\cite{regev2009lattices} over polynomial rings.
The \textsf{RingLWE} problem is to find a secret ring element $s \in R_q \triangleq \mathbb{Z}_q[x]/\langle x^\lambda - 1 \rangle$ given pairs $(a, a \cdot s + e \bmod R_q)$, where $\lambda$ denotes the security parameter, $e$ is some error, $q$ is a parameter of the problem.
We only state this informally here and refer the reader to~\cite{lyubashevsky2010ideal} for a formal definition and discussion.
In~\cite{arunachalam2021quantum}, assuming that \textsf{RingLWE} cannot be solved by quantum computers in $t(\lambda)$ time, their construction produces a PRF secure against $\bigo(t(\lambda))$ quantum adversaries that is implementable by constant-depth, polynomial-size circuits.
We state the precise result below.

\begin{theorem}[Lemmas 3.15 and 3.16 in~\cite{arunachalam2021quantum}]
    \label{thm:PRF}
    Let $\lambda$ denote the security parameter.
    Let the input size be $m = m(\lambda) =  \omega(\log \lambda)$ and set the parameter $q =\lambda^{\omega(1)}$ to be a power of two such that $\log(q) \leq \mathcal{O}(\poly(\lambda))$.
    Let $\mathcal{K} = \{\mathcal{K}_\lambda\}_{\lambda \in \mathbb{N}}$, where $\mathcal{K}_\lambda = R_q^{m+1}$.
    There exists a PRF $\mathcal{RF} = \{f_\lambda\}_{\lambda \in \mathbb{N}}$, where $f_\lambda: R_q^{m+1} \times \{0,1\}^m \to \{0,1\}$, satisfying the following two properties.
    \begin{enumerate}
        \item Every $f_\lambda(\mathbf{k},\cdot) \in \mathcal{RF}$ with $\mathbf{k} \in \mathcal{K}_\lambda$ can be computed by a $\mathsf{TC}^0$ circuit.
        \item Suppose there exists a distinguisher $\mathcal{D}$ for $\mathcal{RF}$, i.e., there exists an $\bigo(t(\lambda))$-time quantum algorithm $\mathcal{D}$ that satisfies
        \begin{equation}
            \left|\Pr_{\mathbf{k} \leftarrow \mathcal{K}_\lambda}[\mathcal{D}^{\ket{f_\lambda(\mathbf{k}, \cdot)}}(\cdot) = 1] - \Pr_{g \in \mathcal{U}}[\mathcal{D}^{\ket{g}}(\cdot) = 1]\right| > \mathsf{negl}(\lambda),
        \end{equation}
        where the key $\mathbf{k}$ is picked uniformly at random from the key space $\mathcal{K}_\lambda$, $g$ is picked uniformly at random from $\mathcal{U}$, the set of all functions from $\mathcal{X}_\lambda$ to $\mathcal{Y}_\lambda$, and $\mathcal{D}^{\ket{f_\lambda(\mathbf{k}, \cdot)}}$ indicates that $\mathcal{D}$ has quantum oracle access to the function $f_\lambda(\mathbf{k}, \cdot)$. Then, there exists a $t(\lambda)$-time quantum algorithm that solves \textsf{RingLWE}.
    \end{enumerate}
\end{theorem}

In Property 2, this is equivalent to saying that the PRF is quantum secure against $\bigo(t(\lambda))$ adversaries, assuming that \textsf{RingLWE} cannot be solved by a $t(\lambda)$-time quantum algorithm.
Also, note that in Property 1, $\mathsf{TC}^0$ circuits refer to constant-depth, polynomial-size circuits with unbounded fan-in \textsf{AND}, \textsf{OR}, \textsf{NOT}, and \textsf{MAJORITY} gates.
We claim that every $\mathsf{TC}^0$ circuit has a quantum circuit computing the same function with poly-logarithmic overhead in depth.

\begin{prop}[Quantum circuits for $\mathsf{TC}^0$]
    Let $C$ be a $\mathsf{TC}^0$ circuit on $m$ inputs computing some Boolean function $f:\{0,1\}^m \to \{0,1\}$.
    Then, there exists a quantum circuit $C'$ on $n = \mathcal{O}(\poly(m))$ qubits of size $\mathcal{O}(n\polylog(n))$ and depth $\mathcal{O}(\polylog(n))$ that implements $f$.
\end{prop}

Here, when we say that $C'$ implements the function $f$, we mean that $C'\ket{x}\ket{z} = C'\ket{x}\ket{z \oplus f(x)}$.

\begin{proof}
    Note that the number of qubits is $n = \mathcal{O}(\poly(m))$ because after each gate in the classical circuit $C$, we must store the result in an ancilla qubit to maintain unitarity.
    Recall that $\mathsf{TC}^0$ circuits are constant-depth, polynomial-size circuits with unbounded fan-in \textsf{AND}, \textsf{OR}, \textsf{NOT}, and \textsf{MAJORITY} gates.
    Thus, it suffices to find the depth of implementing each of these gates quantumly.
    The size then follows because a circuit of depth $d$ on $n$ qubits can have at most $nd$ gates.
    \textsf{NOT} gates can clearly be implemented in constant depth since this is just an $X$ gate.
    An \textsf{AND} gate with $m$ inputs can be completed in logarithmic depth by computing \textsf{AND} pairwise with \textsf{CNOT}.
    Similarly, we can compute an \textsf{OR} gate with the same logarithmic depth.
    It remains to analyze the depth needed for computing a \textsf{MAJORITY} gate.
    Recall that the \textsf{MAJORITY} gate is defined as
    \begin{equation}
        \textsf{MAJ}(x_1,\dots,x_m) = \left\lfloor \frac{1}{2} + \frac{\left(\sum_{i=1}^m x_i\right) - 1/2}{m} \right\rfloor = \left\lfloor \frac{1}{2} + \frac{\sum_{i=1}^m x_i}{m} - \frac{1}{2m} \right\rfloor.
    \end{equation}
    Here, addition is done over the integers and $x_i \in \{0,1\}$.
    We first analyze the depth/size required for the addition $\sum_{i=1}^m x_i$.
    Note that the maximum value of this sum is $m$, which can be stored in $\mathcal{O}(\log m)$ bits.
    Thus, we can write each of the $x_i$ in binary using $\log m$ bits by padding with zeros and perform addition in this way.
    We can perform the addition of the $m$ inputs pairwise, parallelized to $\mathcal{O}(\log m)$ depth and requiring $\mathcal{O}(m)$ addition operations.
    Moreover, one can perform these addition operations using quantum circuits of size and depth $\mathcal{O}(\log m)$~\cite{takahashi2009quantum}.
    The construction in~\cite{takahashi2009quantum} uses Toffoli gates, but these can be decomposed into two-qubit gates with constant overhead~\cite{yu2013five}.
    In total, we have that $\sum_{i=1}^m x_i$ can be implemented by a quantum circuit of depth $\mathcal{O}(\log^2m)$.

    To divide this sum by $m$, note that there exist classical Boolean circuits for integer division of depth $\mathcal{O}(\log \log m)$ since our inputs can be represented in binary using $\log m$ bits~\cite{beame1986log}.
    These Boolean circuits use only standard \textsf{AND}, \textsf{OR}, and  \textsf{NOT} gates.
    As explained previously, these can be implemented quantumly, and for fan-in-$2$ \textsf{AND} and \textsf{OR} gates, this can be done with constant overhead.
    Thus, this division step requires depth $\mathcal{O}(\log \log m)$ in total.

    Finally, we need to compute the remaining addition/subtraction and floor operations.
    The addition/subtraction can be ignored since  they only occur once so that the depth is dominated by the other additions.
    For the floor, because the quantity inside can only be less than or equal to $1$, then this is the same as deciding whether the quantity inside is less than $1$ or not.
    This can be done in a constant number of operations. 

    Putting everything together, we see that the circuit depth for implementing a \textsf{MAJORITY} gate is dominated by $\mathcal{O}(\log^2 m)$.

    Recall again that $\mathsf{TC}^0$ describes constant-depth, polynomial-size circuits with unbounded fan-in \textsf{AND}, \textsf{OR}, \textsf{NOT}, and \textsf{MAJORITY} gates.
    We just analyzed the depth for each of these gates individually.
    In summary, we computed that $\mathcal{O}(\log m)$ quantum depth is sufficient for \textsf{AND} and \textsf{OR}.
    Constant $\mathcal{O}(1)$ depth is sufficient for \textsf{NOT}.
    Finally, $\mathcal{O}(\log^2 m)$ depth is sufficient for \textsf{MAJORITY}.
    In the overall circuit, this totals to $\mathcal{O}(\polylog(m))$ depth.
    Because a circuit of depth $d$ on $n$ qubits can have at most $nd$ gates, then the size of this circuit is $\mathcal{O}(n\polylog(m))$ gates.
    Then, because $n = \mathcal{O}(\poly(m))$, we obtain the claim.
\end{proof}

Alternatively, we note that one can obtain a similar result using~\cite{rosenthal2021query}.
As a simple corollary of this along with \Cref{thm:PRF}, we can bound the depth/size of a quantum circuit for computing a PRF.

\begin{corollary}
    \label{coro:PRF-size}
    Let $\lambda=n$ denote the security parameter.
    Assuming that \textsf{RingLWE} cannot be solved in $t(n)$ time by a quantum computer, there exists a PRF $\mathcal{F} = \{f_\lambda\}_{\lambda \in \mathbb{N}}$ that is secure against $\mathcal{O}(t(n))$ quantum adversaries such that for keys $\mathbf{k} \in \mathcal{K}_\lambda$ (for the same key space as in \Cref{thm:PRF}), $f_\lambda(\mathbf{k},\cdot) : \{0,1\}^m \to \{0,1\}$ is computable by an $n$-qubit quantum circuit of size $\mathcal{O}(n\polylog(n))$ and depth $\mathcal{O}(\polylog(n))$.
\end{corollary}

Note here that by the above analysis, we have $m = \omega(\log(\lambda))$.
Since $\bigo(\poly(m))$ qubits suffice to implement these PRFs, we can take $n=\lambda$, similar to \cite{brakerski2019pseudo}.

Our proofs also require the notion of pseudorandom quantum states.
Informally, pseudorandom quantum states are ensembles of quantum states that are indistinguishable from Haar-random states to any efficient (quantum) algorithm.
Moreover, it is known how to construct these states using efficient quantum circuits.
Recently, pseudorandom quantum states have been of great interest in quantum cryptography~\cite{ananth2022cryptography,kretschmer2023quantum, morimae2022quantum} and complexity theory~\cite{kretschmer2021quantum}.
We define them formally below, following the presentation in~\cite{ji2018pseudorandom,brakerski2019pseudo}.

\begin{definition}[Pseudorandom quantum states (PRS)~\cite{ji2018pseudorandom}]
    \label{def:PRS}
  Let $\lambda=n$ denote the security parameter. Let $\mathcal{K} = \{\mathcal{K}_\lambda\}_{\lambda \in \mathbb{N}}$ be the key space. A keyed family of pure quantum states $\left\{\left|\phi_k\right\rangle\right\}_{k \in \mathcal{K}_\lambda}$ is \emph{pseudorandom against $t(n)$ adversaries} if the following two conditions hold:
\begin{enumerate}
    \item (Efficient generation). There is a polynomial-time quantum algorithm $\mathsf{Gen}$ that generates state $\left|\phi_k\right\rangle$ on input $k$. That is, for all $\lambda \in \mathbb{N}$ and for all $k \in \mathcal{K_\lambda}, \mathsf{Gen}(1^\lambda, k)=\left|\phi_k\right\rangle$.
    \item (Pseudorandomness). Any polynomially many copies of $\left|\phi_k\right\rangle$ with the same random $k \in \mathcal{K}_\lambda$ are computationally indistinguishable from the same number of copies of a Haar-random state. More precisely, for any $t(n)$-time quantum algorithm $\mathcal{D}$ and any $N = \poly(\lambda)$, there exists a negligible function $\mathsf{negl}(\cdot)$ such that for all $\lambda \in \mathbb{N}$,
    \begin{equation}
        \left|\Pr_{k \leftarrow \mathcal{K}_\lambda}\left[\mathcal{D}\left(\left|\phi_k\right\rangle^{\otimes N}\right)=1\right]-\Pr_{|\psi\rangle \leftarrow \mu}\left[\mathcal{D}\left(|\psi\rangle^{\otimes N}\right)=1\right] \right|\leq\operatorname{negl}(\lambda),
    \end{equation}
    where $\mu$ is the Haar measure over pure states on $n$ qubits.
\end{enumerate}
When $t(n)=\poly(n)$, we simply say that the states are \emph{pseudorandom}.
\end{definition}

There exist efficient procedures to generate pseudorandom quantum states under standard cryptographic assumptions.
In particular, we consider the construction by~\cite{brakerski2019pseudo}, which assumes the existence of quantum-secure pseudorandom functions.

\begin{prop}[Corollary of Claims 3 and 4 in~\cite{brakerski2019pseudo}]
    \label{thm:PRS}
    Let $\lambda=n$ denote the security parameter and $t(n) \geq \poly(n)$.
    Assuming that \textsf{RingLWE} cannot be solved by a quantum computer in $t(n)$ time, pseudorandom quantum states secure against $\mathcal{O}(t(n))$ adversaries with key space $\mathcal{K}$ (for the same key space as in \Cref{thm:PRF}) can be prepared using $n$-qubit quantum circuits of depth $\mathcal{O}(\polylog(n))$ and size $\mathcal{O}(n\polylog(n))$.
\end{prop}

\begin{proof}
    Note that using the PRF from~\cite{arunachalam2021quantum} and tracing through the proof of Claim 3 in~\cite{brakerski2019pseudo}, one can clearly see that security holds for $\mathcal{O}(t(n))$ adversaries rather than only efficient adversaries.
    We need to prove that the size and depth are as stated for the construction of pseudorandom quantum states in~\cite{brakerski2019pseudo} using the PRF from~\cite{arunachalam2021quantum}.
    To obtain the depth and size bounds, we analyze the construction in~\cite{brakerski2019pseudo}.
    In Claim 3 of~\cite{brakerski2019pseudo}, they show that their constructed states can be prepared by applying a single layer of Hadamard gates followed by applying a quantum-secure PRF.
    First, the layer of Hadamards has depth $1$ and size $n$.
    Using the construction from~\Cref{coro:PRF-size}, applying the PRF can then be implemented in $\mathcal{O}(\polylog(n))$ depth and $\mathcal{O}(n\polylog(n))$ size.
    Thus, overall, the depth and size are dominated by the cost of evaluating the PRF.
    Moreover, in Claim 4 of~\cite{brakerski2019pseudo}, they prove that this is indeed constructs a pseudorandom quantum state.
\end{proof}

Note again that the number of qubits $n$ in the quantum circuit depends on the security parameter $\lambda$.
In fact, due to the construction used, the $n$ depends on $\lambda$ in the same way as for the PRF construction.
Also note that the above constructions of PRF/PRS can be implemented using a number of Clifford and T gates of the same order.
This is because the $\mathsf{TC^0}$ circuits in the PRF constructions are classical circuits which can be implemented exactly by Toffoli gates, and Toffoli gates can be constructed using a constant number of Clifford and T gates.
Also in the PRS construction, the remaining gates are Hadamard gates which are Clifford gates.
Therefore, the computational hardness results in \Cref{app:states-comp-complexity,app:unitary-comp-complexity} also apply to Clifford+T circuits of the same gate complexity.

\section{Learning quantum states}

Recall that, given copies of a pure state of bounded circuit complexity, we wish to find a classical description for a quantum circuit that approximately implements this state.
It is natural to require the learner to output a circuit description since this ensures that the output of the learner can indeed be used to prepare (approximate) copies of the unknown state. 
This model is similar-in-spirit to learning an (approximate) generator for an unknown classical probability distribution \cite{kearns1994learnability}.
Nevertheless, our sample complexity results hold for learning classical descriptions beyond circuit descriptions, and our computational complexity results immediately extend to learners that output classical descriptions from which a circuit description can be derived efficiently (e.g., matrix product states/operators with constant bond dimension~\cite{foss2021holographic, schon2005sequential}, stabilizer descriptions, etc.).

Specifically, let $\ket{\psi} = U\ket{0}^{\otimes n}$, where $U$ is a unitary consisting of $G$ two-qubit gates.
Throughout this section, we denote $\rho \triangleq \ketbra{\psi}$.
Suppose we are given $N$ identically prepared copies of $\rho$.
The goal is to learn a classical circuit description of a quantum state $\hat{\rho}$ that is $\epsilon$-close to $\rho$ in trace distance, i.e., $\dtr(\hat{\rho}, \rho)=\norm{\hat{\rho} - \rho}_1/2 \leq \epsilon$.

In this appendix, we provide a proof of \Cref{thm:state-learning}, which characterizes the sample complexity for this task.
We restate the theorem below.
\begin{theorem}[State learning, detailed restatement of \Cref{thm:state-learning}]
    \label{thm:state-learning-detail}
    Let $\epsilon,\delta > 0$.
    Suppose we are given $N$ copies of a pure $n$-qubit state density matrix $\rho = \ketbra{\psi}$, where $\ket{\psi} = U\ket{0}^{\otimes n}$ is generated by a unitary $U$ consisting of $G$ two-qubit gates.
    Then, any algorithm that can output $\hat{\rho}$ such that $\dtr(\hat{\rho}, \rho) \leq \epsilon$ with probability at least $1-\delta$ requires at least
    \begin{equation}
        N = \Omega\left(\min\left(\frac{2^n}{\epsilon^2},\frac{G(1-\delta)}{\epsilon^2\log(G/\epsilon)}\right) + \frac{\log(1/\delta)}{\epsilon^2}\right).
    \end{equation}
    Meanwhile, there exists such an algorithm using
    \begin{equation}
        N = \mathcal{O}\left(\min\left(\frac{2^n\log(1/\delta)}{\epsilon^2}, \frac{G\log(G/\epsilon) + \log(1/\delta)}{\epsilon^2}\right)\right).
    \end{equation}
\end{theorem}
Here, the minimum with $2^n/\epsilon^2$ corresponds to the sample-optimal approaches for full quantum state tomography~\cite{haah2017sample,o2016efficient}.
The theorem in the main text corresponds to $\delta = \mathcal{O}(1)$ so that the upper and lower bounds are equal up to logarithmic factors.

In \Cref{app:states-upper} we prove the sample complexity upper bound, and in \Cref{app:states-lower}, we show the sample complexity lower bound.
Moreover, in \Cref{app:states-comp-complexity}, we prove \Cref{thm:state-comp-complex}, which gives a lower bound on the computational complexity required for this task.

\subsection{Sample complexity upper bound}
\label{app:states-upper}

In this section, we prove the sample complexity upper bound for \Cref{thm:state-learning-detail}.
We provide an algorithm for learning the unknown quantum state within trace distance $\epsilon$ by constructing a covering net over the space of all unitaries consisting of $G$ two-qubit gates.
We can then obtain a covering net over all pure quantum states generated by $G$ two-qubit gates by applying each element of the unitary covering net to the zero state.
With this covering net, we can use quantum hypothesis selection~\cite{badescu2020improved} based on classical shadows~\cite{huang2020predicting} (discussed in \Cref{app:hypothesis}) to identify a state in the covering net that is close to the unknown target state.
We note that this strategy may be adapted to other restricted state/unitary classes as long as we can construct a covering net with bounded cardinality.

\begin{prop}[State learning upper bound]
    \label{prop:states-upper}
    Let $\epsilon,\delta > 0$.
    Suppose we are given $N$ copies of a pure $n$-qubit state density matrix $\rho = \ketbra{\psi}$, where $\ket{\psi} = U\ket{0}^{\otimes n}$ is generated by a unitary $U$ consisting of $G$ two-qubit gates.
    Then, there exists an algorithm that can output $\hat{\rho}$ such that $\dtr(\hat{\rho}, \rho) \leq \epsilon$ with probability at least $1-\delta$ using
    \begin{equation}
        \label{eq:state-upper}
        N = \mathcal{O}\left(\min\left(\frac{2^n\log(1/\delta)}{\epsilon^2}, \frac{G\log(G/\epsilon) + \log(1/\delta)}{\epsilon^2}\right)\right)
    \end{equation}
    samples of $\ket{\psi}$.
\end{prop}

Here, we take the minimum with $2^n/\epsilon^2$, as this is the upper bound achieved for full quantum state tomography on an arbitrary $n$-qubit pure state~\cite{haah2017sample,o2016efficient}.
Thus, we focus on proving the second term in the minimum.
We prove this upper bound by considering two cases: (1) $G \geq n/2$ and (2) $G < n/2$.
The upper bounds for each case agree and are given by \Cref{eq:state-upper}.
We first prove the proposition for Case (1) and indicate what changes for Case (2).

\begin{proof}[Proof of Case (1)]
    As previously described, this follows by first creating a covering net over all unitaries consisting of $G$ two-qubit gates and then using quantum hypothesis selection~\cite{badescu2020improved}.
    
    By \Cref{thm:covering}, we know that there exists an $(\epsilon/6)$-covering net $\mathcal{N}_{\epsilon/6}$ of the space of unitaries implemeted by $G$ two-qubits gates with respect to the diamond distance $\dworst$ with metric entropy bounded by
    \begin{equation}
    \label{eq:covering-size}
        \log(|\mathcal{N}_{\epsilon/6}|) \leq 32G \log\left(\frac{72G}{\epsilon}\right) + 2G\log(n).
    \end{equation}
    Applying each unitary $V' \in \mathcal{N}_{\epsilon/6}$ to the zero state, we obtain a new covering net
    \begin{equation}
        \mathcal{N}'_{\epsilon/6} =\{V'\ketbra{0}^{\otimes n} V'^\dagger : V' \in \mathcal{N}_{\epsilon/6}\}
    \end{equation}
    for the set of pure quantum states generated by $G$ two-qubit gates with respect to trace distance.
    We argue that this is true as follows.
    Any pure quantum state generated by $G$ two-qubit gates can be written as $\ket{\phi} = V\ket{0}^{\otimes n}$ for some unitary $V$ implemented by $G$ two-qubit gates and let $\sigma = \ketbra{\phi}$.
    Using the definition of the covering net $\mathcal{N}_{\epsilon/6}$, there exists a unitary $V' \in \mathcal{N}_{\epsilon/6}$ such that $\dworst(V, V') < \epsilon/6$.
    Consider $\ket{\phi'} = V'\ket{0}^{\otimes n}$ and let $\sigma' = \ketbra{\phi'} \in \mathcal{N}'_{\epsilon/6}$.
    By the definition of the diamond distance in terms of a worst case over input states, we also have $\dtr(\sigma, \sigma') \leq \dworst(V, V') \leq \epsilon/12 < \epsilon/6$.
    Thus, $\mathcal{N}_{\epsilon/6}'$ satisfies the definition of a covering net over the pure quantum states generated by $G$ two-qubit gates with respect to trace distance $\dtr$.
    Moreover, we clearly see that $|\mathcal{N}'_{\epsilon/6}| \leq  |\mathcal{N}_{\epsilon/6}|$.

    We can consider this covering net $\mathcal{N}'_{\epsilon/6}$ as the set of hypothesis states in \Cref{prop:hypothesis}. 
    Let $\rho = \ketbra{\psi}$ be the unknown quantum state that we have copies of.
    By \Cref{prop:hypothesis}, there exists an algorithm to learn $\tilde{\rho}$ such that
    \begin{equation}
        \dtr(\rho, \tilde{\rho}) \leq 3\cdot \frac{\epsilon}{6} + \frac{\epsilon}{2} = \epsilon
    \end{equation}
    with probability at least $1-\delta$.
    Here, note that we used $\eta = \epsilon/6$ in \Cref{prop:hypothesis} by definition of an $(\epsilon/6)$-covering net.
    Furthermore, we may choose $\epsilon_2 = \epsilon/2$ and $\delta_1 = \delta/2$. In this way, we obtain $\tilde{\rho}$ such that $\dtr(\rho, \tilde{\rho}) \leq \epsilon$
    with probability at least $1-\delta$. Moreover, by \Cref{prop:hypothesis}, this algorithm to find $\tilde{\rho}$ requires at most
    \begin{equation}
        N = \mathcal{O}\left(\frac{\log(|\mathcal{N}'_{\epsilon/6}|/\delta)}{\epsilon^2}\right) = \mathcal{O}\left(\frac{G\log(G/\epsilon) + G\log(n) + \log(1/\delta)}{\epsilon^2}\right)
    \end{equation}
    copies of $\rho$, where the second equality follows from Eq.~\eqref{eq:covering-size}.
    Because we are considering $G \geq n/2$ in this case, then we have
    \begin{equation}
        N = \mathcal{O}\left(\frac{G\log(G/\epsilon) + \log(1/\delta)}{\epsilon^2}\right),
    \end{equation}
    as claimed.
\end{proof}

Notice in the above proof that we used $G \geq n/2$ in the last step to remove the extra $\log(n)$ factor.
However, in Case (2), we can no longer execute this step and must consider a more careful strategy to remove the dependence on system size $n$.
The key observation is that if $G < n/2$, some qubits in the system will be left in the zero state because no gate has acted upon them (for $G$ two-qubit gates, at most $2G < n$ qubits are acted upon nontrivially).
Notice that we only need to learn the quantum state on these $2G$ qubits rather than the whole system, since we can simply tensor product with the zero state for the remaining qubits.
Thus, we require the ability to discern which qubits have been acted upon by the $G$ two-qubit gates.
Once we find this set of qubits, the idea is to consider a covering net for the set of pure quantum states generated by $G$ two-qubit gates \emph{on this restricted system}.
Then, we can follow a similar argument to the above proof of Case (1).

We prove Case (2) of \Cref{prop:states-upper} in the following sections.
For the rest of this section, let $\rho = \ketbra{\psi}$.
In \Cref{app:states-upper-post}, we discuss an algorithm that identifies the qubits acted on nontrivially by the $G$ two-qubit gates with high probability and show that restricting to these identified qubits does not cause much error.
In \Cref{app:states-upper-permute}, we resolve a technical issue for defining the covering net on the restricted system, which stems from the algorithm possibly not identifying all qubits.
Finally, in \Cref{app:states-upper-proof}, we combine these pieces to provide the full proof of Case (2).

\subsubsection{Postselection}
\label{app:states-upper-post}

First, we present an algorithm to determine which qubits of the unknown quantum state $\rho = \ketbra{\psi}$ have been acted upon nontrivially by the $G$ two-qubit gates.
We then prove a guarantee about the number of samples of $\rho$ needed to determine these qubits with high probability.
We also show that considering $\rho$ to be the zero state on the rest of the qubits does not incur much error.

Suppose that the true set of qubits acted upon by the $G$ two-qubit gates is denoted as $A$.
To determine which qubits are in the set $A$, consider the procedure given in Algorithm \ref{algo:nonzero-qubits}.
\begin{algorithm}
    \label{algo:nonzero-qubits}
    \caption{Identify qubits acted upon nontrivially (state version)}
    \KwIn{Copies of unknown $n$-qubit quantum state $\rho$.}
    \KwOut{List $\hat{A} \subseteq [n]$ of qubits.}
    Initialize $\hat{A} = \emptyset$.\\
    Repeat the following $N = \mathcal{O}\left(\frac{G + \log (1/\delta_1)}{\epsilon_1}\right)$ times: \linebreak
    (a) Measure all qubits of the unknown state $\rho$ in the computational basis.
    \linebreak 
    (b) Given the measurement outcome $\ket{x}$, set $\hat{A} \leftarrow \hat{A} \cup \mathrm{supp}(x)$, where $\mathrm{supp}(x) = \{i \in [n] : x_i \neq 0\}$.
\end{algorithm}
The idea behind this algorithm is simple.
If we measure a qubit in the computational basis and receive a nonzero measurement outcome, then it must have been acted upon by one of the $G$ two-qubit gates because the quantum state is assumed to have been initialized in the zero state.
We prove that $\mathcal{O}\left(\frac{G + \log(1/\delta_1)}{\epsilon_1}\right)$ copies of $\rho$ suffice to obtain, with high probability $1-\delta_1$, the desired property that measuring the qubits in $\hat{B} \triangleq [n] \setminus \hat{A}$ of $\rho$ yields the all zero bit string with high probability $1-\epsilon_1$.

\begin{lemma}
    \label{lemma:good-overlap}
    Let $\epsilon_1, \delta_1 > 0$. Suppose we are given copies of a pure $n$-qubit quantum state $\rho=\ketbra{\psi}$ generated by $G$ two-qubit gates acting on a subset of the qubits $A \subseteq [n]$. Then, Algorithm~\ref{algo:nonzero-qubits} uses $N = \mathcal{O}\left(\frac{G + \log(1/\delta_1)}{\epsilon_1}\right)$ copies of $\rho$ and outputs with probability at least $1-\delta_1$ a list $\hat{A} \subset [n]$ such that
    \begin{equation}
        \expval{\rho_{\hat{B}}}{0_{\hat{B}}} \geq 1 - \epsilon_1,
    \end{equation}
    where $\rho_{\hat{B}}$ denotes the reduced density matrix of $\rho$ when tracing out all qubits other than those in the set $\hat{B} = [n] \setminus \hat{A}$ and $\ket{0_{\hat{B}}}$ denotes the zero state on all qubits in ${\hat{B}}$.
\end{lemma}

\begin{proof}
    Let $A'$ be any possible set that could be output by Algorithm~\ref{algo:nonzero-qubits}.
    Let $B' \triangleq [n] \setminus A'$.
    We first define some random variables to state our claim more precisely.
    Let $E_{i, A'}$ be the event that round $i$ of measurement of the qubits in $B' = [n] \setminus A'$ in Algorithm~\ref{algo:nonzero-qubits} yields the all zero bitstring.
    Let $X_{i, A'}$ be the indicator random variable corresponding to the event $E_{i,A'}$.
    Then, we have that $\bar{X}_{A'} \triangleq \frac{1}{N}\sum_{i=1}^N X_{i, A'}$ is the number of times the qubits in $B'$ are all measured to be zero divided by the total number of measurements.
    In other words, $\bar{X}_{A'}$ is an empirical estimate for the overlap that the state $\rho_{B'}$ on qubits in $B'$ has with the all zero state.
    Moreover, we have
    \begin{equation}
    \mathbb{E}[X_{A'}] \triangleq \mathbb{E}[X_{i,A'}] = \expval{\rho_{B'}}{0_{B'}}
    \end{equation}
    for all $A'$. Note that the first definition makes sense because for any $i$, the $X_{i,A'}$ are identically distributed.
    This says that the true expectation of our random variables is the true overlap of the state $\rho_{B'}$ with the all zero state.
    
    We claim that for any $A'$, if the true overlap is less than $1-\epsilon_1$, then the estimated overlap is less than $1 - \epsilon_1/2$ with high probability.
    Formally, in terms of our random variables, this is the following statement: 
    \begin{claim}
        \label{claim:concentrate}
        For any set $A'$ that could be output by Algorithm~\ref{algo:nonzero-qubits}, if $\mathbb{E}[X_{A'}] < 1 - \epsilon_1$, then $\bar{X}_{A'} < 1-\epsilon_1/2$ with probability at least $1-\delta_1$.
    \end{claim}
    Thus, we have reduced our task to a concentration problem.
    Note that it suffices to prove this because the set $\hat{A}$ actually identified by Algorithm~\ref{algo:nonzero-qubits} has $\bar{X}_{\hat{A}} = 1$.
    This is true because a qubit is only added to the set $\hat{A}$ in the algorithm if it measured and observed a nonzero outcome.
    Thus, all qubits in $\hat{B} = [n] \setminus \hat{A}$ must have given zero when measured throughout all rounds of measurement.
    By definition, this gives us that $\bar{X}_{\hat{A}} = 1$.
    Then, by the contrapositive of Claim~\ref{claim:concentrate}, we see that $\mathbb{E}[X_{\hat{A}}] = \expval{\rho_{\hat{B}}}{0_{\hat{B}}} \geq 1 - \epsilon_1$, with probability at least $1-\delta_1$.
    We now prove this claim using classical concentration inequalities.

    \begin{proof}[Proof of Claim~\ref{claim:concentrate}]
        First, fix some set $A'$ that could be output by Algorithm~\ref{algo:nonzero-qubits}.
        Suppose 
        \begin{equation}
            \mathbb{E}[X_{A'}] \triangleq 1-a < 1-\epsilon_1,
        \end{equation}
        where $a > \epsilon_1$.
        Recall the Bhatia-Davis Inequality, which states that for $X \in [b,d]$ that
        \begin{equation}
            \mathrm{Var}(X) \leq (d - \mathbb{E}[X])(\mathbb{E}[X] - b).
        \end{equation}
        In our case, we have $X_{A'} \in [0,1]$ since they are indicator random variables so that the inequality gives us
        \begin{equation}
            \mathrm{Var}(X_{A'}) \leq (1 - \mathbb{E}[X])\mathbb{E}[X] \leq 1 - \mathbb{E}[X] = a.
        \end{equation}
        Now, recall Bernstein's Inequality, which states that for independent random variables $X_i$ with $|X_i| \leq c$ and $\sigma^2 = \frac{1}{N}\sum_{i=1}^N\mathrm{Var}(X_i)$, we have for any $t > 0$,
        \begin{equation}
            \Pr\left(\frac{1}{N}\sum_{i=1}^N X_i - \mathbb{E}[X] > t\right) \leq \exp\left(-\frac{Nt^2}{2\sigma^2 + 2ct/3}\right).
        \end{equation}
        In our case, $c = 1, \sigma^2 \leq a$, and $t = a/2$. Then, Bernstein's Inequality results in
        \begin{equation}
            \Pr\left(\bar{X}_{A'} - \mathbb{E}[X] > \frac{a}{2}\right) \leq \exp\left(-\frac{Na^2/4}{2a+a/3}\right).
        \end{equation}
        Plugging in $\mathbb{E}[X] = 1-a$ and simplifying, we have
        \begin{equation}
            \Pr\left(\bar{X}_{A'} > 1 - \frac{a}{2}\right) \leq \exp\left(-\frac{3Na}{28}\right) \leq \exp\left(-\frac{3N\epsilon_1}{28}\right).
        \end{equation}
        Since $a > \epsilon_1$, then $1 - a/2 < 1 - \epsilon_1/2$ so that we have
        \begin{equation}
            \Pr\left(\bar{X}_{A'} > 1 - \frac{\epsilon_1}{2}\right) \leq \exp\left(-\frac{3N\epsilon_1}{28}\right).
        \end{equation}
        Plugging in $N = \frac{28\log(2^{2G}/\delta_1)}{3\epsilon_1}$, we have
        \begin{equation}
            \label{eq:fixed-B'}
            \Pr\left(\bar{X}_{A'} > 1 - \frac{\epsilon_1}{2}\right) \leq \frac{\delta_1}{2^{2G}}.
        \end{equation}
        Recall that this inequality was for a single fixed set $A'$, but we want our claim to hold for any set $A'$.
        Thus, we need to union bound overall possible sets $A'$ output by Algorithm~\ref{algo:nonzero-qubits}.
        
        We claim that the number of such sets is at most $2^{2G}$.
        This is clear because if $A'$ is output by the algorithm, then $A' \subseteq A$, where $A$ is the true set of qubits that the $G$ gates act nontrivially on.
        This is true by construction because in order for a qubit to be added to the set output by Algorithm~\ref{algo:nonzero-qubits}, its result upon measurement must have yielded a nonzero outcome so that a gate must have acted upon this qubit.
        Hence $A' \subseteq A$, and because $|A| \leq 2G$, the number of possible subsets $A'$ of $A$ is at most $2^{2G}$.

        Thus, applying a union bound to Eq.~\eqref{eq:fixed-B'}, we see that the probability that, for any $A'$, $\bar{X}_{A'}$ is greater than $1-\epsilon_1/2$ is at most $\delta_1$.
        In other words, $\bar{X}_{A'}$ is less than $1-\epsilon_1/2$ with probability at least $1-\delta_1$.
        Moreover, here we used
        \begin{equation}
            N = \frac{28\log(2^{2G}/\delta_1)}{3\epsilon_1} = \mathcal{O}\left(\frac{G + \log(1/\delta_1)}{\epsilon_1}\right).
        \end{equation}
        This concludes the proof of the claim, which gives the result in \Cref{lemma:good-overlap} as explained previously.
    \end{proof}
\end{proof}

With this, we know that measuring qubits in $\hat{B} = [n] \setminus \hat{A}$ of $\rho$ yields the all zero bistring with high probability.
We want to show that, in fact, we can consider $\rho_{\hat{B}}$ as being the zero state without incurring much error.
In particular, we want to show the following lemma.

\begin{lemma}
\label{lem:postselect}
    Let $\epsilon, \delta_1 > 0$.
    Suppose we are given $N = \mathcal{O}\left(\frac{G + \log (1/\delta_1)}{\epsilon^2}\right)$ copies of an $n$-qubit quantum state $\rho$ generated by $G$ gates. Let $\hat{A} \subset [n]$ be as in Algorithm~\ref{algo:nonzero-qubits} and $\hat{B} = [n] \setminus \hat{A}$. Then, for $\Lambda = \ketbra{0_{\hat{B}}} \otimes I_{\hat{A}}$ (where $\ket{0_{\hat{B}}}$ denotes the zero state on all qubits in ${\hat{B}}$) and for the post-measurement state
    \begin{equation}
        \rho' \triangleq \frac{\sqrt{\Lambda}\rho \sqrt{\Lambda}}{\Tr(\Lambda\rho)},
    \end{equation}
    we have
    \begin{equation}
        \dtr(\rho, \rho') \leq \frac{\epsilon}{24}
    \end{equation}
    with probability at least $1-\delta_1$.
\end{lemma}

In other words, we want to show that our original state $\rho$ is not far in trace distance from the new state $\rho'$, where $\rho'$ is the state $\rho$ with the qubits in $\hat{B}$ projected to the zero state. In this way, we can effectively only consider the system on qubits in $\hat{A}$ when defining the covering net and using hypothesis selection. This turns out to be a bit more nuanced, but this is the general idea. To show this, we will use the Gentle Measurement Lemma, following the presentation in~\cite{wildebook}.

\begin{lemma}[Lemma 9.4.1 in~\cite{wildebook}]
    \label{lemma:gentle-meas}
    Consider a density operator $\rho$ and a measurement operator $\Lambda$, where $0 \leq \Lambda \leq I$. The measurement operator could be an element of a POVM. Suppose that the measurement operator $\Lambda$ has a high probability of detecting the state $\rho$:
    \begin{equation}
        \Tr(\Lambda\rho) \geq 1- \epsilon,
    \end{equation}
    where $\epsilon \in [0,1]$ (the probability of detection is high if $\epsilon$ is close to zero). Then the post-measurement state
    \begin{equation}
        \rho' \triangleq \frac{\sqrt{\Lambda}\rho \sqrt{\Lambda}}{\Tr(\Lambda \rho)}
    \end{equation}
    is $\sqrt{\epsilon}$-close to the original state $\rho$ in trace distance:
    \begin{equation}
        \dtr(\rho, \rho') \leq \sqrt{\epsilon}.
    \end{equation}
    Thus, the measurement does not disturb the state $\rho$ by much if $\epsilon$ is small.
\end{lemma}

With this, we can now prove~\Cref{lem:postselect}.

\begin{proof}[Proof of~\Cref{lem:postselect}]
    As stated above, let $\hat{A} \subset [n]$ be as in Algorithm~\ref{algo:nonzero-qubits}, and let $A \subset [n]$ be true the set of qubits acted non-trivially on by the $G$ gates. Let $\hat{B} \triangleq [n] \setminus \hat{A}$ and let $B \triangleq [n] \setminus A$.

    In order to apply the Gentle Measurement Lemma, we need to show that
    \begin{equation}
        \Tr(\Lambda\rho) \geq 1- \left(\frac{\epsilon}{24}\right)^2.
    \end{equation}
    Since $\Lambda = \ketbra{0_{\hat{B}}} \otimes I_{\hat{A}}$, where $\ket{0_{\hat{B}}}$ denotes the zero state on all qubits in ${\hat{B}}$, we have
    \begin{equation}
        \Tr(\Lambda\rho) = \Tr((\ketbra{0_{\hat{B}}} \otimes I_{\hat{A}})\rho) = \Tr(\ketbra{0_{\hat{B}}}\rho_{\hat{B}}) = \expval{\rho_{\hat{B}}}{0_{\hat{B}}},
    \end{equation}
    where $\rho_{\hat{B}}$ denotes the reduced density matrix obtained by tracing out all qubits in $[n]\setminus {\hat{B}}$.
    Thus, it suffices to show that
    \begin{equation}
        \expval{\rho_{\hat{B}}}{0_{\hat{B}}} \geq 1 - \left(\frac{\epsilon}{24}\right)^2.
    \end{equation}
    Intuitively, this makes sense because in Algorithm~\ref{algo:nonzero-qubits}, we identified the qubits in ${\hat{B}}$ as those being close to the zero state.
    Indeed, this holds by~\Cref{lemma:good-overlap} when choosing $\epsilon_1 = (\epsilon/24)^2$.
    Thus, the result follows.
\end{proof}

\subsubsection{Permutation}
\label{app:states-upper-permute}
Before we can prove \Cref{prop:states-upper}, we must resolve a technical issue. Namely, ideally, we would like to consider a covering net on the subsystem of qubits in the set $A$ (the true set of qubits that the $G$ gates generating the unknown state $\rho$ act nontrivially on). In this way, because $\hat{A} \subseteq A$, where $\hat{A}$ is the set of qubits identified by Algorithm~\ref{algo:nonzero-qubits}, then our postselected state $\rho'$ from \Cref{lem:postselect} should be close to some state in this covering net on the subsystem.
This nearby state in the covering net can then be identified via quantum hypothesis selection~\cite{bădescu2023improved}.
By \Cref{lem:postselect}, this state from hypothesis selection is also close to the original unknown state $\rho$.

However, the problem with the above is that we do not know the true set of qubits $A$; we only know the identified set of qubits $\hat{A}$.
Moreover, it is possible that $\hat{A} \subsetneq A$, i.e., Algorithm~\ref{algo:nonzero-qubits} may not have been able to detect certain qubits as having been acted upon nontrivially by the $G$ gates.
For example, suppose that when preparing the unknown state $\rho$, certain qubits are used as workspace ancillas and are reset to the zero state at the end of the computation.

In order to define a covering net on a system on which the $G$ gates act (the setting of \Cref{lem:2q-unitaries}), we need to somehow identify the qubits in $A \setminus \hat{A}$ that are undetected by the algorithm.
To do so, we argue that we can permute the qubits outside of the set $\hat{A}$ and not deviate much from the original state $\rho$.
In this way, without loss of generality, we can permute the qubits such that those in $A\setminus \hat{A}$ are grouped together in some fixed set of qubits. 
Then, we can define a covering net on the system of qubits defined by this fixed set containing the qubits in $A \setminus \hat{A}$ and our identified set $\hat{A}$.
By construction, we know that the $G$ gates act on this subset of qubits, so this is the correct setting of \Cref{lem:2q-unitaries}.
We note that the permutations used in the proof are a mathematical tool for the analysis, but the learner has to neither know nor perform these permutations.

To formalize this, we first define a permutation and claim that permuting the qubits outside of the set $\hat{A}$ does not change the post selected state $\rho'$.

\begin{definition}[Permutation]
    \label{def:permute}
    A unitary $W \in \mathrm{U}(2^n)$ is a \emph{permutation unitary} if it satisfies the following property: $W$ corresponds to a permutation $\sigma_W \in S_n$ of order $2$, where $S_n$ is the symmetric group of size $n$, and $W$ acts as
    \begin{equation}
        W\ket{x_1\dots x_n} = \ket{x_{\sigma_W(1)}\cdots x_{\sigma_W(n)}},
    \end{equation}
    where $x = x_1\cdots x_n \in \{0,1\}^n$. Moreover, we use $W_S$ for a set $S \subseteq \{1,\dots, n\}$ to denote a permutation unitary where the corresponding permutation $\sigma_{W_S}$ is such that $\sigma_{W_S}|_{\overline{S}} = \mathrm{id}$, where $\overline{S} = [n] \setminus S$. In other words, $\sigma_{W_S}$ only permutes the elements in $S$.
\end{definition}

It is easy to see here that because the corresponding permutation is of order $2$, $W$ is Hermitian. Our next lemma shows that such permutations when acting only on $\hat{B}$ do not change our post selected state.

\begin{lemma}
    \label{lem:permute-post}
    Let $\rho'$ be as in \Cref{lem:postselect}.
    Explicitly, let $\hat{A} \subset [n]$ be as in Algorithm~\ref{algo:nonzero-qubits} and $\hat{B} = [n]\setminus \hat{A}$. Then, for $\Lambda = \ketbra{0_{\hat{B}}}\otimes I_{\hat{A}}$ (where $\ket{0_{\hat{B}}}$ denotes the zero state on all qubits in ${\hat{B}}$), define
    \begin{equation}
        \rho' = \frac{\sqrt{\Lambda}\rho \sqrt{\Lambda}}{\Tr(\Lambda \rho)}.
    \end{equation}
    Then, we have
    \begin{equation}
        \rho'' \triangleq W_{\hat{B}} \rho' W_{\hat{B}} = \rho',
    \end{equation}
    where $W_{\hat{B}}$ is any permutation unitary which only permutes qubits in ${\hat{B}}$.
\end{lemma}

\begin{proof}
    To see the claim, we can simply expand the expression for $\rho''$:
    \begin{align}
        \rho'' &= W_{\hat{B}} \rho' W_{\hat{B}}\\
        &= W_{\hat{B}} \frac{\sqrt{\Lambda} \rho \sqrt{\Lambda}}{\Tr(\Lambda \rho)} W_{\hat{B}}\\
        &= W_{\hat{B}} \frac{\Lambda \rho \Lambda}{\Tr(\Lambda \rho)} W_{\hat{B}}\\
        &= \frac{W_{\hat{B}}(\ketbra{0_{\hat{B}}} \otimes I) \rho (\ketbra{0_{\hat{B}}} \otimes I)W_{\hat{B}}}{\Tr(\Lambda \rho)}\\
        &= \frac{(\ketbra{0_{\hat{B}}} \otimes I) \rho (\ketbra{0_{\hat{B}}} \otimes I)}{\Tr(\Lambda \rho)}\\
        &= \rho',
    \end{align}
    where in the third line we used that $\Lambda$ is a projector so that $\sqrt{\Lambda} = \Lambda$, and in the fifth line, we used the $W_{\hat{B}}$ only permutes the qubits in ${\hat{B}}$, which does not have any effect because here all qubits in ${\hat{B}}$ are in the zero state.
\end{proof}

\begin{lemma}
    \label{lem:permute-orig}
    Let $\epsilon, \delta_2 > 0$.
    The trace distance between $\rho$ and the permuted state $\tilde{\rho} = W_{\hat{B}} \rho W_{\hat{B}}$, where $W_{\hat{B}}$ is any permutation unitary which only permutes qubits in ${\hat{B}}$, is less than $\epsilon/24$:
    \begin{equation}
        \dtr(\rho, \tilde{\rho}) \leq \frac{\epsilon}{12}
    \end{equation}
    with probability at least $1-\delta_2$.
\end{lemma}

\begin{proof}
    This proof combines \Cref{lem:permute-post,lem:postselect}. The idea is the following.
    We know from \Cref{lem:postselect} that $\rho$ and the post selected state $\rho'$ are close in trace distance.
    Moreover, by \Cref{lem:permute-post}, we know that the post selected state $\rho'$ and the permuted post selected state $\rho''$ are the equal (without error).
    We can also show similarly to \Cref{lem:postselect} that the permuted state $\tilde{\rho}$ is close to the post selected state $\tilde{\rho}'$, where this postselection is done in the same way as \Cref{lem:postselect} by replacing $\rho$ with $\tilde{\rho}$. Moreover, we can see that $\rho'' = \tilde{\rho}'$, so the claim then follows by triangle inequality.

    Now, let us formalize this. By \Cref{lem:postselect}, we have
    \begin{equation}
        \label{eq:orig-post}
        \dtr(\rho, \rho') \leq \frac{\epsilon}{24},
    \end{equation}
    with probability at least $1-\delta_2/2$ (choosing $\delta_1 = \delta_2/2$) where
    \begin{equation}
        \rho' = \frac{\sqrt{\Lambda}\rho \sqrt{\Lambda}}{\Tr(\Lambda\rho)}
    \end{equation}
    for $\Lambda = \ketbra{0_{\hat{B}}} \otimes I_{{\hat{A}}}$. By \Cref{lem:permute-post}, we know that
    \begin{equation}
        \label{eq:permute-post}
        \rho'' \triangleq W_{\hat{B}}\rho' W_{\hat{B}} = \rho',
    \end{equation}
    where $W_{\hat{B}}$ is a permutation that only affects qubits in ${\hat{B}}$. 
    Now, consider the permuted state $\tilde{\rho} = W_{\hat{B}} \rho W_{\hat{B}}$.
    Recall that in the proof of \Cref{lem:postselect}, to obtain Eq.~\eqref{eq:orig-post}, it sufficed to show that $\Tr(\Lambda \rho) \geq 1 - (\epsilon/24)^2$, and the result followed by the Gentle Measurement Lemma (Lemma~\ref{lemma:gentle-meas}). Thus, by the same proof, as long as $\Tr(\Lambda \tilde{\rho}) \geq 1- (\epsilon/24)^2$, then we also have
    \begin{equation}
        \label{eq:post-permute}
        \dtr(\tilde{\rho}, \tilde{\rho}') \leq \frac{\epsilon}{24},
    \end{equation}
    with probability at least $1-\delta_2/2$, where
    \begin{equation}
        \tilde{\rho}' \triangleq \frac{\sqrt{\Lambda}\tilde{\rho}\sqrt{\Lambda}}{\Tr(\Lambda \tilde{\rho})}.
    \end{equation}
    We can clearly see that this condition holds:
    \begin{equation}
        \Tr(\Lambda \tilde{\rho}) = \Tr((\ketbra{0_{\hat{B}}} \otimes I_{{\hat{A}}})W_{\hat{B}} \rho W_{\hat{B}}) = \Tr((\ketbra{0_{\hat{B}}} \otimes I_{\hat{A}}) \rho) = \Tr(\Lambda \rho) \geq 1- (\epsilon/24)^2,
    \end{equation}
    where the second equality follows because $W_{\hat{B}}$ only permutes qubits in ${\hat{B}}$, which (rearranging with the trace) does not have any effect on $\ketbra{0_{\hat{B}}}$ because all qubits in ${\hat{B}}$ are in the zero state. Thus, Eq.~\eqref{eq:post-permute} holds.

    We also claim that $\rho'' = \tilde{\rho}'$. This follows by effectively the same proof as \Cref{lem:permute-post}.

    Putting everything together, we have that $\rho' = \rho'' = \tilde{\rho}'$. Thus, by Eq.~\eqref{eq:orig-post},
    \begin{equation}
        \dtr(\rho, \tilde{\rho}') \leq \frac{\epsilon}{24}
    \end{equation}
    with probability at least $1-\delta_2/2$. By triangle equality with Eq.~\eqref{eq:post-permute}, we then obtain the claim:
    \begin{equation}
        \dtr(\rho, \tilde{\rho}) \leq \frac{\epsilon}{12}
    \end{equation}
    with probability at least $1-\delta_2$.
\end{proof}

\subsubsection{Proof of Case (2) of \Cref{prop:states-upper}}
\label{app:states-upper-proof}

With this, we can prove Case (2) of \Cref{prop:states-upper}.
Recall that in Case (2), we require that $G < n/2$.
We provided a sketch of the argument throughout the previous sections, so we put everything together here.

\begin{proof}[Proof of Case (2) of \Cref{prop:states-upper}]
Let $\epsilon, \delta > 0$.
Consider $G < n/2$.
Because $G$ is small compared to $n$, there exist some qubits that have not been acted upon by the $G$ gates used to generate the state $\rho = \ketbra{\psi}$.
Thus, since we assume that the unknown quantum state $\rho$ is constructed by applying a unitary to the all zero state, then these qubits not acted upon by the $G$ gates remain in the zero state.
Using the techniques in \Cref{app:states-upper-post}, we can find the qubits that are acted on nontrivially by the $G$ gates.
Then, we want to consider the covering net on only this set of qubits.
However, because our algorithm does not necessarily find \emph{all} qubits acted on nontrivially by the $G$ gates, we argue in \Cref{app:states-upper-permute} that we can permute the qubits in the system without significantly affecting the original state $\rho$.
In this way, we can consider a permutation which gathers those qubits acted upon nontrivially that our algorithm did not find into some fixed set.
We can then define the covering net on the subsystem consisting of this fixed set along with the identified set of qubits.

Let us now formalize these ideas.
Let $\hat{A}$ be the set of qubits identified by Algorithm~\ref{algo:nonzero-qubits}, and let $A$ be the true set of qubits acted on nontrivially by the $G$ gates.
Let $W_{\hat{B}}$ be a permutation only affecting the qubits in ${\hat{B}} \triangleq [n] \setminus A$ (\Cref{def:permute}) which gathers the qubits in $A \setminus {\hat{A}}$ into some fixed set of qubits $C$. Since $|C| + |{\hat{A}}| = |A \setminus {\hat{A}}| + |{\hat{A}}| = |A| \leq 2G$, then $C \cup {\hat{A}}$ has at most $2G$ qubits and these qubits are acted upon by $G$ gates.

By \Cref{thm:covering} we know that there exists an $(\epsilon/12)$-covering net $\mathcal{N}_{\epsilon/12}$ of the space of unitaries implemented by $G$ two-qubit gates on the permuted system consisting of only qubits in $C \cup \hat{A}$ with respect to the diamond distance $\dworst = \max_\rho \norm{(U\otimes I) \rho (U\otimes I)^\dagger - (V\otimes I)\rho(V\otimes I)^\dagger}_1$.
Moreover, this covering net has metric entropy bounded by
\begin{equation}
    \label{eq:covering-size2}
    \log(|\mathcal{N}_{\epsilon/12}|) \leq 32G \log\left(\frac{144G}{\epsilon}\right) + 2G\log(2G) = \mathcal{O}\left(G \log(G/\epsilon)\right).
\end{equation}
We can instead consider
\begin{equation}
    \mathcal{N}'_{\epsilon/12} = \{V' \ketbra{0_{C \cup \hat{A}}}V'^\dagger : V' \in \mathcal{N}_{\epsilon/12}\},
\end{equation}
where $\ket{0_{C \cup {\hat{A}}}}$ denotes the zero state on all qubits in our subsystem $C \cup \hat{A}$.
By the same argument as in Case (1), $\mathcal{N}'_{\epsilon/12}$ defines a covering net over the set of pure quantum states on the subsystem $C \cup \hat{A}$ generated by $G$ two-qubit gates with respect to trace distance.
Moreover, $|\mathcal{N}'_{\epsilon/12}| \leq |\mathcal{N}_{\epsilon/12}|$.

Since this covering net $\mathcal{N}'_{\epsilon/12}$ is only for states on at most $2G$ qubits, let $\mathcal{N}_{\epsilon/12}''$ be the set of states where each state in $\mathcal{N}_{\epsilon/12}'$ is tensored with the zero state for qubits in $[n] \setminus (C \cup \hat{A})$.
Let $\tilde{\rho} = W_{\hat{B}} \rho W_{\hat{B}}$ be the original state on this permuted system.
By definition of a covering net, we know that there exists some $\sigma_i \in \mathcal{N}_{\epsilon/12}''$ such that
\begin{equation}
    \dtr(\tilde{\rho}, \sigma_i) \leq \frac{\epsilon}{12}.
\end{equation}
We justify this further in the following.
By definition, the only qubits in the state $\tilde{\rho}$ that are acted on nontrivially by the $G$ gates are those in $C \cup \hat{A}$. Since no gates act on qubits outside of $C \cup \hat{A}$, then the other qubits in $\tilde{\rho}$ must be in the zero state.
Hence, we can write $\tilde{\rho} = \tilde{\rho}_{C\cup \hat{A}} \otimes \ketbra{0}^{\otimes (n-|C\cup \hat{A}|)}$, where $\tilde{\rho}_{C \cup \hat{A}}$ denotes the state of the qubits in $C \cup \hat{A}$ which are acted upon by the $G$ gates.
Moreover, by definition of a covering net, then there exists some $\sigma_{i, C\cup \hat{A}} \in \mathcal{N}_{\epsilon/12}'$ such that
\begin{equation}
    \dtr(\tilde{\rho}_{C\cup \hat{A}}, \sigma_{i, C\cup \hat{A}}) \leq \frac{\epsilon}{12},
\end{equation}
where similarly $\sigma_{i, C \cup \hat{A}}$ is a state on the qubits in $C \cup \hat{A}$ which are acted upon by $G$ gates.
Taking the tensor product with the zero state on the remaining qubits does not affect the trace distance.
Thus, we can write $\sigma_i = \sigma_{i, C \cup \hat{A}} \otimes \ketbra{0}^{\otimes (n-|C\cup \hat{A}|)} \in \mathcal{N}_{\epsilon/12}''$, where this satisfies
\begin{equation}
    \dtr(\tilde{\rho},\sigma_i) = \dtr(\tilde{\rho}_{C\cup \hat{A}}, \sigma_{i, C\cup \hat{A}}) \leq \frac{\epsilon}{12},
\end{equation}
as claimed.
Moreover, by \Cref{lem:permute-orig}, choosing $\delta_2 = \delta/2$, we know that
\begin{equation}
    \dtr(\rho, \tilde{\rho}) \leq \frac{\epsilon}{12}
\end{equation}
with probability at least $1-\delta/2$.
Recall that this approximation requires only
\begin{equation}
    N_1 = \mathcal{O}\left(\frac{G + \log(1/\delta)}{\epsilon^2}\right)
\end{equation}
copies of $\rho$ (from \Cref{lem:postselect}) for identifying the set $\hat{A}$.
By triangle inequality, we have that there exists some $\sigma_i \in \mathcal{N}_{\epsilon/12}''$ such that
\begin{equation}
    \dtr(\rho,\sigma_i) \leq \frac{\epsilon}{6}
\end{equation}
with probability at least $1-\delta/2$.

Using hypothesis selection on the covering net $\mathcal{N}_{\epsilon/12}''$ and the unknown state $\rho$, by \Cref{prop:hypothesis}, there exists an algorithm to learn $\sigma$ such that
\begin{equation}
    \dtr(\rho, \sigma) \leq \epsilon
\end{equation}
with probability at least $1- \delta$, where we chose $\eta = \epsilon/6$ and  $\epsilon/2, \delta/2$ for the parameters in \Cref{prop:hypothesis}. Moreover, by \Cref{prop:hypothesis} and \Cref{eq:covering-size2}, this algorithm requires only
\begin{equation}
    N_2 = \mathcal{O}\left(\frac{G\log(G/\epsilon) + \log(1/\delta)}{\epsilon^2}\right)
\end{equation}
copies of $\rho$.
Putting everything together, we have that
\begin{equation}
    \dtr(\rho, \sigma) \leq \epsilon
\end{equation}
with probability at least $1-\delta$, where our algorithm to find $\sigma$ requires only
\begin{equation}
    N = N_1 + N_2 = \mathcal{O}\left(\frac{G\log(G/\epsilon) + \log(1/\delta)}{\epsilon^2}\right).
\end{equation}
This matches our upper bound for Case (1) and thus concludes the proof of \Cref{prop:states-upper}.
\end{proof}

\subsection{Sample complexity lower bound}
\label{app:states-lower}

In this section, we prove the sample complexity lower bound for \Cref{thm:state-learning-detail}.

\begin{prop}[State learning lower bound]
    \label{prop:states-lower}
    Let $\epsilon,\delta > 0$.
    Suppose we are given $N$ copies of an $n$-qubit pure state density matrix $\rho = \ketbra{\psi}$, where $\ket{\psi} = U\ket{0}^{\otimes n}$ is generated by a unitary $U$ consisting of $G$ two-qubit gates.
    Then, any algorithm that can output $\hat{\rho}$ such that $\dtr(\hat{\rho}, \rho) \leq \epsilon$ with probability at least $1-\delta$ requires at least
    \begin{equation}
        N = \Omega\left(\min\left(\frac{2^n}{\epsilon^2},\frac{G(1-\delta)}{\epsilon^2\log(G/\epsilon)}\right) + \frac{\log(1/\delta)}{\epsilon^2}\right)
    \end{equation}
    samples of $\ket{\psi}$.
\end{prop}

Here, similarly to the upper bound, we take the minimum with $\Omega(2^n/\epsilon^2)$, as this is the lower bound achieved for full quantum state tomography~\cite{haah2017sample,o2016efficient}.
We thus focus on the second term in the minimum.
We first consider the number of samples required to learn $n$-qubit pure quantum states generated by $G$ gates \emph{applied only to the first $\floor{\log_2 (G/C)}$ qubits} (for some constant $C$ specified later) of the $n \geq \floor{\log_2 (G/C)}$ qubits in total.
Denote this set of states as $S_1$.
Note that if $n \leq \floor{\log_2 (G/C)}$, then we can simply import the lower bound for full quantum state tomography~\cite{haah2017sample,o2016efficient}.
We later reduce the general case, where the $G$ gates can be applied on any of the qubits, to this case.
Namely, we prove the following proposition.

\begin{prop}
    \label{prop:states-lower-reduce}
    Let $\epsilon,\delta > 0$.
    Suppose we are given $N$ copies of an $n$-qubit pure state density matrix $\rho = \ketbra{\psi}$, where $\ket{\psi} = (U\otimes I)\ket{0}^{\otimes n} \in S_1$ is generated by a unitary $U$ consisting of $G$ two-qubit gates applied only to the first $\floor{\log_2 (G/C)}$ qubits for some constant $C$.
    Then, any algorithm that can output $\hat{\rho}$ such that $\dtr(\hat{\rho}, \rho) \leq \epsilon$ with probability at least $1-\delta$ requires at least
    \begin{equation}
        N = \Omega\left(\min\left(\frac{2^n}{\epsilon^2},\frac{G(1-\delta)}{\epsilon^2\log(G/\epsilon)}\right) + \frac{\log(1/\delta)}{\epsilon^2}\right)
    \end{equation}
    samples of $\ket{\psi}$.
\end{prop}

We note that for constant error $\epsilon$, the $\Omega(G/\log G)$ lower bound can be improved to $\Omega(G)$ using \cite{rosenthal2024efficient,yuen2023improved}.
We prove \Cref{prop:states-lower-reduce} by combining results from~\cite{haah2017sample,shende2005synthesis}.
Namely, the lower bound in~\cite{haah2017sample} works by lower bounding the sample complexity of learning any rank $r$, $d$-dimensional quantum state in terms of the packing number of this space of states.
We apply their results to our setting, where the space of states that the packing net is defined over is $S_1$ instead.
We first recall important results from~\cite{haah2017sample,shende2005synthesis} that we use throughout the proof.
\cite{haah2017sample} lower bounded the sample complexity of learning a $d$-dimensional pure state as follows.
\begin{theorem}[In Proof of Theorem 3 in~\cite{haah2017sample}]
    \label{thm:haah-lower}
    Let $\epsilon \in (0,1)$ and $\delta \in (0,1)$. Suppose there exists a POVM $\{M_\sigma d\sigma\}$ on $(\mathbb{C}^d)^{\otimes N}$ such that for a pure quantum state $\rho \in \mathbb{C}^{d\cross d}$, 
    \begin{equation}
        \int_{\dtr(\sigma, \rho) \leq \epsilon} d\sigma \Tr[M_\sigma \rho^{\otimes N}] \geq 1-\delta.
    \end{equation}
    Then
    \begin{equation}
    \label{eq:cover}
        N \geq \frac{(1-\delta)\ln m - \ln 2}{\chi_0},
    \end{equation}
    where $m$ is the size of an $(2\epsilon)$-packing net of the space of $d$-dimensional pure state density matrices, and 
    \begin{equation}
        \label{eq:chi0}
        \chi_0 \triangleq S(\mathbb{E}_U [U\rho_x U]) - S(\rho_x)
    \end{equation}
    is the Holevo information,    where $\rho_x$ is any element of the $(2\epsilon)$-packing net, $S$ is the von Neumann entropy, and the expectation is taken over the Haar measure.
\end{theorem}

This states that any measurement procedure which can identify a state $\rho$ up to $\epsilon$-trace distance requires at least $N$ copies of $\rho$, where $N$ is given by \Cref{eq:cover} and depends on the size of an $(2\epsilon)$-packing net of the space of $d$-dimensional pure state density matrices. 
Moreover,~\cite{haah2017sample} bounded the size of such a packing net.

\begin{lemma}[Lemma 5 in~\cite{haah2017sample}]
    \label{lem:packing}
    There exists an $\epsilon$-packing net $\{\rho_1,\dots, \rho_m\}$ of the space of $d$-dimensional pure state density matrices satisfying
    \begin{equation}
        c\ln m \geq d,
    \end{equation}
    for $c$ a sufficiently large constant and $d > 3$. This packing net also satisfies 
    \begin{equation}
        \frac{\chi_0}{c} \leq \epsilon^2 \ln\left(\frac{d}{\epsilon}\right)
    \end{equation} 
    for a sufficiently large constant $c > 0$, where $\chi_0$ is given by \Cref{eq:chi0}.
\end{lemma}

Finally, the last result we will need gives a bound on the number of gates needed to generate an arbitrary $n$-qubit pure state.

\begin{lemma}[Section 4 of \cite{shende2005synthesis}]
\label{lem:construct-state}
    Any $n$-qubit pure quantum state can be recursively defined as the result of a quantum circuit implemented by $\mathcal{O}(2^n)$ two-qubit gates applied to the $\ket{0}^{\otimes n}$ state.
    Explicitly, this quantum circuit has at most $C \cdot 2^n$ two-qubit gates for some constant $C$.
\end{lemma}

With these results, we can prove \Cref{prop:states-lower-reduce}.
The idea is that using \Cref{lem:construct-state}, any  pure state on the first $k\sim\log_2 G$ qubits can be generated by $G$ gates.
Then, we can use the same packing net construction as~\cite{haah2017sample} from \Cref{lem:packing}.
Plugging into \Cref{thm:haah-lower} then gives our lower bound.
We also add an additional term to account for expected asymptotic $\delta$ behavior.

\begin{proof}[Proof of \Cref{prop:states-lower-reduce}]
    We wish to construct a $(2\epsilon)$-packing net over the space $S_1$ of $n$-qubit pure quantum states generated by applying $G$ gates to the first $k = \floor{\log_2 (G/C)}$ qubits, where $C$ is taken to be the same constant as in \Cref{lem:construct-state}.
    First, consider only the subsystem consisting of the first $k$ qubits.
    Notice that by \Cref{lem:construct-state}, any $k$-qubit pure state can be generated by at most $G$ gates.
    Thus, the space of $k$-qubit pure states is the same as the space of $k$-qubit pure states generated by at most $G$ gates.
    In this way, we can construct a packing net for our subsystem of only the first $k$ qubits by constructing a packing net for all $k$-qubit pure states.
    By \Cref{thm:haah-lower}, there exists an $(2\epsilon)$-packing net $\mathcal{M}_{2\epsilon} = \{\sigma_1,\dots,\sigma_m\}$ of the space of $k$-qubit pure state density matrices satisfying
    \begin{equation}
        \label{eq:packing-size}
        \ln m \geq \frac{2^k}{c},\;\;\;\chi_0 \leq 4c\epsilon^2 \ln\left(\frac{2^{k-1}}{\epsilon}\right).
    \end{equation}
    From this, we can construct a packing net for our entire $n$-qubit system as follows.
    \begin{equation}
        \label{eq:packing-net-new}
        \mathcal{M}'_{2\epsilon} \triangleq \{ \sigma_i \otimes \ketbra{0}^{\otimes (n-k)}:\sigma_i \in \mathcal{M}_{2\epsilon}\};
    \end{equation}
    We claim that this is indeed a $(2\epsilon)$-packing net of $S_1$.
    Let $\ket{\psi} = (U \otimes I)\ket{0}^{\otimes n}\in S_1$ and let $\rho = \ketbra{\psi}$.
    Because $U$ only acts on the first $k$ qubits, then we can write $\rho = \rho_k \otimes \ketbra{0}^{\otimes (n -k)}$, where $\rho_k = U\ketbra{0}^{\otimes k} U$.
    Thus, we can see that $\mathcal{M}'_{2\epsilon} \subseteq S_1$.
    Importantly, all elements of $\mathcal{M}'_{2\epsilon}$ are $n$-qubit pure states generated by $G$ gates on the first $k$ qubits.
    Moreover, for any $\sigma_i', \sigma_j' \in \mathcal{M}'_{2\epsilon}$, we have
    \begin{equation}
        \dtr(\sigma_i', \sigma_j') = \dtr(\sigma_i \otimes \ketbra{0}^{\otimes (n-k)}, \sigma_j \otimes \ketbra{0}^{\otimes (n-k)}) = \dtr(\sigma_i, \sigma_j) > 2\epsilon,
    \end{equation}
    where the first equality follows by definition of $\mathcal{M}'_{2\epsilon}$ and the last inequality follows because $\sigma_i, \sigma_j \in \mathcal{M}_{2\epsilon}$.
    
    Hence, $\mathcal{M}'_{2\epsilon}$ is indeed a $(2\epsilon)$-packing net of $S_1$, which is the set of states we wish to learn.
    Moreover, it is of the same size as $\mathcal{M}_{2\epsilon}$, which had cardinality $m$ satisfying \Cref{eq:packing-size}.
    Plugging \Cref{eq:packing-size} into \Cref{thm:haah-lower}, we have that in order to learn $\rho$ up to $\epsilon$-trace distance, we require
    \begin{equation}
        N_1 \geq \frac{(1-\delta)\frac{2^k}{c} - \ln 2}{4c\epsilon^2 \ln(2^{k-1}/\epsilon)} \geq C_1\frac{(1-\delta)G - C_2}{\epsilon^2 \ln (G/(2\epsilon))} = \Omega\left(\frac{G(1-\delta)}{\epsilon^2\log(G/\epsilon)}\right),
    \end{equation}
    where in the second inequality, $C_1$ and $C_2$ are constants, where $C_1$ depends on $c$.

    This concludes the proof for the second term in the minimum in \Cref{prop:states-lower-reduce}.
    Again, for $n < \floor{\log_2 (G/C)}$, we can appeal to the full quantum state tomography lower bound of~\cite{haah2017sample,o2016efficient}.
    Thus, we obtain the lower bound
    \begin{equation}
        \label{eq:lower-part-1}
        N_1 = \Omega\left(\min \left(\frac{2^n}{\epsilon^2}, \frac{G(1-\delta)}{\epsilon
        ^2\log(G/\epsilon)}\right)\right).
    \end{equation}
    Notice, however, that in the limit as $\delta\to 0$ one should find $N\to\infty$. This behavior is not captured in \Cref{thm:haah-lower} due to the use of the classical Fano's inequality, which treats the measurement procedure as a classical random variable.
    This behavior is also not present in lower bounds from~\cite{haah2017sample,o2016efficient}, where they assume that $\delta = \Theta(1)$.
    In order to recover the dependence on $\delta$, we prove the following lemma.

    \begin{lemma}
        \label{lem:distinguish}
        Let $\ket{\psi_0}, \ket{\psi_1}$ be any two $n$-qubit pure quantum states.
        Suppose that $\ket{\psi_0}$ and $\ket{\psi_1}$ satisfy $\dtr(\ket{\psi_0}, \ket{\psi_1}) \geq \epsilon$.
        Then, for $\delta \in (0, 1]$,
        \begin{equation}
            N_2 = \Omega\left(\frac{\log(1/\delta)}{\epsilon^2}\right)
        \end{equation}
        copies of $\ket{\psi} \in \{\ket{\psi_0}, \ket{\psi_1}\}$ are needed to distinguish whether $\ket{\psi} = \ket{\psi_0}$ or $\ket{\psi} = \ket{\psi_1}$ with probability at least $1-\delta$.
    \end{lemma}
    
    \begin{proof}
    For pure states, we know that the relationship between fidelity and trace distance is given by
    \begin{equation}
        \dtr(\ket{\alpha}, \ket{\beta}) = \sqrt{1 - |\braket{\alpha}{\beta}|^2}.
    \end{equation}
    In our case, because $\dtr(\ket{\psi_0}, \ket{\psi_1}) \geq \epsilon$, then we have
    \begin{equation}
        \label{eq:fidelity-bound}
        |\braket{\psi_0}{\psi_1}|^2 \leq 1- \epsilon^2.
    \end{equation}
    Using the Holevo-Helstrom Theorem~\cite{Holevo1973, Helstrom1969}, in order to distinguish $\ket{\psi_0}$ from $\ket{\psi_1}$ with probability at least $1-\delta$, one requires at least $N_2$ copies of $\ket{\psi} \in \{\ket{\psi_0}, \ket{\psi_1}\}$ satisfying
    \begin{equation}
        1-\delta \leq \frac{1}{2} + \frac{1}{2}\sqrt{1-\lvert\braket{\psi_0}{\psi_1}\rvert^{2N_2}}
    \end{equation}
    Rearranging this inequality, we have
    \begin{equation}
        \label{eq:copies-bound}
        N_2 \geq \frac{\log(4\delta(1-\delta))}{\log(\lvert\braket{\psi_0}{\psi_1}\rvert^2)} = \frac{\log(\frac{1}{4\delta(1-\delta)})}{\log(\frac{1}{\lvert\braket{\psi_0}{\psi_1}\rvert^2})}
    \end{equation}
    By \Cref{eq:fidelity-bound}, this in particular requires
    \begin{equation}
        N_2 \geq \frac{\log(\frac{1}{4\delta(1-\delta)})}{\log(\frac{1}{1-\epsilon^2})} = \Omega\left(\frac{\log(1/\delta)}{\epsilon^2}\right).
    \end{equation}
    \end{proof} 

    In our case, note that the conditions of \Cref{lem:distinguish} hold by the existence of the packing net in \Cref{eq:packing-net-new}, where $\ket{\psi_0}, \ket{\psi_1}$ can be any two states in the packing net.
    Moreover, because approximating the unknown $\ket{\psi}$ to $(\epsilon/3)$-trace distance suffices to solve the distinguishing task in \Cref{lem:distinguish}, then this lower bound also applies for the task of learning a state $\ket{\psi}$.
    Thus, combining \Cref{lem:distinguish} with \Cref{eq:lower-part-1}, we have
    \begin{equation}
        N = \Omega\left(\max\left(N_1, \frac{\log(1/\delta)}{\epsilon^2} \right)\right) = \Omega\left(\min\left(\frac{2^n}{\epsilon^2},\frac{G(1-\delta)}{\epsilon^2\log(G/\epsilon)}\right) + \frac{\log(1/\delta)}{\epsilon^2}\right),
    \end{equation}
    as claimed.
\end{proof}

This concludes the proof of \Cref{prop:states-lower-reduce}.
Recall that we are seeking a sample complexity lower bound for states for which we allow our $G$ gates to act on any pair of the $n$ qubits rather than only the first $\floor{\log_2 (G/C)}$ qubits.
We complete the proof of \Cref{prop:states-lower} by reducing to the case of \Cref{prop:states-lower-reduce}.

\begin{proof}[Proof of \Cref{prop:states-lower}]
As before, denote the set of $n$-qubit quantum states generated by $G$ gates applied to only the first $\floor{\log_2 (G/C)}$ qubits as $S_1$.
Similarly, denote the set of $n$-qubit quantum states generated by $G$ gates (applied to any of the qubits) as $S_2$.
Our claim is that the sample complexity of learning states in $S_2$ is at least the sample complexity of learning states in $S_1$.

By \Cref{prop:states-lower-reduce}, we know that the sample complexity of learning states in $S_1$ is
\begin{equation}
    N = \Omega\left(\min\left(\frac{2^n}{\epsilon^2}, \frac{G(1-\delta)}{\epsilon^2\log(G/\epsilon)}\right) + \frac{\log(1/\delta)}{\epsilon^2}\right)
\end{equation}
By the definition of sample complexity, this means that there exists some state $\rho \in S_1$ requiring $N$ copies to learn within $\epsilon$ trace distance.
Then, because $S_1 \subseteq S_2$, then $\rho \in S_2$ as well.
Thus, there exists a state $\rho \in S_2$ that requires $N$ copies to learn, so the sample complexity of learning states within $S_2$ is at least $N$ as well.
\end{proof}

\subsection{Computational complexity}
\label{app:states-comp-complexity}

\Cref{thm:state-learning-detail} states that the sample complexity for learning a description of an unknown $n$-qubit pure quantum state is linear (up to logarithmic factors) in the number of gates $G$ used to generate the state.
Nevertheless, the algorithm described in \Cref{app:states-upper} is not computationally efficient, as it constructs and searches over an exponentially large (in $G$) covering net for all pure states generated by $G$ two-qubit gates.
This raises the question: Does there exist a computationally efficient algorithm?

In this section, we first show that there is no polynomial-time algorithm for learning states generated by $G = \mathcal{O}(n\polylog(n))$ gates, assuming $\mathsf{RingLWE}$ cannot be solved efficiently on a quantum computer.
This result also holds for states generated by a depth $d = \mathcal{O}(\polylog(n))$ circuit.
Then we invoke a stronger assumption that $\mathsf{RingLWE}$ cannot be solved by any sub-exponential-time quantum algorithm, and show that any quantum algorithm for learning states generated by $\tilde{\bigo}(G)$ gates must use $\exp(\Omega(G))$ time.
This means that the computational hardness already kicks in at $G=\tilde{\omega}(\log n)$.
Finally, we explicitly construct an efficient learning algorithm for $G=\mathcal{O}(\log n)$, thus establishing $\log n$ gate complexity as a transition point of computational efficiency.
Previous work~\cite{yang2023complexity,abbas2023quantum} arrives at similar hardness results for polynomial circuit complexity, but our detailed analysis allows us to sharpen the computational lower bound and obtain this transition point.

\begin{theorem}[State learning computational complexity lower bound assuming polynomial hardness of \textsf{RingLWE}]
\label{thm:states-comp-complexity}
Let $\lambda=n$ be the security parameter and $\mathcal{K}$ be the key space parametrized by $\lambda$.
Let $U$ be a unitary consisting of $G = \mathcal{O}(n\polylog(n))$ gates (or a depth $d = \mathcal{O}(\polylog(n))$ circuit) that prepares a pseudorandom quantum state $\ket{\phi_k}$ for some randomly chosen key $k \in \mathcal{K}$.
Such a unitary $U$ exists by \Cref{thm:PRS} assuming that \textsf{RingLWE} cannot be solved by polynomial-time quantum algorithms.
Suppose we are given $N = \poly(\lambda)$ copies of $\ket{\phi_k} = U\ket{0}^{\otimes n}$.
There does not exist a polynomial-time algorithm for learning a circuit description of $\ket{\phi_k}$ to within $\epsilon\leq 1/8$ trace distance with success probability at least $2/3$.
\end{theorem}

\begin{proof}
   Suppose for the sake of contradiction that there is an efficient algorithm $\mathcal{A}_0$ that can learn a description of $\ket{\phi_k}$ to within $\epsilon$ trace distance.
   Then by standard boosting of success probability (see e.g, \cite[Proposition 2.4]{haah2023query}), there is an efficient algorithm $\mathcal{A}$ that can learn $\ket{\phi_k}$ to the same accuracy with probability at least $p=1-1/128$ with only a constant factor overhead in time complexity.
   Note that this boosting requires the distance metric to be efficiently computable, which is guaranteed by the SWAP test elaborated below.
   We will construct a polynomial-time quantum distinguisher $\mathcal{D}$ that invokes $\mathcal{A}$ to distinguish between $\ket{\phi_k}$ and a Haar-random state $\ket{\phi}$.
   This contradicts \Cref{def:PRS}.

   The distinguisher $\mathcal{D}$ operates according to Algorithm~\ref{algo:distinguish}.
   \begin{algorithm}
        \label{algo:distinguish}
        \caption{Distinguisher $\mathcal{D}$ for PRS}
        \KwIn{$\rho^{\otimes N} = \ketbra{\psi}^{\otimes N}$}
        \KwOut{$b \in \{0,1\}$}
        Store one copy of $\rho$ in quantum memory.\\
        Run $\mathcal{A}$ on inputs $\rho^{\otimes (N-1)}$, receiving $\hat{\rho}$.\\
        Run the SWAP test on the remaining copy $\rho$ and $\hat{\rho}$, receiving a bit $b \in \{0,1\}$.\\
        Output $b$.
    \end{algorithm}

    Recall that the SWAP test~\cite{barenco1997stabilization,buhrman2001quantum} takes two quantum states $\sigma_1, \sigma_2$ as input and outputs $1$ with probability $(1 + \tr(\sigma_1 \sigma_2))/2$.
    We denote this algorithm as $\mathsf{SWAP}(\sigma_1, \sigma_2)$.
    Note that here we have switched the labels of $0$ and $1$ compared to the canonical SWAP test presented in~\cite{barenco1997stabilization,buhrman2001quantum}.

    Notice that the hypothetical efficient learner $\mathcal{A}$ always produces the circuit description of the output state $\hat{\rho}$ in polynomial time.
    This means that the circuit description and thus the state $\hat{\rho}$ must also be efficiently implementable.
    As the SWAP test is also efficient, Step 3 of Algorithm~\ref{algo:distinguish} can thus indeed be performed efficiently on a quantum computer.
    Hence, the distinguisher is indeed an efficient quantum algorithm.

    Throughout this section, we denote $\rho = \ketbra{\psi}$.
    We analyze the probability that the distinguisher $\mathcal{D}$ outputs $1$ when given the pseudorandom state $\ket{\phi_k}$ versus the Haar-random state $\ket{\phi}$.

    \textbf{Case 1: $\ket{\psi} = \ket{\phi_k}$}, for a randomly chosen $k \in \mathcal{K}$.
    We have $\rho = \ketbra{\psi} = \ketbra{\phi_k}$.
    By the guarantees of $\mathcal{A}$, with probability at least $p$, we have $\dtr(\hat{\rho}, \rho) \leq \epsilon$, where $\hat{\rho}$ is the (potentially mixed) quantum state learned by algorithm $\mathcal{A}$.
    We can rewrite this as
    \begin{equation}
        \label{eq:goodest}
        \expval{\hat{\rho}}{\psi} \geq 1-\epsilon
    \end{equation}
    where we used the relationship between fidelity and trace distance (when one state is pure)
    \begin{equation}
        \dtr(\rho, \hat{\rho}) \geq 1 - \expval{\hat{\rho}}{\psi}.
    \end{equation}
    Then it immediately follows from \Cref{eq:goodest} that
    \begin{equation}
    \begin{split}
        \label{eq:case1}
        \Pr_{\substack{k \leftarrow \mathcal{K}\\\mathcal{A},\mathsf{SWAP}}}\left[\mathcal{D}\left(\ket{\phi_k}^{\otimes N}\right) = 1\right] &= \Pr_{\substack{k \leftarrow \mathcal{K}\\\mathcal{A},\mathsf{SWAP}}}\left[\mathsf{SWAP}\left(\ketbra{\phi_k}, \hat{\rho}\right) = 1\right] \\
        &= \E_{k \leftarrow \mathcal{K}}\left[\Pr_{\mathcal{A}, \mathsf{SWAP}}\left[\mathsf{SWAP}(\ketbra{\phi_k}, \hat{\rho})=1|\ket{\phi_k}\right]\right] \\
        &\geq p\E_{k \leftarrow \mathcal{K}}\left[\frac{1}{2}+\frac{1}{2}(1-\epsilon)\right]
        = p\left(1-\frac{\epsilon}{2}\right),
    \end{split}
    \end{equation}
    where the probability is taken over the random choice of the key $k \in \mathcal{K}$, the randomness in the learning algorithm $\mathcal{A}$ when run on samples $\ket{\phi_k}^{\otimes N}$, and the randomness in the SWAP test.
    In the inequality, we split the probability into two terms conditioned on the success and failure of $\mathcal{A}$, and we lower bound the term conditioned on the failure of $\mathcal{A}$ by zero.

    \textbf{Case 2: $\ket{\psi} = \ket{\phi} \sim \mu$}, where $\mu$ is the Haar measure over pure quantum states.
    We have $\rho = \ketbra{\psi} = \ketbra{\phi}$.
    We want to upper bound the probability that the distinguisher $\mathcal{D}$ outputs $1$ when given copies of $\ket{\phi}$.
    The intuition is that a Haar-random state is likely to be far from any state generated by a circuits with a polynomial-sized description, the space in which output of $\mathcal{A}$ lie.
    Let $S_{\mathcal{A}}(\ket{\phi})$ be the set of quantum states corresponding to all possible outputs of the algorithm $\mathcal{A}$ when run on $N$ copies of $\ket{\phi}$.
    We follow a similar reasoning as in \Cref{eq:case1} and obtain
    \begin{align}
        \Pr_{\substack{\ket{\phi} \sim \mu\\\mathcal{A},\mathsf{SWAP}}}\left[\mathcal{D}\left(\ket{\phi}^{\otimes N}\right) = 1\right] 
        &\leq \operatornamewithlimits{\mathbb{E}}_{\ket{\phi} \sim \mu} \left[\max_{\hat{\rho} \in S_\mathcal{A}(\ket{\phi})} \left(\frac{1}{2} + \frac{1}{2}\expval{\hat{\rho}}{\phi} \right)\right] + (1-p)\\
        &= \frac{1}{2} + \frac{1}{2}\operatornamewithlimits{\mathbb{E}}_{\ket{\phi} \sim \mu}\left[\max_{\hat{\rho} \in S_\mathcal{A}(\ket{\phi})}\expval{\hat{\rho}}{\phi}\right] + (1-p)\\
        &\triangleq \frac{1}{2} + \frac{1}{2}\operatornamewithlimits{\mathbb{E}}_{\ket{\phi} \sim \mu}  [O_\phi] + (1-p),
    \end{align}
    where in the first line we split the probability according to whether $\mathcal{A}$ succeeds or fails, and we upper bound the failing term by $(1-p)$, and in the last line we define the random variable
    \begin{equation}
        O_\phi \triangleq \max_{\hat{\rho} \in S_\mathcal{A}(\ket{\phi})}\expval{\hat{\rho}}{\phi}.
    \end{equation}
    Furthermore, we can split $\operatornamewithlimits{\mathbb{E}}_{\ket{\phi} \sim \mu} [O_\phi]$ into two parts by introducing a cut-off $\theta$:
    \begin{equation}
        \operatornamewithlimits{\mathbb{E}}_{\ket{\phi} \sim \mu}  [O_\phi] \leq \Pr \left[O_\phi \leq 1 - \frac{\theta}{2}\right] \cdot \left(1 - \frac{\theta}{2}\right) + \Pr\left[O_\phi > 1 - \frac{\theta}{2}\right] \cdot 1 \leq 1 - \frac{\theta}{2} + \Pr\left[O_\phi > 1 - \frac{\theta}{2}\right],
    \end{equation}
    where in the first inequality, we used that $O_\phi \leq 1$.
    Plugging this into our previous expression, we have
    \begin{equation}
        \label{eq:distinguish-output-1}
        \Pr_{\substack{\ket{\phi} \sim \mu\\\mathcal{A},\mathsf{SWAP}}}\left[\mathcal{D}\left(\ket{\phi}^{\otimes N}\right) = 1\right] \leq 1 - \frac{\theta}{4} + \frac{1}{2}\Pr\left[O_\phi > 1 - \frac{\theta}{2}\right] + (1-p)
    \end{equation}
    We aim to upper bound the probability $\Pr[O_\phi > 1 - \theta/2]$.
    Notice that we have
    \begin{equation}
        \label{eq:ophi-bound}
        \Pr\left[O_\phi > 1 - \frac{\theta}{2}\right] \leq \sum_{\hat{\rho} \in \mathcal{N}_{\sqrt{\theta/2}}} \Pr_{\ket{\phi} \sim \mu}\left[\expval{\hat{\rho}}{\phi} > 1 - \frac{\theta}{2}\right],
    \end{equation}
    where $N_{\sqrt{\theta/2}}$ be a minimal $(\sqrt{\theta/2})$-covering net with respect to trace distance of the set $S_\mathcal{A}(\ket{\phi})$ of quantum states corresponding to all possible outputs of the algorithm $\mathcal{A}$ when run on $N$ copies of $\ket{\phi}$.
    We can bound this probability using concentration results. Let $d = 2^n$.
    \begin{align}
        \Pr_{\ket{\phi} \sim \mu}\left[\expval{\hat{\rho}}{\phi} > 1 - \frac{\theta}{2}\right] &\leq \Pr_{\ket{\phi} \sim \mu} \left[\exp\left(\frac{d}{2}\expval{\hat{\rho}}{\phi}\right) \geq \exp\left(\frac{d}{2}\left(1 - \frac{\theta}{2}\right)\right)\right]\\
        &\leq \exp\left(-\frac{d}{2}\left(1 - \frac{\theta}{2}\right)\right)\E_{\ket{\phi} \sim \mu}\left[\exp\left(\frac{d}{2}\expval{\hat{\rho}}{\phi}\right)\right]\\
        &= \exp\left(-\frac{d}{2}\left(1 - \frac{\theta}{2}\right)\right)\sum_{k=0}^\infty \frac{1}{k!}\frac{d^k}{2^k} \E_{\ket{\phi} \sim \mu}[\expval{\hat{\rho}}{\phi}^k]\\
        &= \exp\left(-\frac{d}{2}\left(1 - \frac{\theta}{2}\right)\right)\sum_{k=0}^\infty \frac{1}{k!}\frac{d^k}{2^k}\frac{1}{\binom{k+d-1}{k}}\tr(\hat{\rho}^{\otimes k}P_{\mathrm{sym}}^{(d,k)})\\
        &\leq \exp\left(-\frac{d}{2}\left(1 - \frac{\theta}{2}\right)\right)\sum_{k=0}^\infty \frac{1}{2^k}\tr(\hat{\rho}^{\otimes k}P_{\mathrm{sym}}^{(d,k)})\\
        &\leq 2\exp\left(-\frac{d}{2}\left(1 - \frac{\theta}{2}\right)\right).
    \end{align}
    Here, the first two lines follow from the following inequality, which holds for $\alpha > 0$ and a random variable $X$:
    \begin{equation}
        \Pr[X \geq \epsilon] \leq \Pr[\exp(\alpha X) \geq \exp(\alpha \epsilon)] \leq \exp(-\alpha X)\E[\exp(\alpha X)].
    \end{equation}
    The third line follows from the Taylor expansion of $\exp(x)$.
    The fourth line follow from the identity
    \begin{equation}
        \E_{\ket{\phi} \sim \mu} \expval{O}{\phi}^k = \frac{1}{\binom{k+d-1}{k}} \tr(O^{\otimes k} P_{\mathrm{sym}}^{(d,k)}),
    \end{equation}
    where we chose $O = \hat{\rho}$ and $P_{\mathrm{sym}}^{(d,k)}$ is the orthogonal projector onto the symmetric subspace of $(\mathbb{C}^d)^{\otimes k}$. See, e.g., Example 50 in~\cite{mele2023introduction} for a proof of this identity.
    The fifth line follows from the inequality $1/\binom{k+d-1}{k} \leq k!/d^k$.
    Finally, the last line is true by the following inequalities:
    \begin{align}
        \tr(\hat{\rho}^{\otimes k}P_{\mathrm{sym}}^{(d,k)}) &\leq \left\lvert\tr(\hat{\rho}^{\otimes k}P_{\mathrm{sym}}^{(d,k)})\right\rvert\\
        &\leq \norm{\hat{\rho}^{\otimes k}P_{\mathrm{sym}}^{(d,k)}}_1\\
        &\leq \norm{P_{\mathrm{sym}}^{(d,k)}}_\infty \norm{\hat{\rho}}_1^k\\
        &\leq 1,
    \end{align}
    which follows via properties of the trace norm and because $P_{\mathrm{sym}}^{(d,k)}$ is a projector.
    Plugging this back into \Cref{eq:ophi-bound}, we have
    \begin{align}
        \Pr\left[O_\phi > 1- \frac{\theta}{2}\right] &\leq \sum_{\hat{\rho} \in \mathcal{N}_{\sqrt{\theta/2}}} \Pr_{\ket{\phi} \sim \mu}\left[\expval{\hat{\rho}}{\phi} > 1 - \frac{\theta}{2}\right]\\
        &\leq 2\mathcal{N}(S_\mathcal{A}(\ket{\phi}),\dtr, \sqrt{\theta/2}) \exp\left(-\frac{2^n}{2}\left( 1- \frac{\theta}{2}\right)\right)
        \label{eq:O-phi-prob}
    \end{align}
    
    Moreover, since $S_\mathcal{A}(\ket{\phi})$ is the set of quantum states corresponding to all possible outputs of the algorithm $\mathcal{A}$ when run on $\ket{\phi}^{\otimes N}$, then all states in $S_\mathcal{A}(\ket{\phi})$ must have a $\poly(n)$-size circuit description (because $\mathcal{A}$ is assumed to be efficient).
    Thus our covering number upper bound (setting $G=\poly(n)$ in \Cref{app:states-upper}) implies
    \begin{equation}
    \label{eq:covering-poly-state}
        \mathcal{N}(S_\mathcal{A}(\ket{\phi}),\dtr,\sqrt{\theta/2}) = \mathcal{O}\left((1/\theta)^{\poly(n)}\right).
    \end{equation}
    Thus, the above bounds along with \Cref{eq:O-phi-prob} gives us,
    \begin{equation}
        \Pr\left[O_\phi > 1 - \frac{\theta}{2}\right] = \mathrm{negl}(n),
    \end{equation}
    where $\mathrm{negl}(n)$ denotes a negligible function in $n$. 
    Putting everything together with \Cref{eq:distinguish-output-1}, we have
    \begin{equation}
        \Pr_{\substack{\ket{\phi} \sim \mu\\\mathcal{A}, \mathsf{SWAP}}}\left[\mathcal{D}\left(\ket{\phi}^{\otimes N}\right) = 1\right] \leq 
        1 - \frac{\theta}{4} + \mathrm{negl}(n) + (1-p).
    \end{equation}
    Combining with \Cref{eq:case1}, we conclude that
    \begin{align}
        \left|\Pr_{\substack{k \leftarrow \mathcal{K}\\\mathcal{A}, \mathsf{SWAP}}} \left[\mathcal{D}\left(\ket{\phi_k}^{\otimes N}\right) =1\right] - \Pr_{\substack{\ket{\phi} \sim \mu\\\mathcal{A}, \mathsf{SWAP}}} \left[\mathcal{D}\left(\ket{\phi}^{\otimes N}\right) =1\right] \right| &\geq p\left(1-\frac{\epsilon}{2}\right)-2+ \frac{\theta}{4} +p - \mathrm{negl}(n)\\
        &\geq \frac{1}{16} - \mathrm{negl}(n) \\
        &\geq \frac{1}{32},
    \end{align}
    where we have taken $\theta = 1/2, \epsilon \leq 1/8, p=1-1/128$, and the last inequality follows by taking $n$ large enough.
    This contradicts the assumption that $\{\ket{\phi_k}\}_{k \leftarrow \mathcal{K}}$ are pseudorandom quantum states under the assumption that \textsf{RingLWE} cannot be solved by polynomial-time quantum algorithms.
\end{proof}

Next, we invoke the stronger assumption that \textsf{RingLWE} cannot be solved by any sub-exponential-time quantum algorithm and show that learning states generated by $\tilde{\bigo}(G)$ gates requires exponential-in-$G$ time.

\begin{theorem}[State learning computational complexity lower bound assuming sub-exponential hardness of \textsf{RingLWE}, restatement of lower bound in \Cref{thm:state-comp-complex}]
\label{thm:states-comp-complexity-subexp}
Let $\lambda=l=\Theta(G)$ with $l\leq n$ be the security parameter and $\mathcal{K}$ be the key space parametrized by $\lambda$.
Let $U$ be an $l$-qubit unitary consisting of $\mathcal{O}(l\polylog(l))=\bigo(G\polylog(G))$ gates (or a depth $d = \mathcal{O}(\polylog(G))$ circuit) that prepares an $l$-qubit pseudorandom quantum state $\ket{\phi_k}$ against sub-exponential adversaries for some randomly chosen key $k \in \mathcal{K}$.
Such a unitary $U$ exists by \Cref{thm:PRS} assuming that \textsf{RingLWE} cannot be solved by sub-exponential quantum algorithms.
Suppose we are given $N = \poly(\lambda)$ copies of $\ket{\psi_k} = \ket{\phi_k} \otimes \ket{0}^{\otimes (n-l)} = U\ket{0}^{\otimes n}$.
Any quantum algorithm for learning a circuit description of $\ket{\psi_k}$ to within $\epsilon\leq 1/8$ trace distance with success probability at least $2/3$ must use $\exp( \Omega(\min\{G, n\}))$ time.
\end{theorem}

\begin{proof}
    With polynomial hardness of \textsf{RingLWE} replaced by sub-exponential hardness, \Cref{thm:states-comp-complexity} asserts that there are no sub-exponential (in $l$) quantum algorithms that can learn the $l$-qubit pseudorandom state $\ket{\phi_k}$ to within trace distance $\epsilon<1/8$ with success probability at least $2/3$.
    That is, any such learning algorithms must use time at least $\exp(\Omega(l))=\exp(\Omega(\min\{G, n\}))$ time, since $l\leq n$.
    Meanwhile, a learning algorithm for the $n$-qubit state $\ket{\psi_k}$ can be used to learn the $l$-qubit state $\ket{\phi_k}$ in the same runtime by post-selecting on the last $(n-l)$ qubits being $\ket{0}$, because trace distance does not increase under such an operation.
    This implies the $\exp(\Omega(\min\{G, n\}))$ time lower bound for the $n$-qubit learning algorithm.
\end{proof}

Finally, we briefly show that learning becomes efficient when $G=\mathcal{O}(\log n)$.
The idea is that with $\bigo(\log n)$ gates, there can only be at most $\bigo(\log n)$ qubits affected.
Thus we can focus on these qubits and learning the states amounts to manipulating vectors of size at most $2^{\bigo(\log n)}=\poly(n)$, which is efficient.
Specifically, we have the following statement.

\begin{prop}[Learning states with logarithmic circuit complexity efficiently, restatement of upper bound in \Cref{thm:state-comp-complex}]
\label{prop:state-logn-efficiently}
    Let $\epsilon > 0$. Suppose we are given $N$ copies of a pure $n$-qubit state $\rho = \ketbra{\psi}$, where $\ket{\psi} = U\ket{0}^{\otimes n}$ is generated by a unitary $U$ consisting of $G=\mathcal{O}(\log n)$ two-qubit gates.
    There exists a learning algorithm that outputs a $\hat{\rho}$ such that $\dtr(\rho, \hat{\rho})\leq \epsilon$ with probability at least $2/3$ using $\poly(n, 1/\epsilon)$ copies and time.
\end{prop}

\begin{proof}
    We prove this by explicitly constructing a learning algorithm based on junta learning (\Cref{app:states-upper}) and standard tomography methods as follows.

    Firstly, we execute Algorithm~\ref{algo:nonzero-qubits} on copies of $\rho$ and post-select on the trivial qubits being zero as in \Cref{app:states-upper}. 
    This step uses $\poly(n, 1/\epsilon)$ copies and time, and gives us post-selected states $\rho' = \rho'' \otimes (\ketbra{0})^{\otimes(n-2G)}$ that satisfies $\dtr(\rho, \rho')\leq \epsilon/4$ by appropriate choice of accuracy.
    Here $\rho''$ is a state on $2G=\mathcal{O}(\log n)$ qubits.

    Next, we carry out the most straightforward tomography method of measuring all the Pauli coefficients.
    Concretely, we can represent $\rho''=\sum_P \alpha_P P$ as a linear combination of all Pauli strings over the $2G$ qubits.
    Using this representation, we estimate all the coefficients $\alpha_P$ by measuring $\tr(\rho' P)$ and obtain a $\hat{\rho}=\hat{\rho}''\otimes (\ketbra{0})^{\otimes(n-2G)}$.
    By measuring all Pauli string expectation values $\tr(\rho' P)$ to accuracy $\mathcal{O}(\epsilon/4^{2G})$, we have $\dtr(\rho', \hat{\rho})\leq \epsilon/4$ and thus $\dtr(\rho, \hat{\rho})\leq \epsilon/2$.
    From standard Chernoff-Hoeffding concentration inequalities, this can be achieved with $\bigo(4^{2G}/(\epsilon/4^{2G})^2)=\poly(n, 1/\epsilon)$ copies.
    Finally, we diagonalize $\hat{\rho}''$ and calculate its eigenvector $\hat{\ket{\psi''}}$ with the largest eigenvalue, such that $\hat{\ket{\psi''}}$ is the pure state closest to $\hat{\rho}''$ in trace distance.
    Let $\hat{\ket{\psi}} = \hat{\ket{\psi''}}\otimes\ket{0}^{\otimes(n-2G)}$.
    Recall that $\dtr(\rho, \hat{\rho})\leq \epsilon/2$ and $\rho$ is a pure state.
    Therefore, $\dtr(\hat{\ket{\psi}}\hat{\bra{\psi}}, \hat{\rho})\leq \dtr(\rho, \hat{\rho})\leq \epsilon/2$ and thus $\dtr(\hat{\ket{\psi}}\hat{\bra{\psi}}, \rho)\leq \epsilon$.
    We output $\hat{\ket{\psi}}$ as the learning outcome whose circuit description can be found by finding a unitary with $\hat{\ket{\psi''}}$ as its first column using orthogonalization.
    Since we are manipulating matrices of size $\bigo(2^{2G})=\poly(n)$, the computational complexity is also $\mathcal{O}(n, 1/\epsilon)$.
\end{proof}

\section{Learning quantum unitaries}

In this appendix, we give detailed proofs of \Cref{thm:worst-case-unitary} for worst-case unitary learning, \Cref{thm:avg-case-unitary} for average-case unitary learning, and \Cref{thm:classical-description} for learning with classically described data.

\subsection{Worst-case learning}
\label{app:worst-case-unitary}
We begin with the worst-case unitary learning problem, which measures reconstruction error in terms of the diamond distance $\dworst(U, V) = \max_{\rho} \|(U\otimes I)\rho(U\otimes I)^\dagger - (V\otimes I)\rho(V\otimes I)^\dagger\|_1$.
In particular, we consider the task of using queries to an unknown unitary $U$ with bounded circuit complexity $G$ to output a classical circuit description $\hat{U}$ such that $\dworst(\hat{U}, U) \leq \epsilon$ with probability at least $2/3$.
The diamond distance has a similar operational meaning as the trace distance in state learning.
It characterizes the ability to distinguish two processes with arbitrary input states and measurements.
If we can learn the unitary with small error in the diamond distance, then we will only make small error even if we test $\hat{U}$ against $U$ on the worst choice of input states.
However, we find the following result stating that this task necessarily requires a number of queries exponential in $G$, indicating the hardness of worst-case unitary learning.

\begin{theorem}[Worst-case unitary learning, restatement of \Cref{thm:worst-case-unitary}]
Given query access to an $n$-qubit unitary $U$ composed of $G$ two-qubit gates, any algorithm that can output a unitary $\hat{U}$ such that $\dworst(\hat{U}, U)\leq \epsilon\in (0, 1/4]$ with probability at least $2/3$ must query $U$ at least $\Omega\left(2^{\min\{G/(2C), n/2\}}/{\epsilon}\right)$ times, where $C>0$ is a universal constant. 
Meanwhile, there exists such an algorithm using $\mathcal{O}(2^nG\log(\sqrt{2^n}G/\epsilon)/\epsilon)$ queries. 
\end{theorem}

\begin{proof}
    The upper bound follows from the average-case learning algorithm (\Cref{thm:avg-case-unitary}, proved below) when working in the exponentially small error regime. 
    Specifically, \Cref{thm:avg-case-unitary} gives us an algorithm that uses $\bigo(G\sqrt{d}\log(G/\epsilon')/\epsilon')$ queries to output a $\hat{U}$ that satisfies $\davg(\hat{U}, U)\leq \epsilon'$.
    Meanwhile, from \Cref{lem:dist-spectral-diamond}, \Cref{lem:dist-norm-conversion} and \Cref{lem:dist-df'}, we know that $\dworst(\hat{U}, U)\leq 2d_2'(\hat{U}, U)\leq 2\sqrt{d}d_F'(\hat{U}, U)\leq 4\sqrt{d}\davg\leq 4\sqrt{d}\epsilon'$. 
    Setting $\epsilon=4\sqrt{d}\epsilon'$, we arrive at the desired worst-case learning query complexity.
    
    The proof of the lower bound is inspired by the adversary method~\cite[Chapter 6]{van2019quantum} and the optimality of Grover's algorithm~\cite{boyer1998tight}. 
    The idea is to construct a set of unitaries that can be distinguished by the worst-case learning algorithm, but only make minor difference when acting on states so that a minimal number of queries have to be made in order to distinguish them.

    Specifically, we consider all the length-$2^k$ bit-strings $x$ that have Hamming weight $1$, i.e.,  $x_i=1$ for some $i\in[2^k]$ and all the other bits are $0$. 
    We focus on the task of distinguishing this set of strings, denoted by $X$, from the all zero string $Y=\{0\ldots0\}$. 
    We access any such bit-strings $x$ through a phase oracle, which is defined as a $k$-qubit unitary $U_x$ that obeys $U_x\ket{j} = e^{i\epsilon' x_j}\ket{j}$ for all $j\in [2^k]$. 
    In other words, $U_x$ is diagonal and each diagonal element is $e^{i\epsilon'}$ if the corresponding bit is $1$ and is $1$ if the bit is $0$. 
    The unitary for the all zero string is the identity. 
    
    To implement such unitaries with $2$-qubit gates, we note that since the strings have Hamming weight at most one, each of the unitaries is equivalent to a $(k-1)$-controlled phase gate with proper control rule. 
    The control rule can be realized by $\bigo(k)$ pairs of 1-qubit gates acting on each qubit, and the $(k-1)$-controlled phase gate can be decomposed into $\bigo(k)$ 2-qubit gates~\cite{gidney2015constructing}. 
    Therefore, with $\bigo(k)$ gates, one can implement $U_x$ for any $2^k$-bit string $x$ with Hamming weight at most one. 
    
    Suppose $Ck$ gates suffice to implement these $U_x$. Set $k = \min\{\floor{G/C}, n\}$. 
    Then for any $x\in X\cup Y$, $U_x \otimes I_{n-k}$ is an $n$-qubit gate composed of at most $G$ gates. 
    Meanwhile, the unitaries for $X$ are far apart from that for $Y$, because for any $x\in X$, suppose $x_j=1$, we can take another $x'\neq x$ from $X$ with $x'_{j'}=1$, and let $\ket{\psi_{jj'}} = (\ket{j}+\ket{j'})/\sqrt{2}$. 
    Then we have 
    \begin{equation}
    \begin{split}
        \dworst (U_x, U_{0\ldots 0}) &\geq  \|U_x\ket{\psi_{jj'}}\bra{\psi_{jj'}}U_x^\dagger - U_{0\ldots 0}\ket{\psi_{jj'}}\bra{\psi_{jj'}}U_{0\ldots 0}^\dagger\|_1 \\
        &= \left\|\frac{e^{i\epsilon'}-1}{2}\ket{j}\bra{j'} + \frac{e^{-i\epsilon'}-1}{2}\ket{j'}\bra{j}\right\|_1 = 2\sin\frac{\epsilon'}{2} \geq \frac{\epsilon'}{2},
    \end{split}
    \end{equation}
    for $\epsilon'\in (0, 1]$.
    Therefore, if we have a learning algorithm that can learn $U_T^n$ using $m$ queries with accuracy $\epsilon=\epsilon'/4\in(0, 1/4]$ in diamond norm with probability $2/3$, it can also distinguish $X$ from $Y$ with the same probability. 
    Note that this also works if the learning algorithm is for (quotient) spectral distance, but not for $\davg$ because $\davg(U_{x}, U_{0\ldots 0})$ is exponentially small for every $x$ with Hamming weight one.

    In addition, we have the following query complexity lower bound from the adversary method.
    \begin{lemma}
    (Phase adversary method,~\cite[Lemma 6.4]{van2019quantum}).
    Let $D$ be a finite set of functions from a finite set $Q$ to $\mathbb{R}$. 
    To each function $x\in D$, assign an oracle $U_x$ of the form $U_x\ket{q}=e^{ix(q)}\ket{q}$. 
    Let $X$ and $Y$ be two disjoint subsets of $D$. Let $R\subseteq X\times Y$ be a binary relation on $X\times Y$. 
    For $x\in X$, we write $R(x)=\{y\in Y: (x, y)\in R\}$, and similarly $R(y)$ for $y\in Y$. Define
    \begin{align*}
        m = \min_{x\in X}|R(x)|, \quad m' = \min_{y\in Y}|R(y)|, \quad
        l_{q, x} = \sum_{y\in R(x)}|x(q)-y(q)|, \quad l_{q, y} = \sum_{x\in R(y)}|x(q)-y(q)|,
    \end{align*}
    and let $l_{\max} = \max_{q\in Q, x\in X, y\in Y}l_{q, x}l_{q, y}$. 
    Then to distinguish $X$ and $Y$ with success probability at least $2/3$, any algorithm needs at least 
    \begin{equation}
        \Omega\left(\sqrt{\frac{mm'}{l_{\max}}}\right)
    \end{equation}
    queries to the oracle.
    \end{lemma}
    
    For our problem, let $R=X\times Y$. For all bit-strings $x$, define $x(q)=\epsilon x_q$. 
    Then we have $m=|Y|=1, m'=|X|=2^k, l_{q, x}=\epsilon x_q, l_{q, y}=\epsilon$ because for a specific $q$, only one $x\in R(y)=X$ has $x_q=1$. Thus $l_{\max} = \epsilon^2$. 
    Plugging these into the above lemma, we obtain a query complexity lower bound of $\Omega(\sqrt{2^k}/\epsilon)$. 
    Since $k=\min\{\floor{G/C}, n\}$, we arrive at the final query complexity lower bound $\Omega\left(2^{\min\{G/(2C), n/2\}}/{\epsilon}\right)$.
\end{proof}

\subsection{Average-case query complexity upper bounds}
\label{app:avg-case-unitary-up}

Having seen that worst-case unitary learning is hard, we move on to the setting of average-case learning.
In particular, we consider the task of using queries to an unknown unitary $U$ with bounded circuit complexity $G$ to output the classical circuit description of a unitary $\hat{U}$ such that $\davg(\hat{U}, U) = \sqrt{\E_{\ket{\psi}}[\dtr(\hat{U}\ket{\psi}, U\ket{\psi})^2]} \leq \epsilon$ with probability at least $2/3$.
In the following, we give explicit algorithms that solve this learning task with linear-in-$G$ queries, using similar hypothesis selection techniques as in the state learning task (\Cref{app:states-upper}).

\begin{prop}[Average case unitary learning upper bounds, upper bounds in \Cref{thm:avg-case-unitary}]
There exists an algorithm that, given query access to an $n$-qubit unitary $U$ composed of $G$ two-qubit gates, can output a unitary $\hat{U}$ such that $\davg(\hat{U}, U)\leq \epsilon$ with probability at least $2/3$ using 
\begin{equation}
    \bigo\left(\min\left\{\frac{4^n}{\epsilon}, \frac{G\log(G/\epsilon)}{\epsilon^2}, \frac{\sqrt{2^n}G\log(G/\epsilon)}{\epsilon}\right\}\right)
\end{equation}
queries to the unknown unitary $U$. 
Moreover, there is another such algorithm that uses $\bigo(G\log(G/\epsilon)/\epsilon^4)$ queries without employing auxiliary quantum systems.
\label{prop:upper-bound-unitary}
\end{prop}

The $\bigo(4^n/\epsilon)$ scaling comes from the diamond norm learning algorithm in \cite[Theorem 1.1]{haah2023query}, which directly implies an average-case learning algorithm because $\davg(U, V)\leq d_F'(U, V)\leq d_2'(U, V)\leq \frac{1}{\sqrt{2}}\dworst(U, V)$ from \Cref{lem:dist-df',lem:dist-norm-conversion,lem:dist-spectral-diamond}. 
Note that this part of the bound does not make use of the promise that the unknown unitary can be implemented with $G$ two-qubit gates.
In the following, we prove the $G$-dependent parts of the upper bound.

\subsubsection{Unitary learning without ancillary systems}\label{app:unitary-no-ancilla}

We begin by describing the learning algorithm without ancillary systems.
The algorithm works similarly to the state learning procedure.
It constructs a covering net over $G$-gate unitaries with respect to $\davg$, and regards them as candidates for the unknown unitary.
In contrast to our state learning procedure, where the algorithm estimates the trace distance between states, here the algorithm estimates the overlap between unitaries by inputting random states and apply single-shot Clifford classical shadow, which translates into $\davg$.
Then, we select the candidate closest to the unknown unitary as the learning outcome.

Specifically, we consider a $\sqrt{\epsilon'}$-covering net $\mathcal{N}$ of the set of $n$-qubit unitaries implemented by $G$ two-qubit gates with respect to $\davg$ as in \Cref{cor:covering-avg}, and regard the elements $U_i\in\mathcal{N}$ as potential candidates for the unknown unitary $U$. 
Our strategy is to use classical shadow to estimate the distances $\davg(U_i, U)$ for every $U_i$ in the covering net. 
Then we can find the one with minimal distance as the output of our learning algorithm.

To achieve this, consider a randomly sampled tensor product of 1-qubit stabilizer states 
\begin{equation}
    \ket{x} = U_x \ket{0}^{\otimes n} \sim Q = \mathrm{Uniform}[\{\ket{0}, \ket{1}, \ket{x+}, \ket{x-}, \ket{y+}, \ket{y-}\}^{\otimes n}],
\end{equation}
where $U_x = \otimes_{i=1}^n U_{x_i}$ is the state preparation unitary, and $x\in \mathbb{Z}_6^n$ labels the state.
We apply the unknown unitary $U$ to it and obtain $U\ket{x}$. 
Then we invoke a single use of the Clifford classical shadow protocol \cite{huang2020predicting}: We randomly sample an $n$-qubit Clifford gate $C$ and apply it to $U\ket{x}$, and then measure in the computational basis to get an outcome $\ket{b}, b\in\{0, 1\}^n$, with probability $|\braket{b}{CU|x}|^2$. 
Let $\hat{\rho} = (2^n+1)C^\dagger\ket{b}\bra{b}C - I$. From \cite{huang2020predicting}, we know that $\mathbb{E}_{C, b} [\hat{\rho}] = U\ket{x}\bra{x}U^\dagger$. Now we consider the observable $O_{i}=U_i\ket{x}\bra{x}U_i^\dagger$ and the estimator $\hat{o}_{i} = \tr(O_{i} \hat{\rho})$. Then we have the expectation value
\begin{equation}
    \underset{\ket{x}, C, b}{\mathbb{E}}[\hat{o}_i] 
    = \underset{\ket{x}}{\mathbb{E}}\left[\tr(O_i \underset{C, b}{\mathbb{E}}[\hat{\rho}]) \right]
    = \underset{\ket{x}}{\mathbb{E}} \left[|\matrixelement{x}{U_i^\dagger U}{x}|^2 \right]
    = 1-d_Q^2(U_i, U)\, ,
\end{equation}
where $d_Q(U_i, U) = \sqrt{\E_{\ket{\psi}\sim Q}[\dtr(U_i\ket{\psi}, U\ket{\psi})^2]}$ is the root mean squared trace distance with respect to $Q$ as defined in \Cref{lem:equi_local_scram}.
Next, we show that $\hat{o}_i$ has bounded variance. Note that
\begin{equation}
    \mathrm{Var}[\hat{o}_i] 
    = \underset{\ket{x}, C, b}{\mathbb{E}}[\hat{o}_i^2] - \left(\underset{\ket{x}, C, b}{\mathbb{E}}[\hat{o}_i]\right)^2 
    \leq \underset{\ket{x}}{\mathbb{E}}\left[\underset{C, b}{\mathbb{E}}[\hat{o}_i^2]\right] 
    \leq \underset{\ket{x}}{\mathbb{E}}[3\tr(O_i^2)]
    =3,
\end{equation}
where we have used the variance bound for Clifford shadows \cite[Lemma S1 and Proposition S1]{huang2020predicting} and the fact that $\tr(O_i^2)=\tr(O_i)=1$. 

To estimate the expectation values of $\hat{o}_i$, we can draw $m$ i.i.d. samples of such input states $\{\ket{x_j}\}_{j=1}^m$ from $Q$, construct the observables $O_{ij} = U_i\ket{x_j}\bra{x_j}U_i^\dagger$ and carry out the above protocol to get the estimators $\hat{o}_{ij}$ for $1\leq i\leq |\mathcal{N}|, 1\leq j\leq m$. 
Suppose we take $m=NK$ and construct a median-of-mean estimator
\begin{equation}
    \hat{o}_i(N, K) = \mathrm{median}\{\hat{o}_i^{(1)}, \ldots, \hat{o}_i^{(K)}\}, \quad \text{where} \quad \hat{o}_i^{(k)} = \frac{1}{N}\sum_{j=N(k-1)+1}^{Nk}\hat{o}_{ij}, \quad 1\leq k\leq K.
\end{equation}
Then, with the same reasoning as in \cite[Theorem S1]{huang2020predicting}, we have the following concentration guarantee: For any $0<\epsilon', \delta<1$, if $K = 2\log (2|\mathcal{N}|/\delta)$ and $N = 102/\epsilon'^2$, then 
\begin{equation}
    |\hat{o}_i(N, K) - (1-d_Q^2(U_i, U))| \leq \epsilon' \quad \text{for all} \quad 1\leq i\leq |\mathcal{N}|
\end{equation}
with probability at least $1-\delta$.

With $\hat{o}_i$ in hand, we can select $i^\star \in \mathrm{argmax}_i \hat{o}_i$, and output $U_{i^\star}$. 
Then we have
\begin{equation}
\begin{split}
    \davg(U_{i^\star}, U) &\leq \sqrt{2}d_Q(U_{i^\star}, U) \leq \sqrt{2(1 - \hat{o}_{i^\star} + \epsilon')} \\
    &= \sqrt{2(\epsilon' + \min_{i}(1-\hat{o}_i))}
    \leq \sqrt{2(\epsilon' + \min_i (d_Q^2(U_i, U)+\epsilon'))} \leq \sqrt{8\epsilon'}
\end{split}
\end{equation}
with probability at least $1-\delta$, where we have used the concentration guarantee, \Cref{lem:equi_local_scram}, and $\min_i d_Q^2(U_i, U) \leq \min_i 2\davg(U_i, U)^2 \leq 2\epsilon'$ because $\mathcal{N}$ is a $\sqrt{\epsilon'}$-covering net with respect to $\davg$. 
Setting $\epsilon'=\epsilon^2/8$, we arrive at a learning algorithm that uses
\begin{equation}
    m=NK=\bigo(\log(|\mathcal{N}|/\delta)/\epsilon^4)
\end{equation}
samples to learn the unknown unitary with accuracy $\epsilon$ and success probability at least $1-\delta$. 

If we plug in the covering number upper bound $\log\mathcal{N}\leq \bigo(G\log(G/\epsilon) + T\log n)$ from \Cref{cor:covering-avg}, we have sample complexity 
\begin{equation}
    \bigo\left(\frac{G\log(G/\epsilon)+\log(1/\delta)}{\epsilon^4}\right)
\end{equation} for large $G$, say $G\geq n/10$, as desired.

For $G<n/10$, a direct application of the above strategy will give us a suboptimal sample complexity of $\bigo(G\log (n/\epsilon)/\epsilon^4)$. 
To overcome this issue, we can carry out a junta learning step similar to Algorithm~\ref{algo:nonzero-qubits} and \cite{chen2023testing} to identify the subset of qubits $A\subset [n]$ that $U$ acts non-trivially on. Since $U$ only has $G$ $2$-qubit gates, we must have $|A| \leq 2G$. The specific procedure is listed in Algorithm~\ref{algo:junta-unitary}.

\begin{algorithm}
\label{algo:junta-unitary}
\caption{Identify qubits acted upon nontrivially (unitary version)}
\KwIn{Query access to the unknown unitary $U$ with $G$ two-qubit gates.}
\KwOut{List $\hat{A} \subseteq [n]$ of qubits.}
Initialize $\hat{A} = \emptyset$.\\
Repeat the following $N = \mathcal{O}\left(\frac{G + \log (1/\delta)}{\epsilon^2}\right)$ times: \linebreak
(a) Sample a random tensor product of 1-qubit stabilizer states $\ket{x}=U_x\ket{0}^{\otimes n}$, apply $U$ and $U_x^\dagger$, and obtain $U_x^\dagger U U_x\ket{0}^{\otimes n}$
\linebreak 
(b) Measure in the computational basis and obtain a bit string $\ket{b}, b\in\{0, 1\}^n$
\linebreak 
(b) Given the measurement outcome $\ket{b}$, set $\hat{A} \leftarrow \hat{A} \cup \mathrm{supp}(b)$, where $\mathrm{supp}(b) = \{i \in [n] : b_i \neq 0\}$.
\end{algorithm}

Similar to \Cref{app:states-upper-post}, we use Algorithm~\ref{algo:junta-unitary} to identify the non-trivial qubits with high probability.
Importantly, from \Cref{lemma:good-overlap}, we have the following guarantee that shows the expected state on the estimated trivial qubits is close to zero.

\begin{lemma}
    \label{lemma:good-overlap-unitary}
    Let $\epsilon, \delta > 0$. 
    Suppose we are given query access to an $n$-qubit unitary $U$ composed of $G$ two-qubit gates acting on a subset of the qubits $A \subseteq [n]$. 
    Let $\ket{x}=U_x\ket{0}^{\otimes n}$ be a random tensor product of $1$-qubit stabilizer states.
    Let $\rho^x = U_x^\dagger UU_x\ketbra{0}U_x^\dagger U^\dagger U_x$.
    Then, Algorithm~\ref{algo:junta-unitary} uses $N = \mathcal{O}\left(\frac{G + \log(1/\delta)}{\epsilon^2}\right)$ queries to $U$ and outputs, with probability at least $1-\delta$, a list $\hat{A} \subset [n]$ such that
    \begin{equation}
        \expval{\E_{x}[\rho^x_{\hat{B}}]}{0_{\hat{B}}} \geq 1 - \epsilon^2,
    \end{equation}
    where $\rho_{\hat{B}}$ denotes the reduced density matrix of $\rho$ when tracing out all qubits other than those in the set $\hat{B} = [n] \setminus \hat{A}$ and $\ket{0_{\hat{B}}}$ denotes the zero state on all qubits in ${\hat{B}}$.
\end{lemma}
\begin{proof}
    This follows directly from the proof of \Cref{lemma:good-overlap} because Algorithm~\ref{algo:junta-unitary} is the same as executing Algorithm~\ref{algo:nonzero-qubits} on the mixed state $\E_x [\rho^x]$, and for the trivial qubits, the $U_x^\dagger$ following $U_x$ and $U$ restores the state to $\ket{0}$. So the proof goes verbatim as in \Cref{lemma:good-overlap}.
\end{proof}

With this, we can show that ignoring the rest of the qubits ${\hat{B}}=[n]\setminus \hat{A}$ does not make much of a difference. 
Let $B=[n]\setminus A$. 
We again consider a randomly sampled $1$-qubit stabilizer state and apply $U$ to get $\ket{\psi_x}=UU_x\ket{0}^{\otimes n}$. Let $\rho_x = \ket{\psi_x}\bra{\psi_x}$ be the associated density matrix, and let $U_x^{\hat{B}} = \otimes_{j\in {\hat{B}}}U_{x_j}$ be the part of $U_x$ that acts on ${\hat{B}}$. 
Now we measure the qubits in ${\hat{B}}$ in the basis $U_x^{\hat{B}}\ket{b}_{\hat{B}}$, where $b\in\{0, 1\}^{|{\hat{B}}|}$. 
Note that for qubits in ${\hat{B}}$, the reduced density matrix in the basis $U_x^{\hat{B}}\ket{b}_{\hat{B}}$ is the same as the $\rho_{{\hat{B}}}^x$ from \Cref{lemma:good-overlap-unitary} in the junta learning step. 
So we have $\expval{\E_{x}[\rho^x_{\hat{B}}]}{0_{\hat{B}}} \geq 1 - \epsilon^2$. 
After the measurement of the qubits in ${\hat{B}}$, we do a post-selection on the observed measurement outcomes being $U_x^{\hat{B}}\ket{0}_{\hat{B}}$. 
This post-selection is represented by $\Lambda = I_A\otimes (U_x^{\hat{B}}\ket{0}_{\hat{B}}\bra{0}U_x^{{\hat{B}}\dagger})$, with $\Lambda^2 = \Lambda$. 
Let $\rho_x' = \frac{\sqrt{\Lambda}\rho_x\sqrt{\Lambda}}{\tr(\Lambda \rho_x)}$ be the post-selected state. 
Now we want to show $\rho'_x$ is close to $\rho_x$ on average. 
We invoke the following gentle measurement lemma for normalized ensembles.
\begin{lemma}[Gentle measurement lemma for normalized ensembles, variant of {\cite[Lemma 9.4.3]{wilde2013quantum}}]
    Let $\{x, \rho_x\}$ be an ensemble of states. If $\Lambda$ is a positive semi-definite operator with $\Lambda\leq I$ and $\tr(\Lambda \mathbb{E}_x [\rho_x])\geq 1-\epsilon$ where $\epsilon \in [0, 1]$, then
\begin{equation}
    \mathbb{E}_x \left\|\rho_x - \frac{\sqrt{\Lambda}\rho_x\sqrt{\Lambda}}{\tr(\Lambda \rho_x)}\right\|_1 
    \leq 3\sqrt{\epsilon}.
\end{equation}
\label{lem:gentle}
\end{lemma}
\begin{proof}
    Let $\rho_x' = \frac{\sqrt{\Lambda}\rho_x\sqrt{\Lambda}}{\tr(\Lambda \rho_x)}$.
    From \cite[Lemma 9.4.3]{wilde2013quantum}, we know that $\mathbb{E}_x \left\|\rho_x - \sqrt{\Lambda}\rho_x\sqrt{\Lambda}\right\|_1 \leq 2\sqrt{\epsilon}$. Note that the left hand side can be lower bounded by
    \begin{equation}
    \begin{split}
        \mathbb{E}_x \left\|\rho_x - \sqrt{\Lambda}\rho_x\sqrt{\Lambda}\right\|_1 &= \mathbb{E}_x \left\|\rho_x - \rho_x' + \rho_x' - \sqrt{\Lambda}\rho_x\sqrt{\Lambda}\right\|_1 \\
        &\geq \mathbb{E}_x \|\rho_x - \rho_x'\|_1 - \mathbb{E}_x \|\rho_x' - \sqrt{\Lambda}\rho_x\sqrt{\Lambda}\|_1 \\
        &= \mathbb{E}_x \|\rho_x - \rho_x'\|_1 - \mathbb{E}_x (1-\tr(\Lambda \rho_x)) \|\rho_x'\|_1 
        \geq \mathbb{E}_x \|\rho_x - \rho_x'\|_1 - \epsilon,
    \end{split}
    \end{equation}
    where we have used triangle inequality, $\|\rho_x'\|_1=1$, and $\tr(\Lambda \mathbb{E}_x [\rho_x])\geq 1-\epsilon$.
    Therefore, we arrive at
    \begin{equation}
        \mathbb{E}_x \|\rho_x - \rho_x'\|_1 \leq 2\sqrt{\epsilon}+\epsilon \leq 3\sqrt{\epsilon},
    \end{equation}
    because $\epsilon \in [0, 1]$, concluding the proof of \Cref{lem:gentle}.    
\end{proof}

Using \Cref{lem:gentle} for our scenario, we have $\mathbb{E}_x\|\rho_x-\rho_x'\|_1\leq 3\epsilon$ with probability at least $1-\delta$.
After the post-selection, we apply the same Clifford shadow strategy as in the $T\geq n/10$ case, with two differences. 
Firstly, note that after post-selection, the action on every qubit in ${\hat{B}}$ is identity. 
So we can without loss of generality pick an arbitrary subset $A'$ of those qubits in ${\hat{B}}$ as $A\setminus \hat{A}$, and consider an $\sqrt{\epsilon}$-covering net $\mathcal{N}$ of $G$ gate unitaries on qubits $\hat{A}\cup A'$ with respect to $\davg$, with $|\hat{A}\cup A'|= |A|\leq 2G$. 
Then we have $\min_{U_i \in \mathcal{N}}\davg(U_i, U)\leq \epsilon$, and $\log|\mathcal{N}|\leq \bigo(G\log(G/\epsilon)+G\log(|A\cup A'|))\leq \bigo(G\log(G/\epsilon))$.
Secondly, for each element $U_i$ in the covering net, we can construct an observable $O_i = U_i\ket{x}\bra{x}U_i^\dagger$ similar to before, but now the estimator will concentrate around a slightly different expectation value. 
Specifically, if we use a median-of-mean estimator $\hat{o}_i(N, K)$ with $K=2\log(2|\mathcal{N}|/\delta)$ and $N=102/\epsilon^2$, then we have
\begin{equation}
    |\hat{o}_i(N, K) - \mathbb{E}_x[\tr(\rho_x'O_i)]| \leq \epsilon \quad \text{for all} \quad 1\leq i\leq |\mathcal{N}|,
\end{equation}
with probability at least $1-\delta$.
Nevertheless, since $\rho_x'$ and $\rho_x$ are close on average, we have
\begin{equation}
\begin{split}
    |\hat{o}_i(N, K) - (1-d_Q^2(U_i\otimes I, U))| &= |\hat{o}_i(N, K) - \mathbb{E}_x [\tr(\rho_x'O_i)] + \mathbb{E}_x[\tr(\rho_x'O_i)] - \mathbb{E}_x[\tr(\rho_xO_i)]| \\
    &\leq |\hat{o}_i(N, K) - \mathbb{E}_x[\tr(\rho_x'O_i)]| + \mathbb{E}_x[|\tr(\rho_x'O_i) - \tr(\rho_xO_i)|] \\
    &\leq \epsilon + \mathbb{E}_x[\|\rho_x'-\rho_x\|_1 \|O_i\|]
    \leq \epsilon + 3\epsilon=4\epsilon, \quad \text{for all} \quad 1\leq i\leq |\mathcal{N}|,
\end{split}
\end{equation}
with probability at least $1-2\delta$, where we have used triangle inequality, $\|O_i\|=1$, and $\mathbb{E}_x\|\rho_x-\rho_x'\|_1\leq 3\epsilon$.
With this concentration guarantee, we can select the candidate with the largest $\hat{o}_i$: $i^\star \in \mathrm{argmax}_i \hat{o}_i$ and output $U_{i^\star}\otimes I$. As before, we have with probability at least $1-2\delta$,
\begin{equation}
    \davg(U_{i^\star}\otimes I, U) \leq \sqrt{2}d_Q(U_{i^\star}\otimes I, U) \leq \sqrt{2(4\epsilon+2\epsilon+4\epsilon)}=\sqrt{20\epsilon}.
\end{equation}
Redefining $20\epsilon$ to be $\epsilon^2$ and $2\delta$ to be $\delta$, we arrive at a learning algorithm that uses 
\begin{equation}
    m = \bigo\left(\frac{G + \log(1/\delta)}{\epsilon^4}\right) + NK = \bigo\left(\frac{G\log(G/\epsilon)+\log(1/\delta)}{\epsilon^4}\right)
\end{equation}
queries to the unitary to learn it with accuracy $\epsilon$ in $\davg$ and success probability at least $1-\delta$ when $G<n/10$.
Combined with the case of $G\geq n/10$, this concludes the learning algorithm without ancillary system in \Cref{prop:upper-bound-unitary}.

\subsubsection{Unitary learning with ancillary systems}
\label{app:unitary-choi}
The above $\mathcal{O}(1/\epsilon^4)$ scaling is suboptimal.
It arises from the fact that in the classical shadow estimation, the estimated quantity is the square of $\davg$ rather than $\davg$ itself.
To improve the $\epsilon$-dependence, we make use ancillary systems via the Choi–Jamio{\l}kowski duality~\cite{choi1975completely, jamiolkowski1972linear, jiang2013channel}.
Specifically, we consider the maximally entangled state over a pair of $n$-qubit systems $\ket{\Phi} = \frac{1}{\sqrt{d}}\sum_{i=1}^{2^n}\ket{i}\otimes\ket{i}$ and define the Choi state $\kket{U}$ corresponding to a unitary $U$ as $\kket{U} = (U\otimes I)\ket{\Phi}$.
That is, the Choi state $\kket{U}$ of an $n$-qubit unitary $U$ is a pure $(2n)$-qubit state constructed by applying $U$ on half of the qubits in $n$ EPR pairs.
For any subset $A\subseteq [n]$ of the qubits that are acted upon by $U$, we refer to the corresponding $|A|$ qubits in the EPR pairs as the entangled qubits corresponding to $A$.
We note the following fact, which relates the trace distance between Choi states to the average-case distance between the unitaries.

\begin{lemma}[Equivalence of trace distance between Choi states and average-case distance]
\label{lem:trace-dist-choi}
    Let $U, V\in U(2^n)$ be two $n$-qubit unitaries, $\ket{\Phi} = \frac{1}{\sqrt{d}}\sum_{i=1}^{2^n}\ket{i}\otimes\ket{i}$ be a maximally entangled state, and $\kket{U} = (U\otimes I)\ket{\Phi}, \kket{V} = (V\otimes I)\ket{\Phi}$ be the corresponding Choi states.
    Then we have
    \begin{equation}
        \frac{1}{\sqrt{2}}\dtr(\kket{U}, \kket{V}) \leq  \davg(U, V) \leq \dtr(\kket{U}, \kket{V}).
    \end{equation}
\end{lemma}
\begin{proof}
    By the standard conversion between fidelity and the trace distance between pure states, we have
    \begin{equation}
        \dtr(\kket{U}, \kket{V}) = \sqrt{1-|\bbrakket{U}{V}|^2} = \sqrt{1-\frac{1}{d^2}|\tr(U^\dagger V)|^2}.\, ,
    \end{equation}
    where the last step used that $\bra{\Phi}A\otimes B\ket{\Phi} = \tfrac{1}{d}\tr[A^T B]$, compare for instance \cite[Example 1.2]{wolf2012quantumchannels}.
    On the other hand, from \Cref{eq:davg}, we have
    \begin{equation}
        \davg(U, V) = \sqrt{1-\frac{d+|\tr(U^\dagger V)|^2}{d^2+d}}.
    \end{equation}
    Combining these two equations, we get
    \begin{equation}
        \davg(U, V) = \sqrt{\frac{d}{d+1}}\dtr(\kket{U}, \kket{V}) \in \left[\frac{1}{\sqrt{2}}\dtr(\kket{U}, \kket{V}), \dtr(\kket{U}, \kket{V})\right].
        \qedhere
    \end{equation}
\end{proof}

With \Cref{lem:trace-dist-choi}, we construct a covering net over Choi states corresponding to $G$-gate unitaries as follows. 
From \Cref{cor:covering-avg}, we take an $\epsilon'$-covering net $\mathcal{N}$ of $G$-gate unitaries with respect to $\davg$ that has cardinality $|\mathcal{N}|\leq \bigo(G\log(G/\epsilon)+G\log n)$. 
Then for any $G$-gate unitary $U$, there exists a $U_i\in \mathcal{N}$ such that $\davg(U, U_i)\leq \epsilon'$.
Hence $\dtr(\kket{U}, \kket{U_i})\leq \sqrt{2}\davg(U, U_i)\leq \sqrt{2}\epsilon'$ by \Cref{lem:trace-dist-choi}.
Therefore, the Choi states of the unitaries in $\mathcal{N}$ form a $(\sqrt{2}\epsilon)$-covering net of the Choi states of $G$-gate unitaries.

Now, we can use these pure Choi states as candidates for hypothesis selection.
By \Cref{prop:hypothesis}, the hypothesis selection algorithm based on classical shadow uses $\bigo(\log(|\mathcal{N}|/\delta)/\epsilon'^2)$ samples of the Choi state $\kket{U}$ to output a candidate $\kket{\hat{U}}, \hat{U}\in\mathcal{N}$, such that $\dtr(\kket{U}, \kket{\hat{U}})\leq 3\sqrt{2}\epsilon' + \epsilon'$ with probability at least $1-\delta$.
Setting $(3\sqrt{2}+1)\epsilon'=\epsilon$, we find a $\hat{U}$ such that $\davg(\hat{U}, U)\leq \dtr(\kket{U}, \kket{\hat{U}})\leq \epsilon$ with probability at least $1-\delta$ using 
\begin{equation}
    \bigo\left(\frac{G\log(G/\epsilon)+G\log n+\log(1/\delta)}{\epsilon^2}\right)
\end{equation}
queries to the unknown unitary $U$.
When $G\geq n/10$, this gives the desired $\bigo((G\log(G/\epsilon)+\log(1/\delta))/\epsilon^2)$ query complexity.

For $G<n/10$, we again need a junta learning step to identify the set of qubits $A\subseteq [n]$ that are acted on non-trivially.
To do this, we follow the idea of Algorithm~\ref{algo:nonzero-qubits}, Algorithm~\ref{algo:junta-unitary} and \cite[Algorithm 8]{chen2023testing} and consider the following procedure that makes use of Choi states of Pauli matrices $\sigma_0=I, \sigma_1 = X, \sigma_2=Y, \sigma_3=Z$.

\begin{algorithm}
\label{algo:junta-choi}
\caption{Identify qubits acted upon nontrivially (Choi version)}
\KwIn{Query access to the unknown unitary $U$ with $G$ two-qubit gates.}
\KwOut{List $\hat{A} \subseteq [n]$ of qubits.}
Initialize $\hat{A} = \emptyset$.\\
Repeat the following $N = \mathcal{O}\left(\frac{G + \log (1/\delta)}{\epsilon^2}\right)$ times: \linebreak
(a) Prepare the Choi state $\kket{U}$ by applying $U\otimes I$ to the maximally entangled state $\ket{\Phi}$
\linebreak 
(b) Measure in the basis of Pauli Choi states $\kket{\sigma_x} = \otimes_{i=1}^n (\sigma_{x_i}\otimes I)\ket{\Phi}, x\in \mathbb{Z}_4^n$, and obtain a string $\ket{b}, b\in\{0, 1, 2, 3\}^n$
\linebreak 
(b) Given the measurement outcome $\ket{b}$, set $\hat{A} \leftarrow \hat{A} \cup \mathrm{supp}(b)$, where $\mathrm{supp}(b) = \{i \in [n] : b_i \neq 0\}$.
\end{algorithm}

Similarly to \Cref{lemma:good-overlap} and \Cref{lemma:good-overlap-unitary}, we have the following guarantee that the Choi state on the estimated trivial qubits is close to the Choi state of the identity.

\begin{lemma}
    \label{lemma:good-overlap-choi}
    Let $\epsilon, \delta > 0$. 
    Suppose we are given query access to an $n$-qubit unitary $U$ composed of $G$ two-qubit gates acting on a subset of the qubits $A \subseteq [n]$. 
    Let $\rho = \kket{U}\bbra{U}$ be the Choi state of $U$.
    Then, Algorithm~\ref{algo:junta-choi} uses $N = \mathcal{O}\left(\frac{G + \log(1/\delta)}{\epsilon^2}\right)$ queries to $U$ and outputs, with probability at least $1-\delta$, a list $\hat{A} \subset [n]$ such that
    \begin{equation}
        \bbra{I_{\hat{B}}}\rho_{\hat{B}}\kket{I_{\hat{B}}} \geq 1 - \epsilon^2,
    \end{equation}
    where $\rho_{\hat{B}}$ denotes the reduced density matrix for $\rho$ by tracing out all qubits other than those in the set $\hat{B} = [n] \setminus \hat{A}$ and the corresponding entangled qubits, and $\kket{I_{\hat{B}}}$ denotes the Choi state of the identity on qubits in ${\hat{B}}$.
\end{lemma}
\begin{proof}
    The proof goes similarly to that of \Cref{lemma:good-overlap} except that $\ket{0}$ is replaced by $\kket{I}$.
    The measurement over Pauli Choi states in Algorithm~\ref{algo:junta-choi} can be understood as measuring each entangled pair of qubits in the basis $\{\kket{I}, \kket{X}, \kket{Y}, \kket{Z}\}$ and gives an element from $\{0, 1, 2, 3\}=\mathbb{Z}_4$.
    Specifically, let $A'$ be any set that could be output by Algorithm~\ref{algo:junta-choi}.
    We want to identify $A'$ with the actual identified set $\hat{A}$.
    Let $B' \triangleq [n] \setminus A'$.
    Let $E_{i, A'}$ be the event that round $i$ of measurement of the qubits in $B' = [n] \setminus A'$ in Algorithm~\ref{algo:junta-choi} yields the all zero $\mathbb{Z}_4$ string.
    Let $X_{i, A'}$ be the indicator random variable corresponding to the event $E_{i,A'}$.
    Then, we have that $\bar{X}_{A'} \triangleq \frac{1}{N}\sum_{i=1}^N X_{i, A'}$ is the number of times the entangled pair in $B'$ are all measured to be zero divided by the total number of measurements.
    In other words, $\bar{X}_{A'}$ is the estimated overlap that the state $\rho_{B'}$ on qubits in $B'$ has with the identity Choi state on $B'$.
    Moreover, we have
    \begin{equation}
    \mathbb{E}[X_{A'}] \triangleq \mathbb{E}[X_{i,A'}] = \bbra{I_{B'}}\rho_{B'}\kket{I_{B'}}
    \end{equation}
    for all $A'$.
    This says that the true expectation of our random variables is the true overlap of the state $\rho_{B'}$ with the identity Choi state on $B'$.
    Then we have the same \Cref{claim:concentrate} as in \Cref{lemma:good-overlap} and \Cref{lemma:good-overlap-choi} follows.
\end{proof}

With this, we can again show that ignoring the rest of the qubits ${\hat{B}}=[n]\setminus \hat{A}$ does not make much difference. 
Let $B=[n]\setminus A$. 
We prepare the Choi state $\rho = \kket{U}\bbra{U}, \kket{U} = (U\otimes I)\ket{\Phi}$, and measure in the basis of Pauli Choi states over the qubits in $\hat{B}$: $\{\kket{\sigma_x}_{\hat{B}}: x\in\mathbb{Z}_4^{|\hat{B}|}\}$.
After the measurement, we do a post-selection on the observed measurement outcomes being $\kket{I_{\hat{B}}}$. 
This post-selection is represented by $\Lambda = I\otimes \kket{I_{\hat{B}}}\bbra{I_{\hat{B}}}$, with $\Lambda^2 = \Lambda$, and the first identity over the entangled pairs outside $\hat{B}$.
Let $\rho' = \frac{\sqrt{\Lambda}\rho\sqrt{\Lambda}}{\tr(\Lambda \rho)}$ be the post-selected state. 
Now we want to show $\rho'$ is close to $\rho$. 
From \Cref{lemma:good-overlap-choi}, we know that $\tr(\Lambda\rho)\geq 1-\epsilon^2$ with probability at least $1-\delta$.
Then by the gentle measurement lemma (\Cref{lemma:gentle-meas}), we have $\dtr(\rho', \rho)\leq \epsilon$ with the same probability.

Now we can apply the hypothesis selection protocol to $\rho'$ as in the $G>n/10$ case, but with a different covering net.
Specifically, note that after post-selection, the action on every entangled pair in ${\hat{B}}$ is identity. 
So we can with loss of generality pick an arbitrary subset $A'$ of those qubits in ${\hat{B}}$ as $A\setminus\hat{A}$, and consider an $\epsilon$-covering net $\mathcal{N}$ of $G$ gate unitaries on qubits $\hat{A}\cup A'$ with respect to $\davg$, with $|\hat{A}\cup A'|= |A|\leq 2G$, with each element tensor product with identity over the rest qubits. 
Then we have $\min_{U_i \in \mathcal{N}}\dtr(\kket{U_i}, \kket{U})\leq \sqrt{2}\davg(U_i, U)\leq \sqrt{2}\epsilon$, and $\log|\mathcal{N}|\leq \bigo(G\log(G/\epsilon)+G\log(|A\cup A'|))\leq \bigo(G\log(G/\epsilon))$.
Since $\dtr(\rho', \rho)\leq \epsilon$, we also have $\min_{U_i \in \mathcal{N}}\dtr(\kket{U_i}\bbra{U_i}, \rho')\leq \epsilon+\sqrt{2}\epsilon = (\sqrt{2}+1)\epsilon$.

With this covering net, we apply the hypothesis selection based on classical shadow (\Cref{prop:hypothesis}) to $\rho'$. 
This procedure uses $\bigo(\log(|\mathcal{N}|/\delta)/\epsilon^2)$ copies of $\rho'$ (each prepared using one query to $U$) and output a $\kket{U_{i^\star}}$ such that $\dtr(\kket{U_{i^\star}}\bbra{U_{i^\star}}, \rho')\leq 3(\sqrt{2}+1)\epsilon+\epsilon=(3\sqrt{2}+4)\epsilon$ with probability at least $1-\delta$.
This means that 
\begin{equation}
\label{eq:shadow-acc}
    \davg(U_{i^\star}, U)\leq \dtr(\kket{U_{i^\star}}, \kket{U})\leq \dtr(\kket{U_{i^\star}}\bbra{U_{i^\star}}, \rho') + \dtr(\rho', \rho) \leq (3\sqrt{2}+4)\epsilon+\epsilon = 4(\sqrt{2}+1)\epsilon
\end{equation} 
with a total probability at least $1-2\delta$ (considering both the junta learning and hypothesis selection).

Therefore, by redefining $4(\sqrt{2}+1)\epsilon$ to be $\epsilon$ and $2\delta$ to be $\delta$, we arrive at a desired algorithm for $G\leq n/10$ that uses in total 
\begin{equation}
    \bigo\left(\frac{G+\log(1/\delta)}{\epsilon^2}\right) + \bigo\left(\frac{G\log(G/\epsilon)+\log(1/\delta)}{\epsilon^2}\right) = \bigo\left(\frac{G\log(G/\epsilon)+\log(1/\delta)}{\epsilon^2}\right)
\end{equation}
queries to the unknown unitary $U$.
Combined with the $G\geq n/10$ case, we conclude the learning algorithm with ancillary system that achieves the $\bigo((G\log(G/\epsilon)+\log(1/\delta)/\epsilon^2)$ query complexity in \Cref{prop:upper-bound-unitary}.

\subsubsection{Bootstrap to improve $\epsilon$-dependence}

To further improve the $\epsilon$-dependence, we modify the bootstrap method in \cite{haah2023query} and achieve a Heisenberg scaling $\tilde{\bigo}(1/\epsilon)$.
However, with our average case distance, which can only control the average behavior of the eigenvalues of the unitaries, we are not able to perform the bootstrap for general $\epsilon$.
Instead, the bootstrap works only when the error is exponentially small, $\epsilon=O(1/\sqrt{d})$, and achieves the Heisenberg scaling at the cost of a dimensional factor, leading to a query complexity of
\begin{equation}
    \bigo\left(\frac{\sqrt{2^n}(G\log(G/\epsilon)+\log(1/\delta))}{\epsilon}\right).
\end{equation}
Whether a general Heisenberg scaling without dimension-dependent scaling is achievable remains open.

Now we state the bootstrap method in Algorithm~\ref{alg:bootsrap}, which uses the unitary learning algorithm with ancillary systems (\Cref{app:unitary-choi}) as a sub-routine.
We need to prove two things about Algorithm~\ref{alg:bootsrap}: (1) it outputs a $\hat{U}$ that satisfies $\davg(\hat{U}, U)\leq \epsilon$ with probability at least $1-\delta$; (2) the query complexity is $\bigo\left({\sqrt{d}(G\log(G/\epsilon)+\log(1/\delta))}/{\epsilon}\right)$.

\begin{algorithm}
\label{alg:bootsrap}
  \caption{Bootstrapping to Heisenberg scaling}
    \KwIn{
        Query access to the unknown $n$-qubit $G$-gate unitary $U$.\newline
        An error parameter $\epsilon \in (0, 1/\sqrt{d})$.
    }
    \KwOut{A unitary $\hat{U}$.}
  \Let{$t$}{$\ceil{\log_2(1/(\epsilon\sqrt{d}))}$.} \\
  \Let{$V_0$}{$I$.} \\
  \Let{$\mathcal{N}$}{an $(\epsilon/10^5)$-covering net of $G$-gate unitaries with respect to $\davg$.} \\
  \For{$j \gets 0$ to $t$}{
    \Let{$p_j$}{$2^j$.} \\
    \Let{$\eta_j$}{$8^{j-t-1}\delta$.} \\
    Use the algorithm $\mathcal{A}$ in \Cref{app:unitary-choi} with success probability $1-\eta_j$ and accuracy $1/(25000\sqrt{d})$ to find a candidate $R_j$ in $\{(U_i V_j^\dagger)^{p_j}~|~ U_i\in\mathcal{N}\}$ that is closest to $(UV_j^\dagger)^{p_j}$ in $\davg$. \\
    \Let{$V_{j+1}$}{$R_j^{1/p_j}V_j$.}
  }
  \Return{$\hat{U}\gets V_{t+1}$}
\end{algorithm} 

We first prove (1) by induction. 
Before doing so, we need to show that the learning algorithm $\mathcal{A}$ can indeed learn $(UV_j^\dagger)^{p_j}$ well for all $j$. 
Let $c=10^{-5}$. Note that with the definition of $\mathcal{N}$, we know that for any $G$-gate unitary $U$, $\exists U_i\in \mathcal{N}$ such that $\davg(U, U_i)\leq c\epsilon$, and therefore
\begin{equation}
    \davg((U_iV_j^\dagger)^{p_j}, (UV_j^\dagger)^{p_j}) \leq d_F'((U_iV_j^\dagger)^{p_j}, (UV_j^\dagger)^{p_j}) \leq p_j d_F'(U_i, U) \leq 2p_j\davg(U_i, U)\leq 4c/\sqrt{d},
\end{equation}
where we have used Item 1 and 2 in \Cref{lem:dist-df'}, unitary invariance of $d_F'$, and $p_j=2^j\leq 2^t\leq 2/(\epsilon\sqrt{d})$. Thus $\{(U_iV_j^\dagger)^{p_j}, U_i\in\mathcal{N}\}$ forms an $4c/\sqrt{d}$-covering net of $\{(UV_j^\dagger)^{p_j}~|~ \text{$U$ is a $G$-gate unitary}\}$, which can be used by the hypothesis selection algorithm $\mathcal{A}$ as set of candidates. 
The output $R_j$ of $\mathcal{A}$ satisfies $d_F'(R_j, (UV_j^\dagger)^{p_j})\leq 2\davg(R_j, (UV_j^\dagger)^{p_j})\leq 4(\sqrt{2}+1)\cdot 4c/\sqrt{d} < 40c/\sqrt{d}$ (see \Cref{eq:shadow-acc}). 
The number of queries to $U$ that this procedure uses is $\bigo\left(p_j\frac{G\log(G/c\epsilon)+\log(1/\eta_j)}{(4c/\sqrt{d})^2}\right) = \bigo\left(p_j d(G\log(G/\epsilon)+\log(1/{\eta_j}))\right)$.

Now we proceed to prove (1) by induction. Let's assume that the learning algorithm succeeds for all $j=1, \ldots, t$. Let $\delta_j = d_F'(U, V_j) = d_F'(UV_j^\dagger, I)$ be the error after iteration $j-1$. 
We will prove that $\delta_k \leq 2^{-k-5}/\sqrt{d}$. 
For iteration $0$, we have $p_0=1$, and by the accuracy of $\mathcal{A}$, we know $\delta_1=d_F'(U, V_1)<40c/\sqrt{d}<2^{-6}/\sqrt{d}$. 
Now we assume $\delta_k \leq 2^{-k-5}/\sqrt{d}$ and prove $\delta_{k+1}\leq 2^{-k-6}/\sqrt{d}$. 
Note that $(UV_k^\dagger)^{p_k}$ and $R_k$ are sufficiently close to identity in the sense that
\begin{equation}
    d_F'((UV_k^\dagger)^{p_k}, I)\leq p_k d_F'(UV_k^\dagger, I) = p_k\delta_k \leq \frac{2^{-5}}{\sqrt{d}} < \frac{4/(25\pi)}{\sqrt{d}},
\end{equation}
and
\begin{equation}
    d_F'(R_k, I)\leq d_F'(R_k, (UV_k^\dagger)^{p_k}) + d_F'((UV_k^\dagger)^{p_k}, I) \leq \frac{40c}{\sqrt{d}} + \frac{2^{-5}}{\sqrt{d}} < \frac{4/(25\pi)}{\sqrt{d}}.
\end{equation}
Thus, we can invoke Item 3 of \Cref{lem:dist-df'} and obtain
\begin{equation}
\begin{split}
    \delta_{k+1} &= d_F'(U, V_{k+1}) = d_F'(UV_k^\dagger, R_k^{1/p_k}) \leq \frac{2}{p_k}d_F'((UV_k^\dagger)^{p_k}, R_k) \leq \frac{80c}{p_k\sqrt{d}} < \frac{2^{-k-6}}{\sqrt{d}}.
\end{split}
\end{equation}
Therefore, by induction, we have shown that $\delta_k \leq 2^{-k-5}/\sqrt{d}$. 
At the end of the iteration, when $k=t=\ceil{\log_2(1/(\sqrt{d}\epsilon))}$, we have
\begin{equation}
    \delta_{t+1} = d_F'(U, V_{t+1}) \leq \frac{2^{-t-6}}{\sqrt{d}}<\epsilon.
\end{equation}

The above accuracy is conditioned on the success of all executions of the learning algorithm. 
By the union bound, the failure probability is upper bounded by
\begin{equation}
    \sum_{j=0}^t \eta_j = \delta\sum_{j=0}^t 8^{-(t-j)-1} = \delta\sum_{j=0}^t 8^{-j-1}<\delta.
\end{equation}
This concludes the proof of (1).

Next we move on to (2) and count the overall number of queries to the unknown unitary. Summing over all iterations, the number of queries is
\begin{equation}
\begin{split}
    &\bigo\left(\sum_{j=0}^t p_j d(G\log(G/\epsilon)+\log(1/\eta_j)) \right) \\
    &= \bigo\left(dG\log(G/\epsilon)\sum_{j=0}^t 2^j + d\log(1/\delta)\sum_{j=0}^t 2^j (t-j+1) \right) \\
    &= \bigo\left(d(G\log(G/\epsilon)+\log(1/\delta))2^t \right)=\bigo\left(\frac{\sqrt{d}(G\log(G/\epsilon)+\log(1/\delta))}{\epsilon} \right).
\end{split}
\end{equation}
This concludes the proof of the $\bigo(1/\epsilon)$ scaling algorithm in \Cref{prop:upper-bound-unitary}.

Finally, we note that an analogous bootstrap method can also be applied to improve the $\epsilon$-dependence for our unitary learning procedure without auxiliary systems, albeit again incurring a dimension factor.
Namely, a variant of Algorithm \ref{alg:bootsrap} relying on the algorithm of \Cref{app:unitary-no-ancilla} as a subroutine succeeds at outputting a $\hat{U}$ that satisfies $\davg(\hat{U}, U)\leq \epsilon$ with probability at least $1-\delta$ using $\bigo\left({d^{3/2}(G\log(G/\epsilon)+\log(1/\delta))}/{\epsilon}\right)$ queries to the unknown unitary $U$, assuming $\epsilon < \nicefrac{1}{d^{3/2}}$.

\subsection{Average-case query complexity lower bounds}
\label{app:avg-case-unitary-lo}
For the lower bound, we construct a packing net consisting of $G$-gate unitaries that are pairwise sufficiently far apart, so that an average-case learning algorithm can discriminate them. Meanwhile, the success probability of distinguishing a set of unitaries is upper bounded by the number of queries made~\cite{bavaresco2022unitary}. This gives us an $\Omega(G)$ query complexity lower bound. To incorporate $\epsilon$-dependence, we follow \cite{haah2023query} and map the problem to a fractional query problem \cite{berry2015hamiltonian, cleve2009efficient}. This way, we arrive at the following result.

\begin{prop}[Average case unitary learning lower bound, lower bound in \Cref{thm:avg-case-unitary}]
    Let $U$ be an $n$-qubit unitary composed of $G$ two-qubit gates.
    Any algorithm that, given query access to $U$, $U^\dagger$, $cU=\ketbra{0}\otimes I + \ketbra{1}\otimes U$ and $cU^\dagger=\ketbra{0}\otimes I + \ketbra{1}\otimes U^\dagger$, can output a unitary $\hat{U}$ such that $\davg(\hat{U}, U)\leq \epsilon\in (0, 1/32)$ with probability at least $2/3$, must use at least $\Omega(G/\epsilon)$ queries.
\label{prop:lower-bound-unitary}
\end{prop}

Note that the lower bound holds even for learning algorithms that have a stronger form of access to $U$ than considered for our upper bounds. There, we only assumed query access to $U$. In contrast, the lower bound holds even assuming query access to $U$ and $U^\dagger$ as well as controlled versions thereof.

\begin{proof}[Proof of \Cref{prop:lower-bound-unitary}]
    The proof builds on the following lemma that maps the problem to a fractional query one \cite{haah2023query}.
    
    \begin{lemma}[Reduction to fractional query algorithms, {\cite[Lemma 4.5 and proof of Theorem 1.2]{haah2023query}}]
    Let $R\in U(d)$ be a Hermitian unitary (i.e., $R^2=I$). 
    Define $R^\alpha = (I+R)/2+e^{-i\pi\alpha}(I-R)/2$ for some $\alpha\in (0, 1]$. 
    Suppose there exists an algorithm $\mathcal{A}$ that uses $Q$ queries to $R^\alpha$ or $R^{\alpha\dagger}$ and produces some output with probability at least $2/3$. 
    Then there exists another algorithm $\mathcal{A}'$ that uses $50+100\alpha Q$ queries to controlled-$R$ and produces the same output with probability at least $\exp(-\alpha\pi Q)/2$.
    \label{lem:frac_query}
    \end{lemma}
    
    To use this lemma, we need to construct a packing net of Hermitian unitaries, and give an upper bound on the maximum probability of successfully distinguishing them. 
    Thus we need the following two lemmas.
    
    \begin{lemma}[Packing net of Hermitian unitaries, variant of {\cite[Proposition 4.1]{haah2023query}}]
    There exists a set of Hermitian unitaries $\mathcal{P} = \{R_i\}_i\subset U(d)$ with $\log|P|\geq \Omega(d^2)$ and $R_i^2=I$ for $R_i\in\mathcal{P}$, such that for any $R_i\neq R_j\in \mathcal{P}$, $d_F'(R_i, R_j)\geq 1/8$.
    \label{lem:packing_r}
    \end{lemma}
    \begin{proof}
        Let $d=2r+1$ if $d$ is odd, or $d=2r+2$ if $d$ is even. 
        \cite[Lemma 7]{haah2016sample} (or Lemma 8 in \cite{haah2017sample}) asserts that there exists a set of rank $r$ density matrices in dimension $2r$ with cardinality at least $\exp(r^2/8)$, such that all the non-zero eigenvalues are equal to $1/r$, and any two different density matrices have trace distance at least $1/4$. 
        We can write this set as $\{(I_{2r}+V_i)/(2r), i=1, \ldots, N\}$, where $V_i\in U(2r)$ is a Hermitian unitary of trace zero. 
        Then $N\geq \exp(r^2/8)$, and $\forall i\neq j$,
        \begin{equation}
            \frac{1}{4}\leq \frac{1}{2}\left\|\frac{I_{2r}+V_i}{2r} - \frac{I_{2r}+V_j}{2r}\right\|_1 = \frac{1}{4r}\|V_i-V_j\|_1 \leq \frac{1}{\sqrt{2r}}\|V_i-V_j\|_F.
        \end{equation}
        where we have used $\|V_i\|_1\leq \sqrt{2r}\|V_i\|_F$. 
        Then we embed $V_i\to R_i = V_i\oplus I_b\in U(d)$, where $b=1$ or $2$, depending on whether $d$ is odd or even. 
        We have
        \begin{equation}
            d_F(R_i, R_j) = \frac{1}{\sqrt{d}}\|R_i-R_j\|_F \geq \frac{1}{4}\sqrt{\frac{2r}{2r+b}} \geq \frac{1}{8}.
        \end{equation}
        Now we would like to translate $d_F$ into $d_F'$. From \Cref{lem:quotient-packing}, we know that changing to the quotient metric for any set of unitaries only decreases $\log N$ by an additive constant (since here we consider constant $\epsilon$). Therefore, we still have $\log|P|\geq \Omega(d^2)$ for $d_F'(R_i, R_j)\geq \frac{1}{8}$.
    \end{proof}
    
    \begin{lemma}[Upper bound on success probability of distinguishing unitaries, {\cite[Theorem 5]{bavaresco2022unitary}}]
    Let $\mathcal{P}\subseteq U(d)$ be a set of unitaries. 
    Let $\mathcal{A}$ be any algorithm that uses $Q$ queries to an input unitary $U_x$ and output a guess $\hat{x}$. 
    Suppose the input unitary is randomly picked from $\mathcal{P}$ with uniform probability. 
    Then the maximal probability that the output satisfies $\hat{x}=x$ is upper bounded by $\frac{1}{|\mathcal{P}|}\binom{Q+d^2-1}{Q}$.
    \label{lem:max_prob_disting}
    \end{lemma}

    Now we can proceed to prove the lower bound in \Cref{prop:lower-bound-unitary}. 
    Suppose we have a learning algorithm $\mathcal{A}$ that uses $Q$ queries and outputs a $\hat{U}$ that has accuracy $\epsilon$ in $\davg$ with success probability at least $2/3$. 
    From the theory of universal gates \cite{vartiainen2004efficient}, we know that $G=\bigo(4^k)$ gates suffice to implement an arbitrary $k$-qubit unitary, i.e., there exists constant $C$ such that $G$ gates can implement arbitrary unitary on $k=\floor{\log_4(G/C)}$ qubits. 
    Let $d=\min\{2^n, 2^k\}$, and focus on the first $\min\{n, k\}$ qubits. The algorithm $\mathcal{A}$ thus is able to learn any unitary on these qubits.

    Consider the packing net $\mathcal{P}=\{R_i\}$ from \Cref{lem:packing_r} for this choice of $d$. 
    We want to identify $R\in \mathcal{P}$, but using only access to $R^\alpha$ for $1/\alpha = \floor{1/32\epsilon}>1$. 
    If we apply $\mathcal{A}$ to $R^\alpha$, then with probability at least $2/3$, the output $U$ satisfies $\davg(U, R^\alpha)\leq \epsilon$. 
    From the equivalence of $\davg$ and $d_F'$ (\Cref{lem:dist-df'}), and the triangle inequality and unitary invariance of $d_F'$, we have
    \begin{equation}
        d_F'(U^{1/\alpha}, R) \leq \sum_{p=1}^{1/\alpha} d_F'(U^{p}R^{1-\alpha p}, U^{p-1}R^{1-\alpha p+\alpha}) = \frac{2}{\alpha} \davg(U, R^\alpha)\leq \frac{2\epsilon}{\alpha}\leq \frac{1}{16}.
    \end{equation}
    Since $R\in \mathcal{P}$ have pairwise distance at least $1/8$, the algorithm can identify $R$ with success probability at least $2/3$ by finding the closest element of $\mathcal{P}$ to $U^{1/\alpha}$.

    Now, via \Cref{lem:frac_query}, we know that there is a learning algorithm $\mathcal{A}'$ that can use $50+100\alpha Q$ queries to controlled-$R$ to identify $R$ with success probability at least $\exp(-\alpha\pi Q)/2$.
    On the other hand, we know the success probability cannot exceed the upper bound $\binom{Q+d^2-1}{Q}/\mathcal{|P|}$ set by \Cref{lem:max_prob_disting} with $\log|P|\geq \Omega(d^2)$. 
    Combined with a technical lemma \cite[Lemma 4.3]{haah2023query}, this means that the number of queries must be at least $\Omega(d^2)$. That is,
    \begin{equation}
        50+100\alpha Q \geq \Omega(d^2) \implies Q\geq \Omega\left(\frac{d^2}{\alpha}\right) = \Omega\left(\frac{d^2}{\epsilon}\right) = \Omega\left(\frac{\min\{4^n, G\}}{\epsilon}\right).
    \end{equation}
    This concludes the proof of \Cref{prop:lower-bound-unitary}.
\end{proof}

We comment on the connection of our results to the recent work \cite{yang2023complexity} on the hardness of learning Haar-random unitaries, where the authors proved a sample complexity lower bound $\Omega\left(\frac{d^2}{\log^2 d}\right)$ for learning $d$-dimensional Haar-random unitaries to constant accuracy w.r.t.~$d'_F$. 
The direct consequence of our lower bound when applied to learning the whole unitary group $U(d)$, without assumptions of limited complexity, is a lower bound of $\Omega\left(d^2\right)$, which is stronger than that of \cite[Theorem 1]{yang2023complexity} by a factor of $\log^2 d$. 
We note that this difference is a consequence of proof techniques that comes about in two ways. 
One $\log d$ factor comes from their analysis of the differential entropy, which only calculated the contribution of $\Theta\left(\frac{d}{\log d}\right)$ columns of the matrix elements, instead of all $d$ columns. 
This issue does not arise for us because we focus on the discrete entropy with the use of a packing net. 
The other $\log d$ comes from the mutual information upper bound, where they use the straightforward Holevo bound: Each $d$-dimensional quantum state can carry at most $\bigo(\log d)$ bits of information. 
We manage to get rid of this factor by making use of a more refined bound on success probability as in \Cref{lem:max_prob_disting}.

Lastly, we remark on the proof technique used here compared to the Holevo information bound in the state learning case (\Cref{app:states-lower}).
The Holevo bound is particularly useful in proving these lower bounds because, combined with the data processing inequality, it gives an upper bound on the amount of information that can be extracted from quantum states.
In particular, it asserts that, given an ensemble of $d$-dimensional states $\{\rho_X\}$ with random classical labels $X\in [M]$, the maximal mutual information with the underlying random label when using $k$ copies of the state is upper bounded by $\chi(X; \rho_X^{\otimes k})\triangleq S(\E_{X}[\rho_X^{\otimes k}]) - \E_{X}[S(\rho_X^{\otimes k})].$
Meanwhile, the information needed to distinguishing a packing net of $d$-dimensional states is lower bounded by $\Omega(d)$.
Thus, upper bounding the Holevo $\chi$ via the number of samples $k$ can give us sample complexity lower bounds.
A naive upper bound is $\chi\leq S(\E_{X}[\rho_X^{\otimes k}])\leq k\log d$ because $\E_{X}[\rho_X^{\otimes k}]$ is a $d^k$-dimensional mixed state and thus has entropy at most $k\log d$.
This gives us a $\Omega(d/\log d)$ sample complexity lower bound with a sub-optimal logarithmic factor.
To get rid of the $\log d$ factor, \cite{wright2016learn} noted that $k$ copies of a $d$-dimensional pure state live in the symmetric subspace of the $k$-fold tensor power of $d$-dimensional Hilbert space. 
Therefore, the first term $S(\E_{X}[\rho_X])$, along with the Holevo $\chi$, can be more tightly upper bounded by $\log\binom{k+d-1}{k}$, where the binomial coefficient is the dimension of the symmetric subspace.
This can then be used to prove a $\Omega(d)$ lower bound, which is optimal in $d$.

However, an analogous result for unitaries (or more generally channels) queries is still lacking. 
Consider an ensemble of channels $\{C_X\}$ labeled by a classical random variable $X\in [M]$.
In general, one can sequentially query the channel $k$ times interleaved with processing operations to prepare a state carrying the information extracted from the queries. This then has the form $\rho^k_X = \mathcal{C}_k C_X \mathcal{C}_{k-1} C_X \cdots \mathcal{C}_1 C_X (\rho^0)$, where $\mathcal{C}_i$ are fixed channels independent of $X$, and $\rho_0$ is some fixed state.
Then the amount of information that one can extract is given by the Holevo information $\chi(X; \rho_X^k) = S(\E_{X}[\rho_X^k]) - \E_{X}S([\rho_X^k])$.
Upper bounding this quantity is in general difficult. \cite{huang2021info} used induction and obtained $\chi\leq k\log (d^2)$ which corresponds to the naive upper bound in the state case.
Using this, however, can only give us a suboptimal $\Omega(d^2/\log d)$ query complexity lower bound.
We suspect that an improved method, similar in spirit to \cite{wright2016learn}, making use of the fact that all $k$ queries are to the same channel $C_X$ should be possible and give a
\begin{equation}
    \chi(X; \rho_X^k) \leq \log\binom{k+d^2-1}{k}
\end{equation}
upper bound.
This would then also give an information-theoretic perspective on the binomial coefficient appearing in the unitary discrimination result \Cref{lem:max_prob_disting} originally proved by positive-semi-definite programming.
We leave the proof of this Holevo information bound as an open problem for future work.

\subsection{Learning from classically described data}
\label{app:classical-description}

As we have seen in \Cref{thm:state-learning,thm:avg-case-unitary}, the sample complexity of learning $G$-gate states and unitaries are both $\tilde{\Theta}(G)$.
This suggests that they have similar source of complexity.
However, differently from state learning, we can identify two sources of difficulty in unitary learning:
(1) reading out the input and output quantum states, 
and (2) learning the mapping from inputs to outputs.
The similar complexity $\tilde{\Theta}(G)$ of both state and unitary learning suggests that learning the mapping is actually easy and may only need a constant number of queries to the unknown unitary.

To formalize this idea, we consider a different access model for the unitary learning task: We focus on learning the mapping by assuming training data that contains classical descriptions of input and output states.
Specifically, we consider a learning algorithm $H$ that selects $N$ input $n$-qubit states $\{\ket{x_i}\}_{i=1}^N$, and queries the unknown unitary to get $\{U\ket{x_i}\}_{i=1}^N$, where we have (repeated) access to the classical descriptions of all these input and output states. 
Based on these classically described data, we want to use the learning algorithm $H$ to output a $\hat{U}$ that satisfies $\davg(\hat{U}, U)\leq \epsilon$.

A recent line of research on the quantum no-free-lunch theorem \cite{poland2020no, sharma2022reformulation} implies that the above task of learning the mapping from classically described data in the average-case distance requires at least $\Omega(2^n)$ samples.
This seems to contradict our idea that learning the mapping should be easy.
However, \cite{sharma2022reformulation} also demonstrated how to circumvent the quantum no-free-lunch theorem.
In particular, they showed that by entangling our input states with an ancillary system, applying the unitary on the original system, and collecting the output entangled states, we can reduce the sample requirement by a factor equal to the Schimidt rank $r$ of the entangled states.
In the limit of maximally entangled state where $r=2^n$, the output state is in fact the Choi–Jamio{\l}kowski state of the unitary, which already contains all the matrix elements of the unitary.
Therefore, \cite{sharma2022reformulation} concluded that using entangled data can reduce the data requirements and eventually make the unitary learning task easy, requiring only one sample with a maximally entangled input state.

Here, we aim to go beyond this result and provide a unified information-theoretic reformulation of the quantum no-free-lunch theorem (\Cref{thm:qnfl}), which is not limited to entangled data.
We find that the key ingredient to reduce the sample complexity of learning with classical description is to enlarge the representation space (i.e., the space that the output states live in).
While entanglement is one way to achieve such an enlargement, it is not the only one.
In fact, we find an alternative method that only uses classically mixed states and achieve the same reduction in sample complexity.
Specifically, we establish the following theorem.

\begin{prop}[Upper bounds in learning with classical descriptions, restatement of upper bounds in \Cref{thm:classical-description}]
There exists a learning algorithm $H_\mathrm{entangle}$ that, for any $n$-qubit unitary $U\in U(2^n)$, uses $N=\ceil{2^n/r}$ classically described data $\{(\ket{x_i}, (U\otimes I)\ket{x_i})\}_{i=1}^N$, where $\ket{x_i}$ are bipartite entangled states over two $n$-qubit systems with Schmidt rank at most $r$, to output a $\hat{U}$ such that $\davg(\hat{U}, U)\leq \epsilon$ for any $\epsilon>0$.

Similarly, there exists a learning algorithm $H_\mathrm{mixed}$ that, for any $n$-qubit unitary $U\in U(2^n)$, uses $N=\ceil{2^n/r}$ classically described data $\{(\rho_i, (U\otimes I)\rho_i(U\otimes I)^\dagger)\}_{i=1}^N$, where $\rho_i$ are classically mixed states over two $n$-qubit systems with rank at most $r$ of the form
\begin{equation}
    \rho_i = \sum_{j=1}^r p_{ij} \ketbra{\phi_{ij}} \otimes \ketbra{\psi_{ij}}\, ,
\end{equation}
to output a $\hat{U}$ such that $\davg(\hat{U}, U)\leq \epsilon$ for any $\epsilon>0$.
\label{prop:classical-description-upper}
\end{prop}

Note that, since the number $N$ of training data points in \Cref{prop:classical-description-upper} is independent of the desired accuracy $\epsilon >0$, we can also learn w.r.t.~$\dworst$. In fact, we can even learn the unknown unitary exactly.

We prove \Cref{prop:classical-description-upper} by explicitly constructing the learning algorithms.
We remark that the $r=2^n$ case for entangled data has previously appeared in \cite{sharma2022reformulation}, and a different strategy using mixed states was proposed in \cite{yu2023optimal}.

\begin{proof}
    Let $d=2^n$. 
    We begin by describing the algorithm for entangled data.
    We consider the following set of input states
    \begin{equation}
        \ket{x_j}=\frac{1}{\sqrt{\mathcal{Z}_j}}\sum_{i=(j-1)r+1}^{\min\{jr, d\}}\ket{i}\otimes\ket{i}, \quad j=1, \ldots, \ceil{d/r}.
    \end{equation}
    where the normalization $\mathcal{Z}_j=r$ for $1\leq j\leq\ceil{d/r}-1$ and $\mathcal{Z}_j = d-(\ceil{d/r}-1)r$ for $j=\ceil{d/r}$.
    They all have Schmidt rank at most $r$.
    If we apply $U\otimes I$ on  $\ket{x_j}$, the output state reads
    \begin{equation}
        (U\otimes I)\ket{x_j} = \frac{1}{\sqrt{\mathcal{Z}_j}}\sum_{i=(j-1)r+1}^{\min\{jr, d\}}\sum_{k=1}^d \braket{k}{U|i} \ket{k}\otimes\ket{i}.
    \end{equation} 
    Since we have the classical description, we can directly read off the matrix elements $\braket{k}{U|i}$ with $1\leq k\leq d$ and $(j-1)r+1\leq i\leq \min\{jr, d\}$. Combining different $j$, we can gather all the matrix elements we need to learn $U$.

    Next, we describe the algorithm for mixed state data.
    We consider the input states to be
    \begin{equation}
        \rho_j=\sum_{i=(j-1)r+1}^{\min\{jr, d\}}p_j\ket{ii}\bra{ii}, \quad j=1, \ldots, \ceil{d/r}.
    \end{equation}
    where the uniform mixing probability $p_j=1/r$ for $1\leq j\leq\ceil{d/r}-1$ and $p_j = 1/(d-(\ceil{d/r}-1)r)$ for $j=\ceil{d/r}$.
    Then all $\rho_j$ have rank at most $r$.
    If we apply $U\otimes I$ on  $\ket{x_j}$, the output state becomes
    \begin{equation}
        \rho_j=\sum_{i=(j-1)r+1}^{\min\{jr, d\}} p_j (U\otimes I)\ket{ii}\bra{ii}(U\otimes I)^\dagger, \quad j=1, \ldots, \ceil{d/r}.
    \end{equation}
    We can interpret this output mixed state as randomly choosing a basis state in the ancillary system and applying the unitary to the same state in the original system. 
    Since we have the classical description, we can use the ancillary system as a label for which state we inputted (e.g, $\ket{i}$), and read off all the amplitudes of $U\ket{i}$ on the original system, i.e., a column of the $U$ matrix. 
    Then by combing all the different basis elements $\ket{i}$, $1\leq i\leq d$, we obtain all the matrix elements of $U$.
\end{proof}

Now we move on to the lower bound, which states that any noise-robust unitary learning algorithm needs at least $\Omega(2^n/r)$ samples to learn an arbitrary unknown unitary from classically described data.
The noise-robust requirement here is in accordance with realistic learning scenarios where the tomography of input and output states necessarily involves reconstruction imperfection and noise.
Specifically, we have the following proposition.

\begin{prop}[Lower bounds in learning with classical descriptions, restatement of lower bounds in \Cref{thm:classical-description}]
\label{prop:classical-description-lower}
Let $\epsilon\in (0, 1), \eta=\Theta(\epsilon)$. Let $H_\mathrm{entangle}$ be any learning algorithm that, for any $n$-qubit unitary $U\in U(2^n)$, uses classically described data $\{(\ket{x_i}, \ket{y_i})\}_{i=1}^N$, where $\ket{x_i}$ are bipartite entangled states over two $n$-qubit systems with Schmidt rank at most $r$ and $\ket{y_i}$ are $\eta$-noisy versions of $(U\otimes I)\ket{x_i}$ satisfying $\dtr(\ket{y_i}, (U\otimes I)\ket{x_i})\leq \eta$, to output a $\hat{U}$ such that $\davg(\hat{U}, U)\leq \epsilon$.
Then $H_\mathrm{entangle}$ needs at least $N\geq \Omega(2^n/r)$ samples.

Similarly, let $H_\mathrm{mixed}$ be any learning algorithm that, for any $n$-qubit unitary $U\in U(2^n)$, uses classically described data $\{(\rho_i, \sigma_i)\}_{i=1}^N$, where $\rho_i$ are classically mixed states over two $n$-qubit systems with rank at most $r$ of the form 
\begin{equation}
\label{eq:mixed-input-state}
    \rho_i = \sum_{j=1}^r p_{ij} \ketbra{\phi_{ij}} \otimes \ketbra{\psi_{ij}}
\end{equation}
and $\sigma_i$ are $\eta$-noisy versions of $(U\otimes I)\rho_i(U\otimes I)^\dagger$ satisfying $\dtr(\sigma_i, (U\otimes I)\rho_i(U\otimes I)^\dagger))\leq \eta$, to output a $\hat{U}$ such that $\davg(\hat{U}, U)\leq \epsilon$.
Then $H_\mathrm{mixed}$ needs at least $N\geq \Omega(2^n/r)$ samples.
\end{prop}

\Cref{prop:classical-description-lower} is a consequence of the following information-theoretic reformulation of the quantum no-free-lunch theorem.
The intuition behind this theorem is simple.
On the one hand, to learn a unitary, we have to gather enough information to specify it. This required amount of information is quantified by the metric entropy of the unitary class.
On the other hand, the information provided by each sample is limited and can be characterized by the metric entropy of the output state space.
Therefore, the number of samples needed to learn the unitary is given by the former divided by the latter.
In particular, we can see that the data requirement can be reduced if we increase the amount of information carried by each sample, represented by the metric entropy in the denominator.

\begin{theorem}[Information-theoretic reformulation of quantum no-free-lunch theorem]
\label{thm:qnfl}
Let $\eta, \epsilon\in (0, 1)$.
Let $S$ be a set of input states (possibly with ancillas) and $P$ be a distribution over $S$. Let $\{\rho_i \}_{i=1}^N\subset S^N$ be $N$ classically described input states.
Suppose that after applying the unknown $n$-qubit unitary $U$ from a class $\mathcal{U}\subseteq U(2^n)$ of unitaries, they are transformed into the output states $\{\sigma_i\}_{i=1}^N$ through the map $f_U:\rho_i\mapsto\sigma_i = f_U(\rho_i)$.
Let $\tilde{\sigma}_i$ be an $\eta$-noisy version of $\sigma_i$ satisfying $\dtr(\tilde{\sigma}_i, \sigma_i)\leq \eta$.
Let $\mathcal{N}_\eta=\sup_{\rho\in S} \mathcal{N}(\{f_V(\rho): V\in \mathcal{U}\}, \dtr, \eta)$ be the maximal covering number of the set of all possible output states with different unitary acting on the input states.
Let $\mathcal{F}_\mathcal{U}=\{f_V: V\in\mathcal{U}\}$ be the set of maps and $d_{P}(f_V, f_W) = \sqrt{\mathbb{E}_{\rho\sim P}[\dtr(f_V(\rho), f_W(\rho))^2]}$ be the root mean squared trace distance.
Then any learning algorithm $H$ that uses the $\eta$-noisy classically described data $\{\rho_i, \tilde{\sigma}_i\}_{i=1}^N$ and outputs a $\hat{U}$ such that $d_{P}(f_{\hat{U}}, f_U)\leq \epsilon$ with probability at least $2/3$ needs at least 
\begin{equation}
    N\geq \Omega\left(\frac{\log\mathcal{M}(\mathcal{F}_\mathcal{U}, d_{P}, 2\epsilon+6\eta)}{\log\mathcal{N}_\eta}\right)
\end{equation}
samples.

In particular, if $\eta=\Theta(\epsilon)$, $\mathcal{U}=U(2^n)$, $P$ is a locally scrambled ensemble up to the second moment over $n$-qubit pure states (e.g., $n$-qubit Haar measure), $S$ is the support of $P$, and $f_U(\rho)=U \rho U^\dagger$, then at least $\Omega(2^n)$ samples are needed.
\end{theorem}

We remark that $\eta=\Theta(\epsilon)$ is a convenient choice of noise level for stating the results, but in fact a weaker assumption $\log(1/\eta)=\Theta(\log(1/\epsilon))$ suffices.

In the following, we will first show that \Cref{thm:qnfl} implies \Cref{prop:classical-description-lower} (the lower bounds in \Cref{thm:classical-description}). Then we will turn to the proof of \Cref{thm:qnfl}.

\begin{proof}[Proof of \Cref{prop:classical-description-lower}]
    In both cases (entangled or mixed), we prove the $\Omega(2^n/r)$ lower bound in two steps via \Cref{thm:qnfl}: (1) show that the numerator $\log\mathcal{M}(\mathcal{F}_\mathcal{U}, d_{P}, 2\epsilon+6\eta)$ in \Cref{thm:qnfl} is at least $\Omega(4^n\log(1/\epsilon))$; and (2) show that the denominator $\mathcal{N}_\epsilon$ is at most $\bigo(2^n r\log(1/\epsilon))$ when the input states are either entangled pure states of Schmidt rank at most $r$ or mixed states of rank at most $r$.
    Then the desired results follow.
    
    For step (1), we begin by defining the distribution $P$ with respect to which the performance in \Cref{thm:qnfl} is measured. 
    For both entangled and mixed cases, we define $P$ to be the distribution of $\ket{\psi}\otimes \ket{0}^{\otimes n}$ where $\ket{\psi}$ is a Haar-random state on the original system, and $\ket{0}^{\otimes n}$ is a fixed state on the ancillary system.
    Note that this state is indeed both a bipartite entangled state with Schmidt rank at most $r$ and of the form \Cref{eq:mixed-input-state} with rank at most $r$. 
    Moreover, since in both cases, the map $f_U$ is given by acting the unitary $U$ on the original system and the identity on the ancillary system, the distance metric $d_{P}(f_V, f_W)$ is the same as $\davg$.
    Therefore, the packing number satisfies
    \begin{equation}
        \mathcal{M}(\mathcal{F}_{U(2^n)}, d_{P}, 2\epsilon+6\eta) = \mathcal{M}(U(2^n), \davg, 2\epsilon+6\eta).
    \end{equation}

    To find the packing number $\mathcal{M}(U(2^n), \davg, 2\epsilon+6\eta)$, we invoke the covering number bound for $U(2^n)$ with respect to the normalized Frobeinus norm $d_F$ (\Cref{lem:2q-unitaries}), the fact that quotient out global phase only change the metric entropy by a constant (\Cref{lem:quotient-packing}), and the equivalence of $d_F'$ and $\davg$ (\Cref{lem:dist-df'}, Item 1). We have
    \begin{equation}
        \log \mathcal{M}(\mathcal{F}_{U(2^n)}, d_{P}, 2\epsilon+6\eta)= \log\mathcal{M}(U(2^n), \davg, 2\epsilon+6\eta)\geq \Omega\left(4^n\log\frac{1}{\epsilon}\right),
    \end{equation}
    where we used $\eta=\Theta(\epsilon)$.

    Next, for step (2), we compute $\mathcal{N}_\eta$.
    For entangled data, note that applying unitaries on only the first $n$ qubits does not change the bipartite Schmidt rank $r$, so the output states are pure states of the form $\ket{\chi} = \sum_{i, j=1}^{2^n} A_{ij}\ket{i}\otimes \ket{j}$, where $\|A\|_F=1$ because of normalization, and the rank of $A$ corresponds to the Schmidt rank which is at most $r$.
    Furthermore, the Euclidean distance between the output states is equal to the Frobenius distance between the corresponding $A$-matrices.
    With this correspondence, we can explicitly construct a covering net over the output states as follows.
    We take a minimal $\eta$-covering net $\mathcal{N}'$ over the set of complex matrices $A$ with bounded rank $r$ and $\|A\|_F=1$ with respect to the Frobenius distance.
    Since they are contained in the unit ball ($\|A\|_F\leq 1$) in a real linear space of dimension $2\cdot 2^n\cdot r$ \cite[Theorem 1]{flanders1962on}, by the monotinicity of covering number and the standard covering number bound for Euclidean balls via a volume argument \cite[Corollary 4.2.13]{vershynin2018high}, we have $\log|\mathcal{N}'|\leq \mathcal{O}(2^nr\log(1/\eta))$.
    Meanwhile, similar to the proof in \Cref{lem:dist-spectral-diamond}, the trace distance between any two pure states $\ket{\psi}, \ket{\phi}$ are bounded by the Euclidean distance, and thus the Frobenius distance between the corresponding $A$ matrices:
    \begin{equation}
        \dtr(\ket{\psi}, \ket{\phi}) = \sqrt{1-|\braket{\psi}{\phi}|^2}\leq \sqrt{2(1-|\braket{\psi}{\phi}|)}\leq \sqrt{2(1-\mathrm{Re}[\braket{\psi}{\phi}])} = \norm{\ket{\psi}-\ket{\phi}}_2.
    \end{equation}
    Therefore $\mathcal{N}'$ gives an $\eta$-covering net over the output states with respect to the trace distance $\dtr$.
    Hence, $\log\mathcal{N}_\eta\leq \log|\mathcal{N}'| \leq \mathcal{O}(2^nr\log(1/\eta)) = \mathcal{O}(2^nr\log(1/\epsilon))$ since $\eta=\Theta(\epsilon)$, and from \Cref{thm:qnfl} we have the desired lower bound
    \begin{equation}
        N\geq \Omega\left(\frac{4^n\log(1/\epsilon)}{2^n r\log(1/\epsilon)}\right)=\Omega\left(\frac{2^n}{r}\right).
    \end{equation}

    The case of mixed states is similar.
    For a given input state $\rho = \sum_{i=1}^r p_i\ketbra{\phi_i}\otimes \ketbra{\psi_i}$, the output state reads
    \begin{equation}
        \sigma = \sum_{i=1}^r p_iU\ketbra{\phi_i}U^\dagger\otimes \ketbra{\psi_i}.
    \end{equation}
    Now we take a minimal $\eta$-covering net $\mathcal{N}''$ over all pure $n$-qubit states with respect to Euclidean distance, which is a unit ball in a $2\cdot 2^n$ dimensional real linear space. 
    By standard covering number bound for Euclidean balls, we know $\log|\mathcal{N}''|\leq \mathcal{O}(2^n\log(1/\eta))=\mathcal{O}(2^n\log(1/\epsilon))$. Then for any $U\ket{\phi_i}$, there exists a $\ket{\eta_i}\in\mathcal{N}''$ such that $\norm{U\ket{\phi_1}-\ket{\eta_1}}_2\leq \eta$.
    Let $\sigma'=\sum_{i=1}^r p_i\ketbra{\eta_i}\otimes \ketbra{\psi_i}$. 
    Then the trace distance is bounded by
    \begin{equation}
    \begin{split}
        \frac{1}{2}\|\sigma-\sigma'\|_1 &\leq \frac{1}{2}\sum_{i=1}^r p_i\norm{U\ketbra{\phi_i}U^\dagger\otimes \ketbra{\psi_i} - \ketbra{\eta_i}\otimes \ketbra{\psi_i}}_1 \\
        &\leq \frac{1}{2}\sum_{i=1}^r p_i \norm{U\ket{\phi_i}\otimes\ket{\psi_i} - \ket{\eta_i}\otimes \ket{\psi_i}}_2 \\
        &=\frac{1}{2}\sum_{i=1}^r p_i \norm{U\ket{\phi_i} - \ket{\eta_i}}_2 \leq \frac{\eta}{2}\sum_{i=1}^r p_i=\frac{\eta}{2},
    \end{split}
    \end{equation}
    where we have used the subadditivity of trace norm, the fact that trace distance is upper bounded by Euclidean norm for pure states, and $\sum_{i=1}^r p_i=1$.
    Hence the set 
    \begin{equation}
        \left\{\sum_{i=1}^r p_i\ketbra{\eta_i}\otimes \ketbra{\psi_i}: \ket{\eta_i}\in \mathcal{N}'', 1\leq i\leq r\right\}
    \end{equation} 
    forms an $\eta/2$-covering net of set of the output states and has cardinality $|\mathcal{N}''|^r$.
    Therefore, we have $\log\mathcal{N}_\eta = r\log|\mathcal{N}''| = \mathcal{O}(2^n r\log(1/\epsilon))$.
    From \Cref{thm:qnfl}, we again arrive at the desired result \begin{equation}
        N\geq \Omega\left(\frac{4^n\log(1/\epsilon)}{2^n r\log(1/\epsilon)}\right)=\Omega\left(\frac{2^n}{r}\right).
    \end{equation}
    This concludes the proof of \Cref{prop:classical-description-lower}, and together with \Cref{prop:classical-description-upper}, we have proved \Cref{thm:classical-description}.
\end{proof}

Now we move on to prove our quantum no-free-lunch theorem (\Cref{thm:qnfl}).
We first establish the following information-theoretic lower bound on the sample complexity of learning discrete functions.
We remark that a version for binary-valued functions was proved in a different fashion in \cite[Proposition 8]{eskenazis2022low} and \cite[Lemma 4.8]{benedek1991learn}.

\begin{prop}[Information-theoretic lower bound for learning discrete functions]
\label{prop:learning-discrete-func}
    Let $\epsilon>0$, $k\in\mathbb{N}$ and $\mathcal{F}$ be a class of functions mapping $\mathcal{X}$ to $\mathcal{Y}=\{1, \ldots, k\}$ with a distance metric $d$.
    Any learning algorithm $H$ that uses $N$ samples $\{x_i\in\mathcal{X}, y_i=f(x_i)\}_{i=1}^N$ and outputs an $\hat{f}$ such that $d(\hat{f}, f)\leq \epsilon$ with probability at least 2/3 for any $f\in\mathcal{F}$ must use at least 
    \begin{equation}
        N\geq \Omega\left(\frac{\log\mathcal{M}(\mathcal{F}, d, 2\epsilon)}{\log k}\right)
    \end{equation}
    samples.
\end{prop}

\begin{proof}[Proof of \Cref{prop:learning-discrete-func}]
    We begin by taking a maximal $2\epsilon$-packing $\mathcal{P}$ of $\mathcal{F}$, i.e., $|\mathcal{P}| = \mathcal{M}(\mathcal{F}, d, 2\epsilon)$ and for any $f_i\neq f_j \in \mathcal{P}$, $d(f_i, f_j)>2\epsilon$. 
    Now we design a communication protocol between two parties, Alice and Bob, as follows. 
    The packing $\mathcal{P}$ is shared by both parties. 
    Alice takes a random variable $W$ uniformly sampled from $\{1, \ldots, \mathcal{M}(\mathcal{F}, d, 2\epsilon)\}$ and picks the corresponding function $f_W$ from the packing $\mathcal{P}$.
    She then feeds the inputs $x=(x_1, \ldots, x_N)$ into $f_W$, generating a dataset $Z = ((x_1, f_W(x_1)), \ldots, (x_N, f_W(x_N)))$, and sends the dataset to Bob. 
    Bob's task is to use this dataset to determine which function Alice used. 
    Suppose Bob is given a learning algorithm described as in the proposition. The algorithm will learn from the dataset and output a hypothesis function $\hat{f}$ that satisfies
    \begin{equation}
        \mathbb{P}[d(\hat{f}, f_W)\leq \epsilon]\geq 2/3 \, ,
    \end{equation}
    no matter which $W$ was chosen by Alice.
    With $\hat{f}$ in hand, Bob can make the guess
    \begin{equation}
        \hat{W} = \mathrm{argmin}_{f_w \in \mathcal{P}}d(\hat{f}, f_w).
    \end{equation}
    Note that as long as $d(\hat{f}, f_W) \leq \epsilon$, then for any $f_i\neq f_W \in \mathcal{P}$,
    \begin{equation}
        d(\hat{f}, f_i) \geq d(f_W, f_i) - d(\hat{f}, f_W) > 2\epsilon-\epsilon=\epsilon \geq d(\hat{f}, f_W).
    \end{equation}
    Therefore, the error probability of Bob's guess is bounded by
    \begin{equation}
        \mathbb{P}[\hat{W}\neq W] 
        = \mathbb{P}[\exists i\neq W: d(\hat{f}, f_i) \leq d(\hat{f}, f_W)] \leq \mathbb{P}[d(\hat{f}, f_W) > \epsilon] \leq 1/3.
    \end{equation}
    By Fano's inequality \cite[Theorem 2.10.1]{cover1999elements}, the conditional entropy $S(W|\hat{W}) \leq s_2(1/3) + \frac{1}{3} \log \mathcal{M}(\mathcal{F}, d, 2\epsilon)$, where $s_2(\delta) = -\delta\log\delta-(1-\delta)\log(1-\delta)$ is the binary entropy function. Then the mutual information is at least
    \begin{equation}
        I(W;\hat{W}) = S(W) - S(W|\hat{W}) \geq (2/3) \log \mathcal{M}(\mathcal{F}, d, 2\epsilon) - s_2(1/3).
    \end{equation}
    On the other hand, the sample size $N$ controls the amount of information that Bob has access to. 
    Since Bob's guess is produced by the dataset $Z$, by the data processing inequality \cite[Theorem 2.8.1]{cover1999elements}, we have
    \begin{equation}
        I(W;\hat{W}) \leq I(W;Z) = S(Z)-S(Z|W) = S(f_W(x_1), \ldots, f_W(x_N)) \leq N\log k,
    \end{equation}
    where we used $S(Z|W)=0$, since $Z$ is determined by $W$, and the fact that $(f_W(x_1), \ldots, f_W(x_N))$ can take no more than $k^N$ different values.
    Combining the above two inequalities, we arrive at
    \begin{equation}
        N \geq \Omega\left(\frac{\mathcal{M}(\mathcal{F}, d, 2\epsilon)}{\log k}\right).
        \qedhere
    \end{equation}
\end{proof}

With \Cref{prop:learning-discrete-func}, we can prove \Cref{thm:qnfl} by quantizing the output states to the nearest elements in covering nets, similar to an idea employed in \cite{bartlett1994fat}.

\begin{proof}[Proof of \Cref{thm:qnfl}]
    Let $k=\mathcal{N}_\eta$.
    Since $\mathcal{N}_\eta\geq \mathcal{N}(\{f_V(\rho), V\in \mathcal{U}\}, \dtr, \eta)$ for every $\rho\in S$, we can find an $\eta$-covering net $\mathcal{N}_\rho$ of size $k$ for each $\rho\in S$.
    We label the elements of $\mathcal{N}_\rho$ using $\{1, \ldots, k\}$ and define $L_\rho(\sigma)\in [k]$ as the label of a covering net element $\sigma\in \mathcal{N}_\rho$.

    Now we define the quantized function $Qf_U$ that maps an input state $\rho$ to an element of the covering net $\mathcal{N}_\rho$.
    Specifically, for any $\rho\in S$ and any $\sigma \in \{f_V(\rho), V\in \mathcal{U}\},$ there exists a $\sigma' \in \mathcal{N}_\rho$, such that $\dtr(\sigma, \sigma')\leq \eta$. 
    For any unitary $U\in\mathcal{U}$, we define
    \begin{equation}
        Qf_U(\rho) = \mathrm{argmin}_{\sigma \in \mathcal{N}_\rho} \dtr(f_U(\rho), \sigma)
    \end{equation}
    and $LQf_U(\rho) = L_\rho[Qf_U(\rho)]$ be the corresponding label. (Ties are broken arbitrarily.)
    Then $LQf_U$ is a discrete-output function mapping input states $S$ to labels $[k]$ and it is in one-to-one correspondence with $Qf_U$.
    We use $\mathcal{F}^{Q}$ to denote all these labeled quantized functions, $\mathcal{F}^{Q} = \{LQf_U, U\in \mathcal{U}\}$, and define the distance metric on labeled functions as $d_L(LQf_V, LQf_W) = d_{P}(Qf_V, Qf_W)$. 
    A useful property is that for any unitary $U\in\mathcal{U}$, we have 
    \begin{equation}
        d_{P}(f_U, Qf_U) = \sqrt{\mathbb{E}_{\rho\sim P}[\dtr(f_U(\rho), Qf_U(\rho))^2]} \leq \sqrt{\mathbb{E}_{\rho\sim P}[\eta^2]}=\eta.
    \end{equation}

    Now we claim that if there exists a noise-robust learning algorithm $H$ for $\mathcal{U}$ to accuracy $\epsilon$ in $d_{P}$ with probability at least $2/3$, then we can use it to construct a learning algorithm $H^{Q}$ for $\mathcal{F}^{Q}$ to accuracy $\epsilon+2\eta$ in $d_L$ with success probability at least $2/3$. Hence, the sample complexity for $\mathcal{U}$ must satisfy
    \begin{equation}
        N \geq \Omega\left(\frac{\log\mathcal{M}(\mathcal{F}^{Q}, d_L, 2\epsilon+4\eta)}{\log k}\right),
    \end{equation}
    by \Cref{prop:learning-discrete-func}.

    To show this claim, we construct $H^{Q}$ as follows.
    For any $LQf_U\in\mathcal{F}^{Q}$, let the dataset be
    \begin{equation}
        Z = (\rho_1, Qf_U(\rho_1)), \ldots, (\rho_N, Qf_U(\rho_N)).
    \end{equation}
    From the definition of quantized functions, we know that the $Qf_U(\rho_i)$ are $\eta$-noisy version of $f_U(\rho_i)$ because $\dtr(Qf_U(\rho_i), f_U(\rho_i))\leq \eta$.
    Now we define $H^Q$ as 
    \begin{equation}
        H^{Q}[Z] = LQf_{H[Z]}.
    \end{equation}
    Since the learning algorithm $H$ is $\eta$-noise-robust, we have $d_{P}(f_{H[Z]}, f_{U})\leq \epsilon$ and thus $d_{P}(f_{H[Z]}, Qf_U)\leq d_{P}(f_{H[Z]}, f_U) + d_{P}(f_U, Qf_U)\leq \epsilon+\eta$ with probability at least $2/3$.
    Then by the triangle inequality (proved similarly as in \Cref{lem:tri-ineq-avg}), we have
    \begin{equation}
        d_L(H^Q[Z], LQf_U) 
        =d_{P}(Qf_{H[Z]}, Qf_U) 
        \leq d_{P}(Qf_{H[Z]}, f_{H[Z]}) + d_{P}(f_{H[Z]}, Qf_U) 
        \leq \epsilon+2\eta
    \end{equation}
    with probability at least $2/3$.
    Thus the claim is proved.

    At this point, it remains to prove that
    \begin{equation}
        \mathcal{M}(\mathcal{F}^Q, d_L, 2\epsilon+4\eta) \geq \mathcal{M}(\mathcal{F}_\mathcal{U}, d_P, 2\epsilon+6\eta).
    \end{equation}
    To prove this, we can take a maximal $(2\epsilon+6\eta)$-packing $\mathcal{P}$ of $\mathcal{F}_\mathcal{U}$ with respect to $d_{P}$, with $|\mathcal{P}|=\mathcal{M}(\mathcal{F}_\mathcal{U}, d_{P}, 2\epsilon+6\eta)$. Then $\forall f_{U_1} \neq f_{U_2} \in \mathcal{P}$, we have
    \begin{equation}
        2\epsilon+6\eta < d_{P}(f_{U_1}, f_{U_2}) \leq d_{P}(f_{U_1}, Qf_{U_1}) + d_{P}(Qf_{U_1}, Qf_{U_2}) + d_{P}(Qf_{U_2}, U_2) \leq 2\eta + d_{P}(Qf_{U_1}, Qf_{U_2}).
    \end{equation}
    Therefore, $d_L(LQf_{U_1}, LQf_{U_2})=d_P(Qf_{U_1}, Qf_{U_2}) > 2\epsilon+4\eta$. Hence, 
    \begin{equation}
        \mathcal{M}(\mathcal{F}^{Q}, d_L, 2\epsilon+4\eta) \geq |\{LQf_U, U\in \mathcal{P}\}| = |\mathcal{P}| = \mathcal{M}(\mathcal{F}_\mathcal{U}, d_{P}, 2\epsilon+6\eta).
    \end{equation}
    This concludes the proof of the main part in \Cref{thm:qnfl}.

    Finally, we illustrate the special case where $\eta=\Theta(\epsilon)$, $\mathcal{U}=U(2^n)$ is the whole unitary group, $P$ is a locally scrambled ensemble up to the second moment over $n$-qubit pure states (e.g., $n$-qubit Haar measure, see \Cref{def:local-scram-ensem}), $S$ is the support of $P$, and $f_U(\rho) = U\rho U^\dagger$.
    We show that at least $\Omega(2^n)$ samples are needed, thus reproducing the quantum no-free-lunch theorem in the usual sense and generalizing it to locally scrambled ensembles.
    
    To see this, we first compute $\log\mathcal{M}(\mathcal{F}_{U(2^n)}, d_{P}, 2\epsilon+6\eta)$. From the covering number bound for $U(2^n)$ with respect to the normalized Frobeinus norm $d_F$ (\Cref{lem:2q-unitaries}), the fact that quotienting out the global phase only changes the metric entropy by an additive $\mathcal{O}(\log(1/(2\epsilon+6\eta)))$ term (\Cref{lem:quotient-packing}), and by the equivalence of $d_F'$, $\davg$, and $d_P$ (\Cref{lem:dist-df'} Item 1 and \Cref{lem:equi_local_scram}), we know that $\log\mathcal{M}(\mathcal{F}_{U(2^n)}, d_P, 2\epsilon+6\eta)\geq \Omega\left(4^n\log(1/\epsilon)\right)$, where we used $\eta=\Theta(\epsilon)$.
    
    Next, we move on to $\mathcal{N}_\eta$.
    Since the output states are still $n$-qubit pure states, $\mathcal{N}_\eta$ is the covering number of the set of pure states with respect to $\dtr$.
    Considering that $\frac{1}{2}\norm{\ketbra{\psi}}_1$ is less than one for any pure state $\ket{\psi}$, the covering number is upper bounded by the covering number of a unit Euclidean ball in a $\Theta(2^n)$ dimensional linear space.
    Therefore, we have $\log\mathcal{N}_\eta\leq \bigo(2^n\log(1/\eta)) = \bigo(2^n\log(1/\epsilon))$ since $\eta=\Theta(\epsilon)$.
    Hence we arrive at
    \begin{equation}
        N\geq \Omega\left(\frac{4^n\log(1/\epsilon)}{2^n\log(1/\epsilon)}\right) = \Omega(2^n).
    \end{equation}
    This concludes the proof of \Cref{thm:qnfl}.
\end{proof}

The information theoretic version of quantum no-free-lunch theorem (\Cref{thm:qnfl}) also gives us a way to generalize quantum no-free-lunch to a restricted unitary class.
For example, for unitaries with bounded circuit complexity $G$, the packing number in the enumerator is lower bounded by $\Omega(G)$, while the covering number in the denominator is upper bounded by $\bigo(\min\{G\log G+G\log n, 2^n\})$.
This gives us a quantum no-free-lunch theorem for $G$-gate unitaries, where the sample complexity is lower bounded by $\Omega(1)$ for $G\leq \bigo(2^n)$, by $\Omega(G/2^n)$ for $\Omega(2^n)<G\leq \bigo(4^n)$ and $\Omega(2^n)$ for $G\geq \Omega(4^n)$.

\subsection{Computational complexity}
\label{app:unitary-comp-complexity}

Similar to the state learning case, our algorithm for average-case unitary learning described in \Cref{app:avg-case-unitary-up} is not computationally efficient.
In this section, we follow \Cref{app:states-comp-complexity} and first show that there is no polynomial-time algorithm for learning unitaries composed of $G = \mathcal{O}(n\polylog(n))$ two-qubit gates, assuming $\mathsf{RingLWE}$ cannot be solved efficiently on a quantum computer.
This result also holds for unitaries with circuit depth $\mathcal{O}(\polylog(n))$.
Then we invoke a stronger assumption that $\mathsf{RingLWE}$ cannot be solved by any sub-exponential-time quantum algorithm, and show that any quantum algorithm for learning unitaries composed of $\tilde{\bigo}(G)$ gates must use $\exp(\Omega(G))$ time.
Finally, we explicitly construct an efficient learning algorithm for $G=\mathcal{O}(\log n)$, thus establishing $\log n$ gate complexity as a transition point of computational efficiency.

\begin{theorem}[Unitary learning computational complexity lower bound assuming polynomial hardness of \textsf{RingLWE}]
\label{thm:unitary-comp-complexity}
Let $\lambda=n$ be the security parameter.
Let $U$ be a unitary consisting of $G = \mathcal{O}(n\polylog(n))$ gates (or a depth $d = \mathcal{O}(\polylog(n))$ circuit) that implements a pseudorandom function in $\mathcal{RF}$.
Such a unitary $U$ exists by \Cref{coro:PRF-size}.
There exists no polynomial-time quantum algorithm for learning a circuit description of $U$ to  within $\epsilon\leq 1/64$ average-case distance $\davg$ with probability at least $2/3$ from $N = \poly(\lambda)$ queries, if quantum computers cannot solve $\mathsf{RingLWE}$ in polynomial time.
\end{theorem}

\begin{proof}
   Suppose for the sake of contradiction that there is an efficient algorithm $\mathcal{A}_0$ that can learn a description of $U$ to within $\epsilon$ average-case distance with probability at least $2/3$.
   Then by standard boosting of success probability (see e.g, \cite[Proposition 2.4]{haah2023query}), there is an efficient algorithm $\mathcal{A}$ that can learn $U$ to the same accuracy with probability at least $p=1-1/8192$ with only a constant factor overhead in time complexity.
   Note that this boosting requires the distance metric to be efficiently computable, which is guaranteed by the SWAP test elaborated below.
   We will construct a polynomial-time quantum distinguisher $\mathcal{D}$ that invokes $\mathcal{A}$ to distinguish between $U$ and the unitary $V\in\mathcal{U}$ corresponding to a random classical function.
   This contradicts \Cref{thm:PRF} Item 2.

   The distinguisher $\mathcal{D}$ operates according to Algorithm~\ref{algo:distinguish-unitary}.
   \begin{algorithm}
        \label{algo:distinguish-unitary}
        \caption{Distinguisher $\mathcal{D}$ for PRF}
        \KwIn{$N$ query access to $U$}
        \KwOut{$b \in \{0,1\}$}
        Run $\mathcal{A}$ using $(N-1)$ queries to $U$, receiving $\hat{U}$.\\
        Prepare a random tensor product of $1$-qubit stabilizer states $\ket{x}, x\in\mathbb{Z}_6^n$.\\
        Query $U$ one more time to prepare $U\ket{x}$.\\
        Run the SWAP test on $U\ket{x}$ and $\hat{U}\ket{x}$, receiving a bit $b \in \{0,1\}$.\\
        Output $b$.
    \end{algorithm}

    Recall that the SWAP test~\cite{barenco1997stabilization,buhrman2001quantum} takes two quantum states $\ket{\alpha}, \ket{\beta}$ as input and outputs $1$ with probability $(1+|\braket{\alpha}{\beta}|^2)/2$.
    We denote this algorithm as $\mathsf{SWAP}(\ket{\alpha},\ket{\beta})$.

    Note that Step 2 in Algorithm~\ref{algo:distinguish-unitary}, the preparation of tensor product of $1$-qubit stabilizer states $\ket{x}, x\in\mathbb{Z}_6^n$, is computationally efficient, because it can be achieved by random one-qubit gates acting on each of the $n$ qubits.
    Moreover, Step 4 can be implemented efficiently on a quantum computer because $\hat{U}$ is given in terms of efficient circuit description and because the SWAP test is efficiently implementable. 
    Thus, assuming the hypothetical learner $\mathcal{A}$ to be efficient, the distinguisher $\mathcal{D}$ is efficient as well.
    
    We analyze the probability that the distinguisher $\mathcal{D}$ outputs $1$ when given the pseudorandom function $U$ versus the random classical Boolean function $V$.
    We denote the distribution of $\ket{x}$ by $Q$.
    From \Cref{lem:equi_local_scram}, we have 
    \begin{equation}
        d_Q(U, \hat{U}) = \sqrt{\mathbb{E}_{\ket{x}\sim Q}[\dtr(U\ket{x}, \hat{U}\ket{x})^2]}\leq \sqrt{2}\davg(U, \hat{U}).
    \end{equation}
    
    \textbf{Case 1: $U\in \mathcal{RF}$.}
    By the guarantees of $\mathcal{A}$, with probability at least $p$, we have $\davg(\hat{U}, U) \leq \epsilon\leq 1/64$, where $\hat{U}$ is the unitary learned by algorithm $\mathcal{A}$.
    This implies
    \begin{equation}
        \label{eq:goodest-unitary}
        \E_{\ket{x}\sim Q}|\bra{x}\hat{U}^\dagger U\ket{x}|^2 = 1-d_Q^2(U, \hat{U})\geq 1-2\epsilon^2,
    \end{equation}
    where we used the relationship between fidelity and trace distance.
    Then it immediately follows from \Cref{eq:goodest-unitary} that
    \begin{equation}
    \begin{split}
        \label{eq:case1-unitary}
        \Pr_{\substack{U\in \mathcal{RF}, \mathcal{D}}}\left[\mathcal{D}^{\ket{U}}(\cdot) = 1 \right] 
        &= \Pr_{\substack{U\in \mathcal{RF}, \ket{x}\sim Q\\\mathcal{A},\mathsf{SWAP}}}\left[\mathsf{SWAP}\left(U\ket{x}, \hat{U}\ket{x}\right) = 1\right] \\
        &= \E_{U\in \mathcal{RF}, \ket{x}\sim Q}\left[\Pr_{\mathcal{A},\mathsf{SWAP}}\left[\left.\mathsf{SWAP}\left(U\ket{x}, \hat{U}\ket{x}\right) = 1\right\rvert U, \ket{x}\right]\right] \\
        &\geq p \E_{U\in \mathcal{RF}}\left[\frac{1}{2} + \frac{1}{2}\E_{\hat{U}, \ket{x}\sim Q}\left[|\bra{x}\hat{U}^\dagger U\ket{x}|^2\right]\right] \\
        &\geq p \E_{U\in \mathcal{RF}}\left[\frac{1}{2} + \frac{1}{2}(1-2\epsilon^2)\right]
        =p(1-\epsilon^2) > \frac{8189}{8192},
    \end{split}
    \end{equation}
    where in the first inequality we split the probability into two terms conditioned on the success and failure of $\mathcal{A}$, and we lower bound the failure term by zero, and in the last inequality we have used the fact that $p(1-\epsilon^2)\geq (1-1/8192)(1-1/4096)> 8189/8192$.

    \textbf{Case 2: $U=V \in \mathcal{U}$}, where $V$ is the $n$-qubit unitary implementing a randomly chosen classical function.
    We want to upper bound the probability that the distinguisher $\mathcal{D}$ outputs $1$ when given queries to $V$.
    Let $\mathcal{C}$ be the set of all possible output unitaries of $\mathcal{A}$. We follow the same reasoning as in \Cref{eq:case1-unitary} and note that
    \begin{equation}
    \begin{split}
        \Pr_{V\in\mathcal{U}, \mathcal{D}}\left[\mathcal{D}^{\ket{V}}(\cdot)=1\right] 
        &\leq 
        \E_{V\in\mathcal{U}}\left[\max_{W\in \mathcal{C}} \E_{\ket{x}\sim Q}\left[\frac{1}{2}+\frac{1}{2}|\bra{x}V^\dagger W\ket{x}|^2\right]\right] + (1-p)\\
        &\leq \E_{V\in\mathcal{U}}\left[\max_{W\in \mathcal{C}} \left[1-\frac{1}{4}\davg(V, W)^2\right]\right] + (1-p) \\
        &\triangleq \E_{V\in\mathcal{U}}\left[O_V\right] + (1-p),
    \end{split}
    \end{equation}
    where we define $O_V=\max_{W\in \mathcal{C}} \left[1-\frac{1}{4}\davg(V, W)^2\right]$.
    Furthermore, we can split the right hand side into two parts by introducing a constant $\theta$:
    \begin{equation}
        \E_{V\in\mathcal{U}}\left[O_V\right] \leq \Pr[O_V\leq 1-\frac{\theta^2}{4}]\cdot \left(1-\frac{\theta^2}{4}\right) + \Pr[O_V>1-\frac{\theta^2}{4}]\cdot 1 \leq 1-\frac{\theta^2}{4} + \Pr[O_V>1-\frac{\theta^2}{4}],
    \end{equation}
    where we have used the fact that $O_V\leq 1$.
    Note that
    \begin{equation}
    \begin{split}
        \Pr[O_V>1-\frac{\theta^2}{4}] 
        &\leq \Pr_{V\in \mathcal{U}}\left[\exists W\in\mathcal{C}: \davg(V, W)< \theta\right] \\
        &\leq \sum_{W\in \mathcal{N}}\Pr_{V\in \mathcal{U}}\left[\davg(V, W)< \theta\right] \\
        &=\sum_{W\in \mathcal{N}}\frac{1}{|\mathcal{U}|}\sum_{V\in \mathcal{U}}1\left\{\davg(V, W)< \theta\right\} \\
        &\leq \frac{|\mathcal{N}|\max_{W\in \mathcal{N}}N_{W, \theta}}{|\mathcal{U}|}.
    \end{split}
    \end{equation}
    In the second line, we define $\mathcal{N}$ be a minimal $\theta$-covering net over $\mathcal{C}$ with respect to $\davg$. 
    Also, in the last line, we define $N_{W, \theta}\triangleq\sum_{V\in \mathcal{U}}1\{\davg(V, W)< \theta\}$ to be the number of $V\in\mathcal{U}$ that are $\theta$-close to $W$ in $\davg$.

    Now we aim to upper bound $N_{W, \theta}$ by counting.
    We first note that $N_{W, \theta}\leq \max_{V\in \mathcal{U}}N_{V, 4\theta}+1$.
    This is because, by definition of $N_{W, \theta}$, there exist $V_1, \ldots, V_{N_{W, \theta}}\in\mathcal{U}$ such that $\davg(V_i, W)< \theta, 1\leq i\leq N_{W, \theta}$.
    Then for $V_1$ and any $V_i, 2\leq i\leq N_{W, \theta}$, we have
    \begin{equation}
        \davg(V_1, V_i)\leq d_F'(V_1, V_i)\leq d_F'(V_1, W)+d_F'(V_i, W) \leq 2\davg(V_1, W)+2\davg(V_i, W) < 4\theta.
    \end{equation}
    This means that there are at least $N_{W, \theta}-1$ elements of $\mathcal{U}$ that are $(4\theta)$-close to $V_1$.
    Therefore, $N_{V_1, 4\theta}\geq N_{W, \theta}-1$ and hence $N_{W, \theta}\leq \max_{V\in \mathcal{U}}N_{V, 4\theta}+1$.

    Next, we upper bound $N_{V, 4\theta}$ for any $V\in \mathcal{U}$.
    Recall that each $V\in \mathcal{U}$ is an oracle unitary of a Boolean function on $\{0, 1\}^n$.
    We can represent it by $f_V(i)\in\{0, 1\}, 1\leq i\leq 2^n$.
    Consider a different $V'\in\mathcal{U}$ corresponding to the Boolean function $f_{V'}$.
    If $f_V$ and $f_{V'}$ differ on at least $\ceil{64\theta^2\cdot 2^n}$ of the $2^n$ possible inputs $i\in [2^n]$, then the corresponding columns of the unitaries $V$ and $V'$ must also differ.
    In particular, in each of these columns, there will be a matrix element that is $1$ for $V$ but $0$ for $V'$.
    This means that $V$ and $V'$ are $4\theta$ apart from each other w.r.t.~$\davg$:
    \begin{equation}
        \davg(V, V')\geq \frac{1}{2}\min_{e^{i\phi}\in U(1)}\norm{V - V'e^{i\phi}}_F \geq \frac{1}{2\sqrt{d}}\min_{e^{i\phi}\in U(1)}\sqrt{64\theta^2 \cdot 2^n |1-0\cdot e^{i\phi}|^2} = 4\theta.
    \end{equation}
    Therefore, all functions $f_{V'}$ corresponding to the $V'\in \mathcal{U}$ counted in $N_{V, 4\theta}$ must differ from $f_V$ on strictly less than $\ceil{64\theta^2\cdot 2^n}$ of the $2^n$ inputs.
    This gives us
    \begin{equation}
        N_{V, 4\theta}\leq \sum_{k=0}^{\ceil{64\theta^2\cdot 2^n}} \binom{2^n}{k},
    \end{equation}
    where each term represents choosing $k$ inputs where the output is different from $f_V$.
    The right hand side can be further bounded as 
    \begin{equation}
        \sum_{k=0}^{\ceil{64\theta^2\cdot 2^n}} \binom{2^n}{k} \leq \left(\frac{e2^n}{\ceil{64\theta^2\cdot 2^n}}\right)^{\ceil{64\theta^2\cdot 2^n}}\leq 2^{(64\theta^2\cdot 2^n+1)\log_2(e/64\theta^2)}.
    \end{equation}
    Note that when $\theta=1/16$, we have $64\theta^2=1/4$ and $64\theta^2\log_2(e/64\theta^2)=\log_2(4e)/4<0.87$.
    Therefore, recalling that the set of all $n$-bit classical Boolean functions has size  $|\mathcal{U}|=2^{2^n}$, we obtain
    \begin{equation}
        \Pr[O_V>1-\frac{\theta^2}{4}] \leq |\mathcal{N}|2^{-0.13\cdot 2^n+\log_2(4e)+1},
    \end{equation}
    where the extra one in the exponent takes the one in $N_{W, \theta}\leq \max_{V\in \mathcal{U}}N_{V, 4\theta}+1$ into account.

    Finally, we move on to bound $|\mathcal{N}|$.
    Similar to \Cref{eq:covering-poly-state} in the state learning case, since our learning algorithm is a polynomial time algorithm that can only output circuit descriptions with size $\poly(n)$, we must have
    \begin{equation}
        |\mathcal{N}| \leq \bigo\left(\left(1/\theta\right)^{\poly(n)}\right) = \bigo\left(2^{\poly(n)}\right).
    \end{equation}
    Thus we arrive at
    \begin{equation}
    \label{eq:case2-unitary}
        \Pr[O_V>1-\frac{\theta^2}{4}] \leq \bigo\left(2^{\poly(n)-0.13\cdot 2^n}\right)=\mathrm{negl}(n)
    \end{equation}
    and therefore
    \begin{equation}
        \Pr\left[\mathcal{D}^{\ket{V}}(\cdot)=1\right] \leq 1-\frac{\theta^2}{4} + \mathrm{negl}(n) + (1-p) = \frac{8185}{8192} + \mathrm{negl}(n),
    \end{equation}
    where we have used $\theta=1/16$ and $p=1-1/8192$.

    Combining \Cref{eq:case1-unitary} and \Cref{eq:case2-unitary}, we have
    \begin{equation}
        \left|\Pr_{U \in \mathcal{RF}}[\mathcal{D}^{\ket{U}}(\cdot) = 1] - \Pr_{V \in \mathcal{U}}[\mathcal{D}^{\ket{V}}(\cdot) = 1]\right| \geq \frac{4}{8192} - \mathrm{negl}(n)\geq \frac{1}{4096}
    \end{equation}
    for large $n$.
    This contradicts the defining property of pseudorandom functions $\mathcal{RF}$ (\Cref{thm:PRF} Item 2) under the assumption that $\mathsf{RingLWE}$ is hard.
\end{proof}

Next, we invoke the stronger assumption that \textsf{RingLWE} cannot be solved by any sub-exponential-time quantum algorithms and show that learning unitaries composed of $G=\mathcal{O}(\log n\cdot\mathsf{polyloglog}n)$ gates is computationally hard.

\begin{theorem}[Unitary learning computational complexity lower bound assuming sub-exponential hardness of \textsf{RingLWE}, restatement of lower bound in \Cref{thm:unitary-comp-complex}]
\label{thm:unitary-comp-complexity-subexp}
Let $\lambda=l=\Theta(G)$ with $l\leq n$ be the security parameter.
Let $V$ be an $l$-qubit unitary consisting of $\mathcal{O}(l\polylog(l))=\mathcal{O}(G\polylog(G))$ gates (or a depth $d = \mathcal{O}(\polylog(G))$ circuit) that implements a pseudorandom function in $\mathcal{RF}$.
Such a unitary $V$ exists by \Cref{coro:PRF-size}.
Let $U=V\otimes I$, where the identity $I$ is over the last $(n-l)$ qubits.
Any quantum algorithm for learning a circuit description of the $n$-qubit unitary $U$ to within $\epsilon\leq 1/64$ average-case distance $\davg$ with probability at least $2/3$ from $N = \poly(\lambda)$ queries to $U$ must use $\exp(\Omega(\min\{G, n\}))$ time, if quantum computers cannot solve $\mathsf{RingLWE}$ in sub-exponential time.
\end{theorem}

\begin{proof}
    With polynomial hardness of \textsf{RingLWE} replaced by sub-exponential hardness, \Cref{thm:unitary-comp-complexity} asserts that there are no sub-exponential (in $l$) quantum algorithms that can learn the $l$-qubit unitary $V$ to within average case distance $\epsilon<1/64$ with success probability at least $2/3$.
    That is, any such learning algorithms must use time at least $\exp(\Omega(l))=\exp(\Omega(\min\{G, n\}))$, since $l\leq n$.
    Meanwhile, a polynomial learning algorithm for the $n$-qubit unitary $U=V\otimes I$ can be used to learn the $l$-qubit unitary $V$ in the same runtime by discarding the last $(n-l)$ qubits, because trace distance does not increase under such operation and thus neither does $\davg$.
    This implies the $\exp(\Omega(\min\{G, n\}))$ time lower bound for the $n$-qubit learning algorithm.
\end{proof}

Finally, we briefly show that learning becomes efficient when $G=\mathcal{O}(\log n)$.
The idea is that with $\bigo(\log n)$ gates, there can only be at most $\bigo(\log n)$ qubits affected.
Thus we can focus on these qubits and learning the unitary amounts to manipulating at most $2^{\bigo(\log n)}=\poly(n)$ size matrices, which is efficient.
Specifically, we have the following statement.

\begin{prop}[Learning unitaries with logarithmic circuit complexity efficiently, restatement of upper bound in \Cref{thm:unitary-comp-complex}]
\label{prop:unitary-logn-efficiently}
    Let $\epsilon > 0$. Suppose we are given $N$ queries to an $n$-qubit unitary $U$ consisting of $G=\mathcal{O}(\log n)$ two-qubit gates.
    There exists a learning algorithm that outputs a $\hat{U}$ such that $\davg(U, \hat{U})\leq \epsilon$ with probability at least $2/3$ using $\poly(n, 1/\epsilon)$ queries and time.
\end{prop}

\begin{proof}
    We prove this by a learning algorithm similar to \Cref{prop:state-logn-efficiently} via junta learning based on Choi states (\Cref{app:unitary-choi}) as follows.
    
    Firstly, we prepare the Choi state of $U$ by applying it to a maximally entangled state over $2n$ qubits, execute Algorithm~\ref{algo:junta-choi}, and post-select on the trivial qubits being in the state $\kket{I}$ as in \Cref{app:unitary-choi}.
    This step uses $\poly(n, 1/\epsilon)$ queries and time, and gives us the post-selected Choi state which is non-trivial on only $4G=\bigo(\log n)$ qubits.
    Then we use the Pauli tomography method as in \Cref{prop:state-logn-efficiently} to learn a trace-distance approximation to the $4G$-qubit Choi state $\kket{\hat{V}}$ using $\poly(n, 1/\epsilon)$ queries and time.
    We can enforce this approximation to be a valid Choi state by projecting it to the subspace spanned by $(A\otimes I)\ket{\Phi}$ and normalize the projected state, where $A$ is an arbitrary matrix and $\ket{\Phi}$ is the maximally entangled state.
    This can be done via a projector which is a $2^{4G}=\poly(n)$ dimensional matrix.
    Finally, we calculate the corresponding unitary $\hat{V}$ and set $\hat{U}=\hat{V}\otimes I$.
    Note that this step is efficient as it only involves manipulating matrices of size $2^{4G}=\poly(n)$.
    Since the trace distance between Choi states is equivalent to the average-case distance between the corresponding unitaries, this gives us a $\poly(n, 1/\epsilon)$ learning algorithm for average-case unitary learning.
\end{proof}

\section{Learning physical functions}
\label{app:learning-function}

As stated in the main text, learning classical functions that map variables controlling the input states and evolution to some property of the outputs is an alternative way of learning about Nature.
Learning such functions has long been a central task of statistics and, more recently, classical and quantum machine learning.
However, the physical mechanism that gives rise to these functions has largely been overlooked for the convenience of mathematical abstraction.

In fact, we can formulate the physical mechanism underlying a classical function as an experiment procedure involving a unitary with bounded circuit complexity. 
Specifically, we consider the following general experimental setting.
\begin{enumerate}
    \item Given a set of $\nu$ variables $x\in [0, 1]^\nu$, we prepare a pure state that can depend on $x$ in a fixed way.
    \item We evolve the state using a unitary $U(x; \{U_i\}_{i=1}^G, a)$ that contains at most $G$ two-qubit gates $\{U_i\}_{i=1}^G$, which can be tuned arbitrarily, and any number of fixed unitaries, which can depend on $x$, according to a circuit architecture $a$ in an architecture class $A$.
    \item We measure the output state with a fixed observable $O$ and read out the expectation value as the function output.
\end{enumerate}
We can w.l.o.g.~absorb the state preparation into the unitary. Then the experiment gives rise to the function
\begin{equation}
    f(\cdot; \{U_i\}, a): [0, 1]^\nu\ni x \mapsto f(x; \{U_i\}, a) = \braket{0^n}{U(x; \{U_i\}, a)^\dagger O U(x; \{U_i\}, a)|0^n}.
\end{equation}
We define 
\begin{equation}
    \mathcal{F}_{G, A}^\nu=\{f(\cdot, \{U_i\}, a): a\in A, U_i\in U(2^2), i=1, \ldots, G\}\subseteq \mathbb{R}^{[0, 1]^\nu}
\end{equation} to be the function class given by a class of architectures $A$ for $G$-gate unitaries. 
We call such functions \textit{physical functions}, and $\mathcal{F}_{G, A}^\nu$ the class of $\nu$-variable physical functions with $G$ gates and architectures $A$.

This experiment can also be understood as a quantum machine learning problem, where we want to collect training data $\{x, f(x)\}$ to learn to approximate certain functions in a function class using the ansatz described above.
Then, the tunable gates $\{U_i\}$ can be understood as variational/trainable parameters of our quantum neural network.
We note that the data encoding unitaries may simply use $x$ as the angles for rotation, or it can also be arbitrarily complex (e.g., complex enough to implement a quantum random access memory \cite{giovannetti2008quantum} that prepares the amplitude encoding of the data) as long as it is not trainable. 
This encompass the case where the input data are classical descriptions of the input pure state. 
Also the order of the data encoding unitaries and the trainable unitaries can be arbitrary, thus accommodating data re-uploading strategies \cite{perez2020data, caro2021encodingdependent}.

We will show that to approximate a certain class of functions well, we need a minimal number of samples to learn and a minimal number of gates $G$ (\Cref{thm:approx-functions}).
In particular, we consider the class of $1$-bounded and $1$-Lipschitz functions on $[0, 1]^\nu$, which can (up to equivalence classes) be represented by the unit ball $B^{1, \infty}$ in the Sobolev space $W^{1, \infty}_{[0, 1]^\nu}$.
We establish the following theorem, where the learning criterion is the standard one for learning real functions \cite[Definition 16.1]{anthony1999neural}.

\begin{theorem}[Sample and gate complexity lower bounds on functions given by $G$-gate unitaries to approximate bounded Lipschitz functions, restatement of {\Cref{thm:approx-functions}}]
\label{thm:approx-functions-full}
    Let $\mathcal{F}_{G, A}^\nu\subseteq \mathbb{R}^{[0, 1]^\nu}$ be the function class given by an architecture class $A$ of $G$ two-qubit unitaries. 
    Let $\epsilon\in (0, 1)$ and let $l(|h(x)-y|)$ be a loss function where $l$ is a strictly increasing function with derivative larger than some positive constant on $[1, \infty)$. 
    Suppose for any $1$-bounded and $1$-Lipschitz function $f\in B^{1, \infty}$, there exists an $h\in \mathcal{F}_{G, A}^\nu$ such that $\|f-h\|_\infty<\epsilon$. 
    Then the smallest training data size $N$ such that there exists a learning algorithm $H: ([0, 1]^\nu, [0, 1])^N \to \mathcal{F}_{G, A}^\nu$ that satisfies 
    \begin{equation}
        \mathbb{P}_{S \sim P^N}\left\{ \mathbb{E}_{(X, Y)\sim P} l(|H[S](X)-Y|) - \inf_{f\in\mathcal{F}_{G, A}^\nu}\mathbb{E}_{(X, Y)\sim P} l(|f(X)-Y|) \leq \epsilon \right\} \geq 0.99,
    \end{equation}
    for any probability distribution $P$ over $[0, 1]^{\nu}\times [0, 1]$ must be at least
    \begin{equation}
        N\geq \Omega\left(\frac{1}{\epsilon^\nu}\right).
    \end{equation} 
    Moreover, we need at least 
    \begin{equation}
        G\geq \tilde{\Omega}\left(\frac{1}{\epsilon^{\nu/2}}\right)
    \end{equation}
    two-qubit unitaries if $A$ contains variable circuit structures, or $G\geq \tilde{\Omega}(1/\epsilon^\nu)$ if the circuit structure is fixed.
    The $\tilde{\Omega}$ for variable circuit structures hides logarithmic factors in $\epsilon$ as well as in the number of qubits $n$, while the $\tilde{\Omega}$ for fixed structure only hides logarithmic factors in $\epsilon$.
\end{theorem}

This means that to approximate $1$-bounded and $1$-Lipschitz functions in $\nu$-variables well to $\bigo(1/n^D)$ accuracy, we need at least $\tilde{\Omega}(n^{\nu D/2})$ two-qubit unitaries and $\Omega(n^{\nu D})$ samples to train on. And $\sim 1/\exp(n)$ accuracy can only be achieved with exponential-size quantum circuits and exponentially many samples. This result establishes a limitation on the maximal efficiency of using parameterized quantum circuits to approximate functions, complementary to existing works on universal approximation theorems for parameterized quantum circuits \cite{gonon2023universal, perez2021one, schuld2021effect, manzano2023parametrized}.

The exponential dependence \ins{of the training data size $N$} on the number of variables \edit{$N$}{$\nu$} suggests that if one has an extensively large input vector (whose length scales with $n$), then the number of samples and gates needed to approximate such functions is exponentially large.
Moreover, if the variables are encoded using amplitude encoding (e.g., via QRAM), which accommodates exponentially many variables ($\sim 2^n$), then the gate and sample requirement would grow double exponentially in $1/\epsilon$.
This phenomenon, named curse of dimensionality, was also established in the theory of classical neural networks \cite[Chapter 3]{grohs2022mathematical}. We show that it still exists in quantum machine learning.

This curse can be circumvented by introducing more structure or constraints on the function class. 
For example, if we constrain to Fourier-integrable functions, a $\nu$-independent number of $\bigo(1/\epsilon^2)$ parameters suffice for both classical \cite[Theorem 3.9]{grohs2022mathematical} and quantum \cite{gonon2023universal} machine learning. 
However, the curse of dimensionality shows that many-variable $1$-bounded and $1$-Lipschitz functions are not physical \cite{poland2020no, barthel2018fundamental} because Nature cannot efficiently implement them.

In order to prove \Cref{thm:approx-functions-full}, we proceed in three steps.
Firstly, we show that the complexity of the function class $\mathcal{F}_{G, A}^\nu$ is limited by the number of gates $G$.
Then we prove that to approximate certain functions ($1$-bounded and $1$-Lipschitz functions) well enough, the complexity must not be too small.
Finally, we show that to learn a function class from data, the number of samples we need is lower bounded by the complexity of the function class.

\subsection{Circuit complexity and function complexity}

The complexity of the function class $\mathcal{F}_{G, A}^\nu$, measured by the pseudo-dimension or fat-shattering dimension \cite{mohri2018foundations, wolf2018mathematical}, is limited by the number of trainable gates $G$ and the size of the architecture class $A$. This is because from the linearity of quantum mechanics, the function $f(x; \{U_i\}, a)$ is a polynomial in the matrix elements of the trainable unitaries $\{U_i\}$, and the degree of this polynomial is limited by $G$. Following the idea of \cite{caro2020pseudo}, we formalize this idea into the following lemma.

\begin{lemma}[Functions given by $G$-gate unitaries are bounded degree polynomials]
Let $\mathcal{F}_{G, A}^\nu$ be the function class given by an architecture class $A$ of $G$ two-qubit unitaries. Then there exists a set of functions $P_{G, A}^\nu$ in $32G+\nu$ real variables with size $|P_{G, A}^\nu|=|A|$ such that the following two properties hold.
\begin{enumerate}
    \item $\forall f\in \mathcal{F}_{G, A}^\nu$, there exist a $p\in P_{G, A}^\nu$ and an assignment of the first $32G$ variables such that $p$ under this assignment is the same as $f$ in the last $\nu$ variables;
    \item Each $p\in P_{G, A}^\nu$ depends polynomially on the first $32G$ variables with degree at most $2G$.
\end{enumerate}
\label{lem:func_poly}
\end{lemma}
\begin{proof}
    We begin by noting that for any fixed architecture $a\in A$, the function $f(x, \{U_i\}, a)$ is a function of $32G+\nu$ real variables, where the first $32G=2\cdot 2^2\cdot 2^2 \cdot G$ variables are the real and imaginary parts of the matrix elements of $\{U_i\in U(2^2)\}$, and the last $\nu$ variables are the input data $x\in [0, 1]^\nu$. 
    
    Next, we aim to prove that $f$ is a bounded degree polynomial in the unitary matrix elements. We follow the idea of \cite[Lemma 1]{caro2020pseudo} and analyze the function $f(x, \{U_i\}, a)$ gate by gate. 
    We note the following fact from linear algebra: for any state $\ket{\psi}$ and matrix $U$, the product $U\ket{\psi}$ is a state whose amplitudes are linear combinations of the amplitudes of $\ket{\psi}$ and of matrix elements of $U$. 
    Therefore, by applying $\{U_i\}_{i=1}^G$ and other unitaries that do not depend on $\{U_i\}$ sequentially according to the architecture $a$, we get a state whose amplitude is a polynomial of the matrix elements of $\{U_i\}$ with degree at most $G$. Hence, the output scalar $\braket{0^n}{U(x; \{U_i\}, a)^\dagger O U(x; \{U_i\}, a)|0^n}$ is a polynomial of the matrix elements of $\{U_i\}$ with degree at most $2G$. 
    Fixing those $32G$ variables corresponds to fixing $\{U_i\}$ and thus specifying any particular function in $\mathcal{F}_{G, A}^\nu$ with this architecture $a$. 
    Taking into account the dependence on $x$ and gathering the function for each architecture $a\in A$, we arrive at the desired set of functions $P_{G, A}^\nu$ with $|P_{G, A}^\nu|=|A|$.
\end{proof}

The fact that these functions are of bounded degree in the variables specifying the trainable unitaries implies an upper bound on pseudo-dimension. 
We prove this with a reasoning analogous to \cite{goldberg1995bounding} and \cite[Theorem 2]{caro2020pseudo}.

\begin{prop}[Pseudo-dimension upper bound for functions given by $G$-gate unitaries]
Let $\mathcal{F}_{G, A}^\nu$ be the function class given by an architecture class $A$ of $G$ two-qubit unitaries. 
Then the pseudo-dimension of $\mathcal{F}_{G, A}^\nu$ is at most $128G\log_2(16eG|A|)$.
\label{prop:pdim_gates}
\end{prop}
\begin{proof}
    Let $\{(x_i, y_i)\}_{i=1}^m \subseteq [0, 1]^\nu\times\mathbb{R}$ be a set of data points satisfying that for any $C\subseteq \{1, \ldots, m\}$, there exists $f_C\in\mathcal{F}_{G, A}^\nu$ such that $f(x_i)-y_i\geq 0$ if and only if $i\in C$. That is, $\{(x_i, y_i)\}_{i=1}^m$ is pseudo-shattered by $\mathcal{F}_{G, A}^\nu$.
    From \Cref{lem:func_poly} we know that there exists a set of functions $P$ in $32G+\nu$ real variables with size $|P|=|A|$ such that for every $C$, there is a $p_C\in P$ and an assignment $\Xi_C$ to the first $32G$ variable that satisfies $p_C(\Xi_C, x_i)-y_i\geq 0$ if and only if $i\in C$. 
    This means that the set of functions $\{p(\cdot, x_i)-y_i: i=1, \ldots, m, p\in P\}$ is a set of $m|A|$ polynomials of degree at most $2G$ in $32G$ real variables that has at least $2^m$ different consistent sign assignments\footnote{A consistent sign assignment to a set of polynomials $p_1, \ldots, p_k$ is a vector $b\in\{-1, 0, 1\}^k$ such that there exists a set of input variables $z_1, \ldots, z_N\in \mathbb{R}$ such that $\mathrm{sgn}[p_i(z_1, \ldots, z_N)]=b_i$ for all $1\leq i\leq k$.}. 
    Now we invoke the following technical lemma.
    \begin{lemma}[Bounded degree polynomials have a bounded number of consistent sign assignments, {\cite{warren1968lower, goldberg1995bounding, caro2020pseudo}}]
    Let $P$ be a set of real polynomials in $v$ variables with $|P|\geq v$, each of degree at most $D\geq 1$. Then the number of consistent sign assignments to $P$ is at most $(8De|P|/v)^v$.
    \end{lemma}
    Thus we have
    \begin{equation}
        2^m \leq \left(\frac{8\cdot 2G\cdot em|A|}{32G}\right)^{32G}.
    \end{equation}
    Taking the logarithm yields
    \begin{equation}
        m\leq 32G(\log_2 (16eG|A|) + \log_2(m/(32G))).
    \end{equation}
    Let's first assume $m\geq 32G$. 
    If $\log_2(16eG|A|)\geq \log_2(m/(32G))$, then we have $m\leq 64G\log_2(16eG|A|)$. 
    Otherwise, $\log_2(16eG|A|)< \log_2(m/(32G))$ and we have $m\leq 64G\log_2(m/(32G))$, which translates into $\frac{\log_2(m/(32G))}{m/(32G)}\geq \frac{1}{2}$. 
    Thus $m/(32G)\leq 4$ and $m\leq 128G$. 
    In both cases, we have $m\leq 128G\log_2(16eG|A|)$. 
    If $m<32G$, this is also true. 
    Therefore, we have pseudo-dimension (by definition in \Cref{def:pdim}) at most $128G\log_2(16eG|A|)$.
\end{proof}

A special case is for fixed circuit architecture $|A|=1$, where we have pseudo-dimension at most $128G\log_2(16eG)$. 
On the other hand, if we allow variable structure of the trainable unitaries, then $|A|\leq \binom{n}{2}^G\leq n^{2G}$, and we have pseudo-dimension at most $128G\log_2(16eG) + 256G^2\log_2(16eGn)$.

\subsection{Function complexity and approximation power}

Now that we know the pseudo-dimension of such function class is upper bounded via the number of gates $G$, we can derive the minimal number of gates needed to obtain certain function approximation power. 
Consider the class of $1$-bounded and $1$-Lipschitz functions on $[0, 1]^\nu$, which can be represented by the unit ball $B^{1, \infty}$ in the Sobolev space $W^{1, \infty}_{[0, 1]^\nu}$. In order to approximate these functions well, the pseudo-dimension (and also the fat-shattering dimension) of our function class cannot be too small.
\begin{lemma}[Pseudo/fat-shattering dimension and approximation power, variant of {\cite[Theorem 2.10]{wolf2018mathematical}} and {\cite[Theorem 4]{yarotsky2017error}}]
Let $\epsilon>0$ and $\mathcal{F}\subseteq \mathbb{R}^{[0, 1]^\nu}$ be a class of functions such that for any $f\in B^{1, \infty}$, there is an $h\in\mathcal{F}$ such that $\|f-h\|_\infty<\epsilon$. 
Then the pseudo-dimension of $\mathcal{F}$ must be at least $1/(4\epsilon)^\nu$. 
The $\epsilon$-fat-shattering dimension of $\mathcal{F}$ must be at least $1/(8\epsilon)^\nu$.
\label{lem:pdim_approx}
\end{lemma}
\begin{proof}
    Let $m\in \mathbb{N}$ to be chosen later. Let $x_1, \ldots, x_M \in [0, 1]^d$ be $M=(m+1)^\nu$ points on a cubic lattice such that $\|x_i-x_j\|\geq 1/m$ for all $i\neq j$. 
    Let $y\in \mathbb{R}^M$, and we will now construct a smooth function that takes the $y$ values at these lattice points. 
    Specifically, we define
    \begin{equation}
        f(x) = \sum_{i=1}^M y_i \phi(m(x-x_i)),
    \end{equation}
    where $\phi(z)=\prod_{j=1}^\nu \varphi(z_j)$ and $\varphi$ is a smoothed version of the triangular function that takes value $0$ at $|z|\geq 1/2$ and value $1$ at $z=0$ and $|\partial_j \phi(z)|\leq C$ for any $C>2$. 
    In this way, we have $f(x_i) = y_i$ for all $1\leq i\leq M$.

    Next, for any $\alpha\in\{0, 1\}^M$, set $y_i = \alpha_i/(Cm)$. 
    This means that $|y_i|\leq 1/(Cm)$ and thus $f\in B^{1, \infty}$. 
    Then by assumption there must be an $h\in\mathcal{F}$ such that $\|f-h\|_\infty <\epsilon$. In particular, we have $|f(x_i)-h(x_i)|=|y_i-h(x_i)| <\epsilon$ for all $i$. 

    Now, for the pseudo-dimension, we can choose $m$ large enough (say, $m=\floor{1/(C^2\epsilon)}$) such that $\epsilon<1/(2Cm)$. Then
    \begin{equation}
        h(x_i) \geq \frac{1}{2Cm} \iff \alpha_i=1, y_i=\frac{1}{Cm}.
    \end{equation}
    Therefore, by definion in \Cref{def:pdim}, $\{x_1, \ldots, x_M\}$ is pseudo-shattered by $\mathcal{F}$, and thus the pseudo-dimension of $\mathcal{F}$ is at least $M = (m+1)^\nu \geq 1/(C^2\epsilon)^\nu$. Taking the limit $C\to 2$ yields the desired result.

    For fat-shattering dimension, we can choose $m$ large enough (say, $m=\floor{1/(C^3\epsilon)}$) such that $\epsilon<1/(4Cm)$. Then
    \begin{equation}
        \alpha_i=1 \implies h(x_i)\geq \frac{1}{Cm}-\epsilon \geq \frac{1}{2Cm} + \epsilon,
    \end{equation}
    and 
    \begin{equation}
        \alpha_i=0 \implies h(x_i)\leq \epsilon \leq \frac{1}{2Cm} - \epsilon.
    \end{equation}
    Therefore, by definion in \Cref{def:fat}, $\{x_1, \ldots, x_M\}$ is $\epsilon$-fat-shattered by $\mathcal{F}$, and thus the $\epsilon$-fat-shattering dimension of $\mathcal{F}$ is at least $M = (m+1)^\nu \geq 1/(C^3\epsilon)^\nu$. Taking the limit $C\to 2$ yields the desired result.
\end{proof}

\subsection{Function complexity and sample complexity}

Now we aim to show that in order to learn a function class, the number of samples we need is lower bounded by its complexity.
In particular, we achieve this through the fat-shattering dimension.

\begin{prop}[Sample complexity lower bound for real-valued functions by fat-shattering dimension, variant of {\cite[Theorem 19.5]{anthony1999neural}}]
    Let $\mathcal{F}\subseteq[0, 1]^\mathcal{X}$ with loss function $l_h(x, y) = l(|h(x)-y|)$. 
    Suppose $l$ is an increasing (almost everywhere) differentiable function, i.e., $C=\inf_{ t\geq 1} l'(t)>0$. 
    For $0<\epsilon<1, 0<\delta\leq 0.01$, the smallest training data size $N$ such that there exists a learning algorithm $H: (\mathcal{X}, [0, 1])^N \to \mathcal{F}$ that satisfies 
    \begin{equation}
        \mathbb{P}_{S \sim P^N}\left\{ \mathbb{E}_{(X, Y)\sim P} l(|H[S](X)-Y|) - \inf_{f\in\mathcal{F}}\mathbb{E}_{(X, Y)\sim P} l(|f(X)-Y|) \leq \epsilon \right\} \geq 1-\delta,
    \end{equation}
    for any probability distribution $P$ over $\mathcal{X}\times [0, 1]$ must be at least
    \begin{equation}
        N\geq C\frac{\mathrm{fat}(\mathcal{F}, \epsilon/\alpha)-1}{32\alpha}, \quad \forall \alpha \in (0, 1/4).
    \end{equation}
    Note that this contains $L_p$ loss functions as a special case, where $l_h(x, y)=|h(x)-y|^p$, and $l'(t)=pt^{p-1}\geq p=C$.
\label{prop:fat}
\end{prop}
\begin{proof}
    Similarly to the proof of Theorem 19.5 in \cite{anthony1999neural}, the idea is to reduce the problem to a discrete classification problem. 
    Consider the class $H_d$ of all functions mapping from a finite set $\{x_1, \ldots, x_d\}\subset \mathcal{X}$ to $\{0,1\}$. 
    It's known that any learning algorithm for $H_d$ has sample complexity at least $(d-1)/(32\epsilon)$ for small $\epsilon, \delta$ (\cite[Theorem 5.3]{anthony1999neural}). 
    Here we show that, for any fixed $\alpha$ between $0$ and $1/4$, any learning algorithm for $\mathcal{F}$ to accuracy $\epsilon$ can be used to construct a learning algorithm for $H_d$ to accuracy $\alpha/C$, where $d=\mathrm{fat}(\mathcal{F}, \epsilon/\alpha)$. 
    Then the proposition follows.

    To see this, suppose $\{x_1, \ldots, x_d\}$ is $\epsilon/\alpha$-shattered by $\mathcal{F}$, witnessed by $r_1, \ldots, r_d$. 
    Suppose $L$ is a learning algorithm for $\mathcal{F}$, then we can construct a learning algorithm for $H_d$ as follows. 
    For each labeled example $(x_i, y_i)$, assuming $y_i$ is deterministic given $x_i$, the algorithm passes to $L$ the labeled example $(x_i, \tilde{y}_i)$, where $\tilde{y}_i=2$ if $y_i=1$ and $\tilde{y}_i=-1$ if $y_i=0$. 
    Let $P$ be the original distribution on $\mathcal{X}\times \{0, 1\}$, and $\tilde{P}$ the induced distribution on $\mathcal{X}\times \{-1, 2\}$. 
    Then suppose $L$ produces a function $f: \mathcal{X}\to [0, 1]$, the learning algorithm for $H_d$ then outputs $h: \mathcal{X}\to \{0, 1\}$, where $h(x_i)=1$ if and only if $f(x_i)>r_i$. 
    Thus we only need to prove if $\mathbb{E}_{\tilde{P}}l_f-\inf_{g\in\mathcal{F}}\mathbb{E}_{\tilde{P}}l_g<\epsilon$, then $\mathbb{E}_P 1(h(x)\neq y)\leq \alpha/C.$

    To show this, we claim that
    \begin{align}
        \inf_{g\in\mathcal{F}}\mathbb{E}_{\tilde{P}}l_g = \inf_{g\in\mathcal{F}}\mathbb{E}_{\tilde{P}}l(|g(x)-\tilde{y}|) \leq \mathbb{E}_{\tilde{P}} \min \{l(|\hat{y}-\tilde{y}|), \hat{y}\in\{r(x)\pm \epsilon/\alpha\}\},
    \end{align}
    where $r(x_i)=r_i$. 
    This is because that $\tilde{P}$ is concentrated on the shattered set. 
    Then for any assignment $\{\hat{y}_i\in \{r_i\pm\epsilon/\alpha\}, i=1, \ldots, d\}$, there exists a $g\in\mathcal{F}$ s.t. $g(x_i)\geq \hat{y}_i$ if $\hat{y}_i=r_i+\epsilon/\alpha$ and $g(x_i)\leq \hat{y}_i$ if $\hat{y}_i=r_i-\epsilon/\alpha$. 
    In particular, we consider the assignment of $\hat{y}_i$ s.t. $l(|\hat{y}_i-\tilde{y}_i|)$ is minimized. 
    Then there exists a function $g^*$ staistifying the following property. 
    If $\tilde{y}_i=-1$, then the minimizer is $\hat{y}_i=r_i-\epsilon/\alpha$, and we have $l(|g^*(x_i)-\tilde{y}_i|)\leq l(|\hat{y}-\tilde{y}_i|)$ since $\tilde{y}_i < g^*(x_i)\leq \hat{y}_i$. 
    Similarly, if $\tilde{y}_i=2$, then the minimizer is $\hat{y}_i=r_i+\epsilon/\alpha$, and we still have $l(|g^*(x_i)-\tilde{y}_i|)\leq l(|\hat{y}-\tilde{y}_i|)$ since $\hat{y}_i \leq g^*(x_i) \leq \tilde{y}_i$. 
    Therefore, since $\tilde{y}_i$ and $y_i$ is deterministic given $x_i$, we have found a single $g^*$ s.t. $\mathbb{E}_{\tilde{P}}l(|g^*(x)-\tilde{y}|) \leq \mathbb{E}_{\tilde{P}} \min \{l(|\hat{y}-\tilde{y}|), \hat{y}\in\{r(x)\pm \epsilon/\alpha\}\}$. 
    Hence, the infimum over $g\in\mathcal{F}$ $\inf_{g\in\mathcal{F}}\mathbb{E}_{\tilde{P}}l(|g(x)-\tilde{y}|) \leq \mathbb{E}_{\tilde{P}}l(|g^*(x)-\tilde{y}|) \leq \mathbb{E}_{\tilde{P}} \min \{l(|\hat{y}-\tilde{y}|), \hat{y}\in\{r(x)\pm \epsilon/\alpha\}\}$. 
    Therefore,
    \begin{align}
        \mathbb{E}_{\tilde{P}}l_f-\inf_{g\in\mathcal{F}}\mathbb{E}_{\tilde{P}}l_g \geq \mathbb{E}[l(|f(x)-\tilde{y}|)-\min\{l(|\hat{y}-\tilde{y}|), \hat{y}\in\{r(x)\pm \epsilon/\alpha\}\}].
    \end{align}
    Consider the quantity inside the expectation, for $x=x_i$ with $y=0$, $\tilde{y}=-1$, let $a=f(x_i)+1, b=r_i-\epsilon/\alpha+1$. 
    Then by Lagrange's mean value theorem, there exists a $c$ between $a$ and $b$, such that this quantity can be written as 
    \begin{align}
        l(a)-l(b)=l'(c)(a-b)=l'(c)(f(x_i)-r_i+\epsilon/\alpha).
    \end{align}
    If $l(|f(x)-\tilde{y}|)-\min\{l(|\hat{y}-\tilde{y}|), \hat{y}\in\{r(x)\pm \epsilon/\alpha\}\}<C\epsilon/\alpha$, then
    \begin{align}
        f(x_i)-r_i
        &< \frac{\epsilon}{\alpha}\frac{C-l'(c)}{l'(c)}<0,
    \end{align}
    and we have $f(x_i)<r_i$ and $h(x_i)=0=y_i$.
    Similar arguments apply for $y=1$. 
    Thus,
    \begin{align}
        \mathbb{E}_P 1(h(x)\neq y) &\leq \tilde{P}[|f(x)-\tilde{y}|^p-\min\{|\hat{y}-\tilde{y}|^p, \hat{y}\in\{r(x)\pm \epsilon/\alpha\}\}\geq C\epsilon/\alpha] \\
        &\leq \frac{\alpha}{C\epsilon}\mathbb{E}_{\tilde{P}}[|f(x)-\tilde{y}|^p-\min\{|\hat{y}-\tilde{y}|^p, \hat{y}\in\{r(x)\pm \epsilon/\alpha\}\}] \\
        &\leq \frac{\alpha}{C\epsilon} (\mathbb{E}_{\tilde{P}}l_f-\inf_{g\in\mathcal{F}}\mathbb{E}_{\tilde{P}}l_g) \leq \frac{\alpha}{C}.
    \end{align}
    This completes the proof of \Cref{prop:fat}.
\end{proof}

With \Cref{prop:pdim_gates}, \Cref{lem:pdim_approx} and \Cref{prop:fat}, we can finally proceed to prove \Cref{thm:approx-functions-full}.

\begin{proof}[Proof of \Cref{thm:approx-functions-full}]
    To show the gate number lower bound, note that from \Cref{prop:pdim_gates}, $\mathrm{Pdim}(\mathcal{F}_{G, A}^\nu)$ is upper bounded by $128G\log_2(16eG)+256G^2\log_2(16eGn)$ for variable circuit structures and by $128G\log_2(16eG)$ for fixed circuit structure.
    Meanwhile, from \Cref{lem:pdim_approx}, we know that to approximate any $\nu$-variable $1$-bounded and $1$-Lipschitz functions to $\epsilon$ error in $\norm{\cdot}_\infty$, we must have $\mathrm{Pdim}(\mathcal{F}_{G, A}^\nu)\geq 1/(4\epsilon)^\nu$ and $\mathrm{fat}(\mathcal{F}_{G, A}^\nu, \epsilon)\geq 1/(8\epsilon)^\nu$.
    Therefore, for variable circuit structures, we have
    \begin{equation}
        1/(4\epsilon)^\nu \leq \mathrm{Pdim}(\mathcal{F}_{G, A}^\nu) \leq 128G\log_2(16eG)+256G^2\log_2(16eGn),
    \end{equation}
    and thus $G\geq \tilde{\Omega}(1/(\epsilon)^{\nu/2})$. Similarly, for fixed circuit structure, we have $G\geq \tilde{\Omega}(1/\epsilon^{\nu})$.

    To show the sample complexity lower bound, note that from \Cref{prop:fat}, we have the sample complexity $N\geq C\frac{\mathrm{fat}(\mathcal{F}_{G, A}^\nu, \epsilon/\alpha)-1}{32\alpha}$. Setting $\alpha=1/8$ and using the fat-shattering bound from \Cref{lem:pdim_approx}, we arrive at $N\geq \Omega(1/\epsilon^\nu)$. 
\end{proof}

\ins{
\section{Details of the numerical experiments}
\label{app:numerics}
}

\ins{
In this section, we provide the implementation details of the numerical experiments presented in \Cref{sec:numerics}.
All simulations are conducted using JAX \cite{jax2018github} and TensorCircuit \cite{zhang2023tensorcircuit} with the matrix product state backend.
}

\ins{
We consider a large system size $n = 10000$ to illustrate the sample complexity's independence of $n$.
For simplicity, we consider the setting in which the $G$ gates used to generate the unknown target states are sampled from a discrete gate set of size two.
Here, we sample two Haar-random 2-qubit gates to form this gate set.
We also assume that the $n=10000$ qubits form a one-dimensional line, and each gate only acts on neighboring qubits.
The gate configurations are randomly sampled over the first $4$ qubits (Case (a)) or over all qubits (Case (b)).
We further assume that the gate positions are known to the learning algorithm to accelerate the simulation.
Of course, this assumption can be removed by enumerating over all $\binom{n'}{2}^G$ possible configurations (with $n'=4$ or $\leq 2G$ the number of qubits acted upon non-trivially by the $G$ gates) or by performing Algorithm~\ref{algo:nonzero-qubits}.
However, this will introduce an additional overhead to the implementation, which we want to avoid for the sake of these numerics.
}

\ins{
For each $G$ and sampled target state, we perform the learning algorithm detailed in \Cref{app:states-upper} and calculate the fidelity $F$ between the output state and the target state.
To accelerate the learning algorithm for different sample sizes, we first sample the Clifford classical shadows with the maximal sample size, and then sub-sample the results for smaller sample sizes.
In Case (b), where the effective system size $n'$ is larger (up to $2G=12$), we adopt the shallow shadow modification of Clifford classical shadows \cite{bertoni2024shallow,schuster2024random}.
Specifically, we replace the global Clifford rotation with a brickwork Clifford circuit of depth $10$.
We also replace Clifford gates with Haar random gates to reduce statistical fluctuations.
}

\ins{
The procedure described above constitutes a single sweep over the entire $G-N$ plane.
We repeat this procedure independently $100000$ times for Case (a) and $25000$ times for Case (b).
We record the resulting fidelity and calculate the average and median values for each $G$ and $N$.
The results are presented in \Cref{fig:numerics}.
}

\ins{
We note that for each $G$, when the sample size is above the sample complexity, the reconstruction fidelity of most trials are exactly one up to machine precision.
This is because we are using a discrete gate set, so the target states can be found unambiguously.
Therefore, the sample complexity lines for $F_\mathrm{med}=0.999, 0.9999, 0.99999$ coincide, and we only plot the one for $F_\mathrm{med}=0.999$.
Meanwhile, the line for averaged fidelity $F$ changes with the value of $F$, because the average fidelity $F$ takes into account those failing cases where fidelity can be close to zero.
}

\vspace{2em}
\section{Preliminary results on learning brickwork circuits}
\label{app:brickwork}

As stated in the outlook section, an interesting circuit structure is the brickwork circuit, which is generated by repeatedly applying the following two layers of gates (suppose $n$ is even): (1) $U_{1, 2}\otimes U_{3, 4}\otimes\cdots\otimes U_{n-1, n}$ and (2) $U_{2, 3}\otimes U_{4, 5}\otimes\cdots\otimes U_{n-2, n-1}$, where $U_{i, j}$ denotes a 2-qubit unitary acting on the $i$th and $j$th qubit.
Here we utilize the tools from unitary $t$-designs \cite{brandao2016local} to prove that the metric entropy of $G$-gate brickwork circuits is lower bounded by $\Omega(tn)$, if they can implement (approximate) unitary $t$-designs.
Specifically, we have the following result.

\begin{prop}
[Metric entropy lower bound of brickwork circuits]
Let $U^{n, \mathrm{brick}}_G \subseteq U(2^n)$ be the set of $n$-qubit unitaries that can be implemented with $G$-gate brickwork circuits. 
Suppose that the uniform distribution over $U^{n, \mathrm{brick}}_G$ forms an $\epsilon$-approximate $t$-design of $U(2^n)$ for some $\epsilon\in (0,1/2)$.
Then we have
\begin{equation}
    \log\mathcal{M}(U^{n, \mathrm{brick}}_G, \davg, \epsilon) \geq \Omega(tn).
\end{equation}
\end{prop}
\begin{proof}
    Suppose $U_G^{n, \mathrm{brick}}$ with the uniform distribution forms an $\epsilon$-approximate $t$-design $\mathcal{E}$ of $U(2^n)$.
    We begin by recalling a moment bound for approximate unitary designs.
    \begin{lemma}[Moment bound of approximate unitary designs, {\cite[proof of Lemma 1]{brandao2021models}}]
    Suppose $\mathcal{E}$ is an $\epsilon$-approximate unitary $t$-design of $U(d)$. Then for any unitary $V\in U(d)$, we have
    \begin{equation}
    \mathbb{E}_{U\sim\mathcal{E}}\left[|\tr(U^\dagger V)|^{2t}\right] \leq (1+\epsilon)t!\, .
    \end{equation}
    \end{lemma}
    Consequently, by Markov's inequality, we have the following lemma saying that a random element of a design is far apart from a fixed unitary with high probability.
    \begin{lemma}
    (Design elements are far away from any fixed unitary).
    Suppose $\mathcal{E}$ is an $\epsilon$-approximate unitary $t$-design of $U(d)$. Then for any unitary $V\in U(d)$, we have
    \begin{equation}
        \mathbb{P}_{U\sim\mathcal{E}}\left[ \|U-V\|_F^2 \leq 2d(1-\Delta) \right] \leq \mathbb{P}_{U\sim\mathcal{E}}\left[ |\tr(U^\dagger V)| \geq d\Delta \right] \leq \frac{1+\epsilon}{\Delta^{2t}} \frac{t!}{d^{2t}}.
    \end{equation}
    \label{lem:design_far_away}
    \end{lemma}
    \begin{proof}
        To prove this, we use the above moment bound and Markov's inequality:
        \begin{equation}
            \mathbb{P}_{U\sim\mathcal{E}}\left[ |\tr(U^\dagger V)| \geq d\Delta \right] = \mathbb{P}_{U\sim\mathcal{E}}\left[ |\tr(U^\dagger V)|^{2t} \geq d^{2t}\Delta^{2t} \right] \leq \frac{\mathbb{E}_{U\sim\mathcal{E}}\left[|\tr(U^\dagger V)|^{2t}\right]}{d^{2t}\Delta^{2t}} \leq \frac{1+\epsilon}{\Delta^{2t}} \frac{t!}{d^{2t}}.
        \end{equation}
        Furthermore, since $\|U-V\|^2_F=2d-2\mathrm{Re}[\tr(U^\dagger V)]\leq 2d(1-\Delta)$ implies $|\tr(U^\dagger V)| \geq \mathrm{Re}[\tr(U^\dagger V)] \geq d\Delta$, \Cref{lem:design_far_away} follows.
    \end{proof}

    Now, we apply a probabilistic argument by randomly choosing $M$ i.i.d. unitaries $U_1, \ldots, U_M$ from $\mathcal{E}$. 
    The probability that any two of them are far away from each other is given by
    \begin{align}
        &\mathbb{P}_{U_1, \ldots, U_M \sim \mathcal{E}}[\forall 1\leq i\neq j\leq M, \|U_i-U_j\|_F^2 \geq 2d(1-\Delta)] \\
        &=1-\mathbb{P}_{U_1, \ldots, U_M \sim \mathcal{E}}[\exists 1\leq i\neq j\leq M, \|U_i-U_j\|_F^2 \leq 2d(1-\Delta)] \\
        &\geq 1 - \sum_{1\leq i\neq j\leq M}\mathbb{P}_{U_1, \ldots, U_M \sim \mathcal{E}}[\|U_i-U_j\|_F^2 \leq 2d(1-\Delta)] \\
        &\geq 1-\frac{M(M-1)}{2} \frac{1+\epsilon}{\Delta^{2t}} \frac{t!}{d^{2t}}\\
        &\geq 1-M^2 \frac{1+\epsilon}{\Delta^{2t}} \frac{t!}{d^{2t}}\, ,
    \end{align}
    where we have used the union bound in the first inequality and \Cref{lem:design_far_away} in the last.
    Therefore, as long as we take $M < \sqrt{\floor{\frac{\Delta^{2t}}{1+\epsilon} \frac{d^{2t}}{t!}}}$, we have 
    \begin{equation}
        \mathbb{P}_{U_1, \ldots, U_M \sim \mathcal{E}}[\forall 1\leq i\neq j\leq M, \|U_i-U_j\|_F^2 \geq 2d(1-\Delta)]>0.
    \end{equation}
    Hence there must be at least one instance $V_1, \ldots, V_M \in \mathcal{E}$ such that $\|V_i-V_j\|_F^2 \geq 2d(1-\Delta)$ for any pair $V_i, V_j$. 
    These unitaries form a $\sqrt{2d(1-\Delta)}$-packing net of $U_G^{n, \mathrm{brick}}$ with respect to  $\|\cdot\|_F$. 
    Thus we have
    \begin{equation}
        \log \mathcal{M}(U_G^{n, \mathrm{brick}}, \|\cdot\|_F, \sqrt{2d(1-\Delta)}) 
        \geq \Omega\left( \frac{1}{2}\log\floor{\frac{\Delta^{2t}}{1+\epsilon} \frac{d^{2t}}{t!}}\right).
    \end{equation}
    If we set $\sqrt{2d(1-\Delta)}=\sqrt{d}\epsilon$ (i.e., $\Delta=1-\epsilon^2/2$), we arrive at
    \begin{equation}
        \log \mathcal{M}(U^{n, \mathrm{brick}}_G, d_F, \epsilon) \geq \Omega(tn).
    \end{equation}
    From the fact that quotienting out a global phase only changes the metric entropy by an additive $\Omega(\log(1/\epsilon))$ terms (\Cref{lem:quotient-packing}) and the equivalence of $d_F'$ and $\davg$ (\Cref{lem:dist-df'} Item 1), we arrive at the desired result.
\end{proof}

\end{document}